\documentclass[11pt]{article}
\usepackage[utf8]{inputenc}
\usepackage[T1]{fontenc}
\usepackage{lmodern}
\usepackage[margin=1in]{geometry}
\usepackage{amsmath,amssymb,amsfonts,amsthm}
\usepackage{graphicx}
\usepackage{booktabs}
\usepackage{xcolor}
\usepackage{tikz}
\usetikzlibrary{bayesnet,positioning,calc}
\usepackage[round,authoryear]{natbib}
\usepackage{hyperref}
\usepackage{url}
\usepackage[section]{placeins}
\graphicspath{{./}{plots/}{Fig/}}

\newcommand{\Var}{\operatorname{Var}}
\newcommand{\Cov}{\operatorname{Cov}}
\newcommand{\E}{E}
\newcommand{\eb}{\mathrm{EB}}

\newcommand{\rev}[1]{#1}
\newcommand{\bl}[1]{\textcolor{blue}{#1}}

\theoremstyle{plain}
\newtheorem{theorem}{Theorem}
\newtheorem{proposition}[theorem]{Proposition}
\newtheorem{lemma}{Lemma}
\theoremstyle{definition}
\newtheorem{example}{Example}

\theoremstyle{remark}

\title{Shrinkage through multiple identifiability}
\author{%
Carlos Garc\'ia Meixide\textsuperscript{1,2}\qquad David R\'ios Insua\textsuperscript{1}\\[1.2ex]
{\small \textsuperscript{1}Instituto de Ciencias Matem\'aticas, CSIC, Madrid, Spain}\\
{\small \textsuperscript{2}Departamento de Matem\'aticas, Universidad Aut\'onoma de Madrid, Spain}%
}
\date{ }

\newif\ifmerged\mergedtrue
\begin{document}
\maketitle
\begin{abstract}
We propose an empirical Bayes framework for combining estimators obtained from multiple identification functionals associated with the same estimand. We adaptively pool a collection of asymptotically linear estimators, each of which may target a different parameter because of violations of their identification assumptions. Although all estimators are computed from the same sample and are therefore dependent, we show that a working independence construction preserves consistency of the posterior mean under centered heterogeneity. Inference is driven by a latent heterogeneity parameter governing the dispersion of the different identification images. When this parameter is zero, the functionals share a common estimand and we construct frequentist confidence intervals using either sandwich variance estimation or subsampling. When it is positive, the functionals are interpreted as exchangeable draws from a latent population of causal effects, and we construct asymptotically valid Bayesian prediction intervals for the latent target of a new identification functional. These inferential procedures answer different questions, rely on distinct assumptions, and are therefore complementary rather than competing. We illustrate the framework by augmenting evidence from randomized controlled trials with observational studies.
\end{abstract}


	\section{Introduction}

Randomized controlled trials (RCTs) are widely regarded as the gold standard for internal validity. However they are costly, slow, and operationally constrained. Recruitment is one bottleneck: if the experimental treatment can be obtained outside the trial, patients may be reluctant to enrol for fear of being allocated to the control arm, which inflates variance and erodes power even when identification is clear. When the treatment is available only within the trial, this reluctance disappears. Real-world data (RWD) can help, but mainly on the control side: records of standard care are often more plentiful than trial controls, which motivates augmenting RCTs with external controls \citep{viele2014use,schmidli2014robust,galwey2017supplementation,schmidli2020beyond,kojima2023dynamic,collignon2020clustered}. Such data are rarely available for a new intervention, which has little track record in routine practice. External and trial populations differing in region, measurement, or care pathway creates bias, so precision gains must be weighed against robustness to cross study discrepancies \citep{burger2021use,busgang2022selecting,collignon2020clustered}. Handling such discrepancies calls for causal inference methodology that accounts for hidden confounding, where a given causal parameter defined in the counterfactual full data world is typically linked to the observable data-generating distribution through more than one functional. 

{Examples of multiple identifying functionals for the same causal parameter arise routinely: through back-door versus front-door formulas \citep{henckel2019graphical,gorbach2023contrasting,mohammadtaheri2023optimal}; through several pre-treatment cutoffs in regression discontinuity \citep{thistlethwaite1960regression,hahn2001identification,lee2010regression}; through several candidate control groups in difference-in-differences; and through instrumental variables indexing several data environments \citep{peters2016invariant,meixide2025domainadaptation,meixide2025,meixide2025causalinferenceimpliedinterventions}. Under regular asymptotic linearity each functional admits an efficient influence curve, and different functionals generally attain different semiparametric efficiency bounds, so no single choice is uniformly optimal \citep{bickel1993efficient,robins1994estimation,hahn1998role}. Recent empirical Bayes proposals borrow across asymptotically linear estimators from related populations \citep{law2023}; closer to us, \citet{wu2026illusion} combine one experiment with many \emph{independent} observational studies whose bias distribution is identified through auxiliary zero-effect calibration studies. In contrast, our estimators arise from multiple identifying \emph{functionals} on a \emph{single} sample, so it is the dispersion among them---not an external calibration channel---that identifies the heterogeneity hyperparameter. }

{To wit, the starting point of this paper is a collection $\{\widehat\psi_j\}_{j=1}^{J}$ of asymptotically linear estimators corresponding to the empirical versions of $J$ identification functionals   $\{\psi_j\}_{j=1}^{J}$ evaluated on the probability distribution $P$ of the data, targeting a common causal parameter $\psi^*$. We model their latent targets through a Gaussian hierarchical working model, which is an algebraic device for deriving the estimators rather than an inferential assumption (Section~\ref{sec:identifiability}),
\begin{equation}\label{eq:working-model}
\psi_j \;=\; \psi^* + \varepsilon_j,\qquad \varepsilon_j\sim N(0,\tau^2),\qquad \widehat\psi_j\mid\psi_j\sim N(\psi_j,v_j),\qquad j=1,\ldots,J,
\end{equation} where $\tau^2\ge 0$ is a hyperparameter codifying heterogeneity. When every functional is exactly valid, $\psi_j(P)=\psi^*$ for all $j$ and $\tau^2=0$. When identification assumptions fail, the deviations $\varepsilon_j=\psi_j(P)-\psi^*$ induce $\tau^2>0$. EB estimation of $(\psi^*,\tau^2)$ yields a precision weighted combined estimator $\widehat\psi_{\mathrm{EB}}$, and the hyperparameter $\tau^2$ acquires a direct identification meaning by quantifying the dispersion of the $J$ identifying targets around $\psi^*$. The $J$ estimators are computed on the same sample, so they are strongly dependent. We handle this by adopting a {working independence device}, showing that the resulting point estimator remains consistent nevertheless, and correcting inference afterwards.}

{A particularly transparent instance arises when the functionals index an exhaustive partition of the covariate space into mutually exclusive strata, so that $\psi_j=\E[Y(1)-Y(0)\mid W\in s_j]$ is the conditional average treatment effect (CATE) of cell $s_j$---say $\{\text{woman, large tumor}\}$, $\{\text{man, large tumor}\}$, and so on. The law of total expectation makes the global average treatment effect the population-weighted average of these cell effects, $\psi^*=\E[\mathrm{CATE}(W)]$, so the deviations $\varepsilon_j=\psi_j-\psi^*$ are automatically centered: $\E[\varepsilon_j]=0$. This is precisely the centered-heterogeneity regime of Theorem~\ref{thm:route1} in this paper, and the EB estimator then delivers the partial pooling practitioners want \citep{hahn2018regularization}: it reads the global dispersion $\tau^2$ of the conditional effects and shrinks accordingly, borrowing strength toward the consensus when the cells agree ($\tau^2\approx0$) and protecting each cell's signal when they do not, so that no single populous stratum dominates. This recasts the classical \emph{reference-class problem}---how narrowly to condition before a subgroup becomes too sparse to estimate---as a matter of pooling rather than choice: cells too thin to estimate on their own, such as age $\times$ tumor size, jointly pin down $\psi^*$, while the conformal interval of Section~\ref{sec:bayesian} predicts $\psi_{J+1}$ for a patient whose stratum was never observed.}

{The paper's contributions are fourfold. \emph{First}, we establish consistency of the estimator $\widehat\psi_{\mathrm{EB}}$ for $\psi^*$ in two non-nested regimes: exact identifiability with fixed $J$, and a centered-bias regime in which the functionals are biased but the biases are mean zero across $j$ and $J$ grows with $n$, where consistency follows from a law of large numbers over the functional index. \emph{Second}, we show that the {working independence device} delivers consistent inference through an observation-level sandwich variance. \emph{Third}, we connect the $\tau^2>0$ regime to a prediction problem: by de Finetti's theorem the latent targets are conditionally i.i.d., making $\psi_{J+1}$ a genuine prediction target, and we construct split conformal intervals \citep{lei2018distribution} whose asymptotic validity rests on the correlated sampling noise being $O_p(n^{-1/2})$ while the latent heterogeneity is $O_p(1)$. \emph{Fourth}, we relate the framework to the minimax RCT/observational fusion literature \citep{LiGilbertDuanLuedtke2025,minimaxRCTobs}, recovering the low-bias pooling recommendation and extending it beyond strict exchangeability. The work treats the causal parameter itself as the object of the hierarchical model rather than as a byproduct of an outcome-level prior; see \citet{hahn2018regularization} for a related Bayesian treatment of average treatment effects under confounding.}

{The phase transition at $\tau^2=0$ organizes inference, as Table~\ref{tab:regimes} summarizes. At $\tau^2=0$ every functional targets the same $\psi^*$ and the relevant object is a confidence interval for it. For $\tau^2>0$ the latent target $\psi_{J+1}$ of a new functional becomes genuinely random and a split conformal prediction interval is available; the centered-bias regime reports both at once, a frequentist guarantee for the fixed $\psi^*$ and a Bayesian one for the random $\psi_{J+1}$ (Section~\ref{sec:bayesian}). Appendix~\ref{sec:catalogue} catalogues five canonical design families in which such functional multiplicity arises---regression discontinuity with multiple cutoffs, staggered difference-in-differences, difference-in-differences with multiple control groups, instrumental variables indexing $q$ environments (where $J= 2^q - q - 1$), and inverse probability weighting versus outcome regression---and all proofs are collected in Appendix~\ref{app:proofs-main}.}

\begin{table}[t]
\centering
\caption{The phase transition at $\tau^2=0$ organizes available inferential outputs. \emph{Estimation} means consistency of $\widehat\psi_{\mathrm{EB}}$ for the fixed target $\psi^*$; \emph{CI} means a confidence interval for $\psi^*$ (the sandwich of Proposition~\ref{prop:sandwich} requires exact identifiability); \emph{PI} means validity of prediction interval for the latent target $\psi_{J+1}$ of a new functional (conformal). At $\tau^2=0$ the prediction target is degenerate ($\psi_{J+1}=\psi^*$).}
\label{tab:regimes}
\renewcommand{\arraystretch}{1.3}
\setlength{\tabcolsep}{5pt}
\footnotesize
\begin{tabular}{lccc}
\toprule
Regime & Estimation of $\psi^*$ & CI for $\psi^*$ & PI for $\psi_{J+1}$ \\
\midrule
Exact identifiability ($\tau^2=0$)                              & $\checkmark$ (Thm~\ref{thm:consistency}) & $\checkmark$ sandwich / subsamp. & $\times$ \\
Centered heterogeneity ($\tau^2>0$, $E\varepsilon_j=0$) & $\checkmark$ (Thm~\ref{thm:route1})      & $\times$ & $\checkmark$ (Thm~\ref{thm:conformal}) \\
Uncentered ($\tau^2>0$)                                         & $\times$                                 & $\times$     & $\checkmark$ (Thm~\ref{thm:conformal}) \\
\bottomrule
\end{tabular}
\end{table}

	%==============================================================================
	\section{Maximum marginal likelihood}
	\label{sec:identifiability}

{Consider a collection of $J$ estimators $\widehat\psi_1,\ldots,\widehat\psi_J$, each arising from a distinct identifying functional of the same target $\psi$ and computed on a common sample. Because all of them are evaluated on the same observations, their joint sampling distribution obviously has non diagonal covariance. To obtain a tractable shrinkage rule, we adopt Gaussian hierarchical working model~\eqref{eq:working-model}, which treats the $\widehat\psi_j$ as conditionally independent given their latent targets $\psi_j$, with known marginal variances $v_j$. {This \emph{working independence device} is, at the end, an assumption---generally false, since the $\widehat\psi_j$ share one sample and are correlated---used only to derive the closed-form $\widehat\psi_{\mathrm{EB}}$, whose consistency does not require it.}

	{In the working model~\eqref{eq:working-model}, $\psi^*$ is the actual scalar causal target, $\psi_j=\psi_j(P)$ is the identifying target of the $j$-th functional, and $\varepsilon_j$ is its identification bias. We restrict attention to estimators $\widehat\psi_j=\psi_j(P_n)$ that are regular and asymptotically linear (RAL) \citep{tsiatis2006semiparametric}, i.e.\ there exists an efficient influence curve $D_j^*(O,P)$, with $E(D_j^*(O,P))=0$, such that
		\[
		\sqrt{n}\,(\widehat\psi_j-\psi_j(P))=\frac{1}{\sqrt{n}}\sum_{i=1}^n D_j^*(O_i,P)+o_p(1).
		\]
		By the central limit theorem, $\sqrt{n}\,(\widehat\psi_j-\psi_j(P))
    \xrightarrow{d}
N\!\big(0,\Var(D_j^*(O,P))\big)$ with the first order sampling variance of $\widehat\psi_j$ being
		$v_j=\frac{\Var(D_j^*(O,P))}{n}$.
}

{The Gaussian structure of~\eqref{eq:working-model}---both the prior $\varepsilon_j\sim N(0,\tau^2)$ and the likelihood $\widehat\psi_j\mid\psi_j\sim N(\psi_j,v_j)$---is a working device whose sole role is to fix the algebraic \emph{form} of the estimators, namely the inverse-variance (generalized least squares) combination $\widehat\psi_{\mathrm{EB}}$ and the pairwise $\widehat\tau^2$; none of the guarantees that follow inherit it. \emph{(i)} Consistency (Theorems~\ref{thm:consistency}--\ref{thm:route1}) uses only the convex-combination structure of $\widehat\psi_{\mathrm{EB}}$, so an arbitrary---skewed, multimodal, or discrete---prior $G$ leaves it intact. \emph{(ii)} The confidence interval (Proposition~\ref{prop:sandwich}) is a nonparametric, model-robust sandwich built from the observation-level influence curves, bypassing the parametric working model. \emph{(iii)} The conformal interval (Section~\ref{sec:bayesian}) is distribution-free in the continuous prior $G$ (Assumptions~(A1),~(A3)): the unknown limits of the training hyperparameters enter every score as a common shift and scale that the rank-based quantiles cancel. Normality therefore governs the \emph{efficiency} of the pooling, never the \emph{validity} of inference. The status of $\tau^2$ is itself regime-dependent: under the exchangeability Assumption~(A1) the $\varepsilon_j$ are i.i.d.\ from $G$ and $\tau^2=\operatorname{Var}_G(\varepsilon_j)=\int \varepsilon^2\,\mathrm{d}G(\varepsilon)$ is its variance functional, whereas in the centered heterogeneity regime of Theorem~\ref{thm:route1} it is just the uniform upper bound of condition~(C2).}

	{The three rows of Table~\ref{tab:regimes} correspond to three regimes of the mixing law. When every functional is valid, that is, $\psi_j(P)=\psi^*$ for all $j$, the mixing distribution in~\eqref{eq:working-model} degenerates at $\psi^*$ and $\tau^2=0$, so the $\widehat\psi_j$ differ only by sampling error. When, on the other hand, some functionals are invalid, then some deviations $\varepsilon_j=\psi_j(P)-\psi^*$ are nonzero and~\eqref{eq:working-model} becomes a random-effects structure with $\tau^2>0$ measuring identification heterogeneity, paralleling the between-study variance of random-effects meta-analysis and historical controls \citep{dersimonian2007random,collignon2020clustered}. A canonical example is multiple valid instruments $Z_j$, each identifying a local average treatment effect $\psi_j(P)=\E[Y(1)-Y(0)\mid\text{Complier}(Z_j)]$ over its own complier subpopulation under monotonicity and exclusion: relative to the population effect $\psi^*=\E[Y(1)-Y(0)]$ the deviations $\varepsilon_j=\psi_j(P)-\psi^*$ need not vanish even when all instruments are valid, giving $\tau^2>0$---the last row of Table~\ref{tab:regimes}.

Between this operating mode and every functional identifying the causal parameter lies the centered heterogeneity regime, $\tau^2>0$ with $E[\varepsilon_j]=0$: individual functionals remain biased, yet their precision-weighted average stays consistent for $\psi^*$ as $J$ grows (Section~\ref{sec:consistency}).}

	%==============================================================================
	%==============================================================================

 {Under model~\eqref{eq:working-model}, i.e. assuming normality and ignoring the cross functional dependence induced by shared data, the marginal log likelihood is}
	\[
	\ell(\psi,\tau^2)
	=
	-\frac12
	\sum_{j=1}^J
	\left\{
	\log(v_j+\tau^2)
	+
	\frac{(\widehat\psi_j -\psi)^2}{v_j+\tau^2}
	\right\}.
	\]

	\noindent The score equation in $\psi$, $\frac{\partial}{\partial \psi} \ell(\psi,\tau^2)=\sum_{j=1}^J (\widehat\psi_j-\psi)/(v_j+\tau^2)=0$, gives the maximum marginal likelihood estimator (MMLE) in closed form,
	\begin{equation}
\label{eq:eb-estimator}
\widehat\psi_{\mathrm{EB}}
=
\frac{\sum_{j=1}^J (v_j+\widehat\tau^2)^{-1}\,\widehat\psi_j}
{\sum_{j=1}^J (v_j+\widehat\tau^2)^{-1}}.
\end{equation}
	Differentiating $\ell(\psi,\tau^2)$ with respect to $\tau^2$ and setting the result to zero, the strict MMLE $\widehat\tau^2$ must satisfy the implicit equation
	\begin{equation}\label{eq:tau2_mmle_implicit}
	\sum_{j=1}^J
	\frac{(\widehat\psi_j -\widehat\psi_{\mathrm{EB}})^2 - (v_j+\widehat\tau^2)}{(v_j+\widehat\tau^2)^2}
	= 0,
		\end{equation}
	\label{implicit}

	\noindent where $\widehat\psi_{\mathrm{EB}}$ is the precision weighted mean in (\ref{eq:eb-estimator}).
	{Because $\widehat\tau^2$ appears in every denominator and the $v_j$ differ, this equation has no closed-form solution. We retain the closed-form pairwise estimator below for analytic transparency, since $\widehat\psi_{\mathrm{EB}}$ is consistent for any non-negative $\widehat\tau^2$ by Theorem~\ref{thm:consistency}.}
 
	{We bypass the implicit equation \eqref{implicit} via pairwise differences. Under the {working independence device}, the marginal law of $D_{jk}=\widehat\psi_j-\widehat\psi_k$ is $N(0,2\tau^2+v_j+v_k)$, so $E[(\widehat\psi_j-\widehat\psi_k)^2-(v_j+v_k)]=2\tau^2$. Averaging over the $\binom{J}{2}$ pairs of estimators yields an estimator algebraically equivalent to the standard unweighted variance estimator in random effects models \citep{hedges1983random},}

	\begin{equation}\label{tau2-pairwise}
	\widehat\tau^2 = \max\!\left\{
	\frac{1}{2\binom{J}{2}}
	\sum_{j<k}
	\Big[
	(\widehat\psi_j - \widehat\psi_k)^2 - (v_j + v_k)
	\Big],
	\;0
	\right\}.
	\end{equation}
	
	{The estimator \eqref{tau2-pairwise} differs from the strict MMLE characterized by~\eqref{eq:tau2_mmle_implicit}. Indeed, under the true joint law of $(\widehat\psi_1,\ldots,\widehat\psi_J)$, computed on the same sample,
	\begin{equation}\label{misv}
	E\!\left[(\widehat\psi_j-\widehat\psi_k)^2\right]
	= v_j+v_k+2\tau^2-2\Cov(\widehat\psi_j,\widehat\psi_k).
	\end{equation}
	Thus, ignoring the cross covariances makes the moment equation (\ref{eq:tau2_mmle_implicit}) identify a different population quantity rather than $\tau^2$ itself. The resulting $\widehat\tau^2$ is therefore a working nuisance parameter, not a consistent estimate of the structural heterogeneity variance. However, as Section~\ref{sec:consistency} will show, this does not invalidate consistency of the EB point estimator for $\psi^*$. This working-plug-in status is the canonical concern of the empirical Bayes critique tradition \citep{polson2026oldlook}; we defuse it through Proposition~\ref{prop:sandwich} and Theorem~\ref{thm:conformal} below by design rather than by hyperprior elicitation.}

	{Note that the framework does not require a formal pre-inference test for $\tau^2=0$: the sandwich confidence interval of Section~\ref{sec:frequentist} targets $\psi^*$, the conformal prediction interval of Section~\ref{sec:bayesian} targets $\psi_{J+1}$, and both can be reported jointly. The magnitude of $\widehat\tau^2$ relative to the typical $v_j$ then serves as an informal diagnostic of which output is the more informative. The exchangeability assumption is imposed on the latent targets $\psi_j$, not on the estimators $\widehat\psi_j$, whose variances $v_j$ are structurally distinct across $j$ and hence non-exchangeable. Appendix~\ref{app:extended-discussion} elaborates this diagnostic, the connection to the empirical Bayes $g$-modeling framework \citep{Efron2014}, and the two readings---epistemic-exchangeable versus deterministic-functional---of the hierarchical layer \citep{french2000statistical,wu2024bayesian}.}

	%==============================================================================

	\section{Point consistency}
	\label{sec:consistency}
	%==============================================================================
Let us first establish the convergence of the heterogeneity estimator under the working model and then show that the EB point estimator remains consistent for $\psi^*$ despite the misspecification originated from the working independence device. 
	
	\begin{proposition}[Convergence of the pairwise difference estimator $\widehat\tau^2$]
		\label{prop:tau2-consistency}
		{Under the working model~\eqref{eq:working-model}, the pairwise difference estimator $\widehat\tau^2$ converges in probability to some limit $\tau_0^2\ge 0$ as $n\to\infty$. If all identifying functionals share the same target, i.e.\ $\psi_j(P)=\psi^*$ for every $j$, then $\tau_0^2=0$.}
	\end{proposition}

	{ Under the true joint law, even under exact identifiability the pairwise moment satisfies $E[(\widehat\psi_j-\widehat\psi_k)^2]<v_j+v_k$ whenever $\mathrm{Cov}(\widehat\psi_j,\widehat\psi_k)>0$ as displayed in (\ref{misv}), which is the general situation for identifying functionals evaluated on a common sample. The unadjusted pairwise difference estimator therefore systematically underestimates the actual heterogeneity. The next theorem shows that this inconsistency of the nuisance parameter does not actually compromise the consistency of the target estimator.}{ In what follows, $\psi$ denotes the free argument of the score equations.}

	\begin{theorem}[Consistency under exact identifiability]
		\label{thm:consistency}
		{Suppose that $\widehat\psi_j$ is a consistent estimator of $\psi^*$ for all $j=1, \ldots , J$. Let $\widehat\tau^2$ be any estimator, rule, or fixed non negative constant with $\widehat\tau^2\xrightarrow{p}\tau_0^2$ for some $\tau_0^2\ge 0$. Then $\widehat\psi_{\mathrm{EB}}\xrightarrow{p}\psi^*$, regardless of $\tau_0^2$ and of the omitted cross functional dependence among the $\widehat\psi_j$. }
	\end{theorem}

	{Consistency of the point estimator follows from the convex combination representation. Valid standard errors do not follow automatically though, since the inverse Fisher information computed under the {working independence device} ignores the positive covariance between the $\widehat\psi_j$ and is anti conservative. Section~\ref{sec:frequentist} develops appropriate uncertainty quantification.}
    
Theorem~\ref{thm:route1} below does not require $\psi_j(P)=\psi^*$ for any individual $j$. It replaces exact identifiability with correct location of the prior and growth in the number of functionals.

\begin{theorem}[Consistency under centered heterogeneity]
\label{thm:route1}
{Let $J=J_n$ denote the number of identifying functionals. Assume:}
\begin{enumerate}
    \item[(C1)] Uniform estimator consistency: $\widehat\psi_j-\psi_j(P)\xrightarrow{p}0$ for each $j$, with $\sup_{j\le J_n} v_j=O(n^{-1})$, where $v_j=\operatorname{Var}(D_j^*(O,P))/n$ is the first order variance of $\widehat\psi_j$.\rev{\ In addition, the efficient influence curves are uniformly bounded in $L^2(P)$, $\sup_{j}\|D_j^*(O,P)\|_{L^2(P)}\le C<\infty$, and the number of functionals grows sub exponentially in the sample size, $\log J_n=o(n)$.}
    \item[(C2)] Correctly centered deviations: the identifying limits satisfy $\psi_j(P)=\psi^*+\varepsilon_j$, with $\varepsilon_1,\varepsilon_2,\ldots$ independent, $E[\varepsilon_j]=0$, and $\sup_j E[\varepsilon_j^2]\le\tau^2<\infty$ (here $\tau^2$ is a uniform bound on the second moments, not a common variance).
    \item[(C3)] Convergent nuisance parameter: $\widehat\tau^2\xrightarrow{p}\tau_0^2$ for some $\tau_0^2\ge 0$.
    \item[(C4)] Non dominant weights: $\max_{1\le j\le J_n} \phantom{s} w_j\xrightarrow{J_n \to \infty } 0$, where $w_j=(v_j+\tau_0^2)^{-1}/\sum_{k=1}^{J_n}(v_k+\tau_0^2)^{-1}$.
\end{enumerate}
{Then $\widehat\psi_{\mathrm{EB}}\xrightarrow{p}\psi^*$. }
\end{theorem}

{Condition (C4) requires that no single functional asymptotically absorbs all the magnitude of the precision weights; it holds whenever $J\to\infty$ and the $v_j$ are uniformly bounded away from zero, as with Wald/IV estimators where $v_j\ge c/n$ under non-weak instruments (then $w_j\le \mathrm{const}\cdot(J\tau_0^2)^{-1}\to 0$ when $\tau_0^2>0$, and $(\max_j v_j^{-1})/(\sum_k v_k^{-1})\to 0$ at the boundary $\tau_0^2=0$). The structural primacy of $J\to\infty$ over $n\to\infty$ here echoes \citet{collignon2020clustered}: when replication sits at the level of design rather than observation, no within-unit sample size substitutes for replication across units, and the irreducible between-unit variance bounds the variance of any combined estimator from below.}

{Importantly, note that both theorems are non-nested: Theorem~\ref{thm:consistency} requires $\psi_j(P)=\psi^*$ for all $j$, but permits any fixed $J$, whereas Theorem~\ref{thm:route1} relaxes exact identifiability to centered biases at the cost of $J\to\infty$ and non-dominant weights. Moreover, their inferential implications differ under exact identifiability the asymptotic cross-partial $\sum_j-(\psi_j(P)-\psi^*)/(v_j+\tau^2)^2$ vanishes, global $P$-orthogonality holds \citep{battey2024roleparametrizationmodelsmisspecified}, and the sandwich variance of Section~\ref{sec:frequentist} is valid. In turn, under centered biases $P$-orthogonality fails and the appropriate outputs are subsampling intervals for $\psi^*$ or split conformal intervals for $\psi_{J+1}$ (Section~\ref{sec:bayesian}).}

%==============================================================================
\section{Prediction intervals for a new identifying target}
\label{sec:bayesian}
%==============================================================================

{This section addresses a different inferential target: a prediction interval for the latent target $\psi_{J+1}$ of a new identifying functional drawn from the same population as $\psi_1,\ldots,\psi_J$. Therefore we aim for a random future quantity rather than the fixed $\psi^*$. Let us recall the structural basis of the hierarchical prior and derive the split conformal procedure.}

{The latent targets $\psi_1,\psi_2,\ldots$ being assumed to be i.i.d.\ draws from a prior $G$ follows from a symmetry judgment on the identifying functionals. Treating them as exchangeable before observing the data commits one, via de Finetti's theorem \citep{french2000statistical}, to a mixture representation $P(\psi_1,\ldots,\psi_J)=\int\prod_{j=1}^J P(\psi_j\mid\theta)\,\mathrm{d}\pi(\theta)$, of which~\eqref{eq:working-model} is the Gaussian member.}

{Here the paper's two kinds of guarantee meet, and the distinction is one of \emph{evaluation criterion}, not of method: a frequentist criterion averages over datasets at a fixed parameter, a Bayesian one also averages over the parameter. The confidence interval for $\psi^*$ (Section~\ref{sec:frequentist}) is frequentist---its target is fixed and the probability runs over the data alone---whereas the conformal interval below is Bayesian, since according to Theorem~\ref{thm:conformal}, the coverage probability is taken over the latent law $(\psi_1,\dots,\psi_{J+1})\sim G^{J+1}$ as well as the data, averaging over the distribution of the random target it covers. The frequentist interval for $\psi^*$ holds in the exact identifiability regime---by the sandwich of Proposition~\ref{prop:sandwich} under exact identifiability---and the Bayesian guarantee for $\psi_{J+1}$ when $\tau^2>0$.}

{At $\tau^2=0$, only the correlated sampling noise $\xi_{jn}=n^{-1}\sum_i D_j^*(O_i,P)$ makes the $\widehat\psi_j$ vary and no i.i.d.\ signal remains, whereas at $\tau^2>0$ the i.i.d.\ latent deviations ($O_p(1)$) dominate the shrinking noise ($O_p(n^{-1/2})$) and conformal prediction becomes viable, as the score decomposition (\ref{eq:score-decomp}) makes precise.}

{Let us therefore adopt a split conformal construction that separates hyperparameter estimation from calibration. }{For this, partition randomly the $J$ indices into two disjoint sets: a training set $\mathcal{T}$ of size $J_{\mathrm{train}}=\lfloor J/2\rfloor$ and a calibration set $\mathcal{C}$ of size $J_{\mathrm{cal}}=J-J_{\mathrm{train}}$. The procedure is as follows:}

\begin{enumerate}
	\item[(i)] \textit{Training step.} Estimate the EB hyperparameters from $\mathcal{T}$ alone:
    \[
    \widehat\tau^2_{\mathrm{train}} = \max\!\left\{
    \frac{1}{2\binom{J_{\mathrm{train}}}{2}}\sum_{\substack{j < k \\ j,k \in \mathcal{T}}}
    \bigl[(\widehat\psi_j - \widehat\psi_k)^2 - (v_j + v_k)\bigr],\;0\right\},
    \]
    and $\widehat\psi_{\mathrm{EB}}^{\mathrm{train}}=\sum_{j \in \mathcal{T}}(v_j + \widehat\tau^2_{\mathrm{train}})^{-1}\widehat\psi_j\big/\sum_{j \in \mathcal{T}}(v_j + \widehat\tau^2_{\mathrm{train}})^{-1}$ as in~\eqref{eq:eb-estimator} and \eqref{tau2-pairwise} restricted to $\mathcal{T}$.
    \item[(ii)] \textit{Calibration step.} For each $j\in\mathcal{C}$, compute the split conformal score
    \[
    \widehat S_j^{\mathrm{split}} = \frac{\widehat\psi_j - \widehat\psi_{\mathrm{EB}}^{\mathrm{train}}}{\sqrt{\widehat\tau^2_{\mathrm{train}} + v_j}}.
    \]
	\item[(iii)] \textit{Prediction step.} A $(1-\alpha)$ prediction interval for $\psi_{J+1}$ is constructed by inverting the conformity score. Let $\widehat{q}_{\alpha/2}$ and $\widehat{q}_{1-\alpha/2}$ be the empirical $\alpha/2$ and $1-\alpha/2$ quantiles of the calibration scores $\{\widehat S_j^{\mathrm{split}}\}_{j\in\mathcal{C}}$. The two-sided prediction interval is
	\begin{equation}\label{eq:conformal-interval}
	C_{1-\alpha} := \left[\widehat\psi_{\mathrm{EB}}^{\mathrm{train}} + \widehat{q}_{\alpha/2}\sqrt{\widehat\tau^2_{\mathrm{train}} + v_{J+1}}, \quad \widehat\psi_{\mathrm{EB}}^{\mathrm{train}} + \widehat{q}_{1-\alpha/2}\sqrt{\widehat\tau^2_{\mathrm{train}} + v_{J+1}}\right].
	\end{equation}
	This interval has asymptotic coverage $\geq 1-\alpha$ for $\psi_{J+1}$ under the conditions of Theorem~\ref{thm:conformal} (the quantiles are formally the $\lceil (1-\alpha)(J_{\mathrm{cal}}+1)\rceil$-th order statistics, asymptotically equivalent to the empirical ones). As shown in its proof (see Appendix~\ref{app:proofs-main}), the scores converge in distribution to $(\varepsilon_j - (\mu-\psi^*))/\sigma$, with $\mu,\sigma^2$ the probability limits of $\widehat\psi_{\mathrm{EB}}^{\mathrm{train}},\widehat\tau^2_{\mathrm{train}}$; the common shift and scale do not affect coverage.
\end{enumerate}

\noindent {Sample splitting removes the self-influence of hyperparameter fit: since $\widehat\psi_{\mathrm{EB}}^{\mathrm{train}}$ and $\widehat\tau^2_{\mathrm{train}}$ depend only on $\mathcal{T}$, the calibration scores are not contaminated by the fit. It does not, however, erase the shared-data correlation: distinct calibration indices are still evaluated on the same $n$ observations, so the validity argument relies on asymptotic dominance of the i.i.d.\ component over the correlated noise, not on exact finite-sample exchangeability.}

{Combining~\eqref{eq:working-model} with the asymptotic linearity of each $\widehat\psi_j$, write $\widehat\psi_j=\psi^*+\varepsilon_j+\xi_{jn}$, where $\varepsilon_j=\psi_j-\psi^*$ is the latent deviation and $\xi_{jn}=\widehat\psi_j-\psi_j=n^{-1}\sum_{i=1}^{n} D_j^*(O_i,P)$ is the shared data sampling noise for $j$. As a consequence, for a calibration index $j\in\mathcal{C}$, the split score $S_j^{\mathrm{split}}$ admits the decomposition}
\begin{equation}
\label{eq:score-decomp}
\widehat S_j^{\mathrm{split}} = \underbrace{\frac{\varepsilon_j}{\sqrt{\widehat\tau^2_{\mathrm{train}} + v_j}}}_{\text{(I) oracle signal}}
+ \underbrace{\frac{\xi_{jn}}{\sqrt{\widehat\tau^2_{\mathrm{train}} + v_j}}}_{\text{(II) shared-data noise}}
+ \underbrace{\frac{\psi^* - \widehat\psi_{\mathrm{EB}}^{\mathrm{train}}}{\sqrt{\widehat\tau^2_{\mathrm{train}} + v_j}}}_{\text{(III) training error}},
\end{equation}
This decomposition is algebraically exact: all three terms share the denominator $\sqrt{\widehat\tau^2_{\mathrm{train}}+v_j}$, so no cross-term remainder arises.

{Importantly, the three terms behave very differently asymptotically. Term~(I), the \emph{oracle signal}, converges to $\varepsilon_j/\sigma$ with the $\varepsilon_j$ i.i.d.\ from $G$, since $v_j=O(n^{-1})\to0$ and $\widehat\tau^2_{\mathrm{train}}\xrightarrow{p}\sigma^2>0$ by (A3): this is the exchangeable component, and any finite positive limit $\sigma$ suffices. Term~(II), the \emph{shared-data noise}, is $O_p(n^{-1/2})$ and correlated across $j$ through $\operatorname{Cov}(\xi_{jn},\xi_{kn})=n^{-1}\operatorname{Cov}(D_j^*,D_k^*)$: it is the non-exchangeable component, required to be uniformly negligible by Assumption~(A2) below. Term~(III), the \emph{training error}, converges to the constant $c:=(\psi^*-\mu)/\sigma$, identical for every calibration index and for the test index $J+1$, therefore acting as a common shift not affecting coverage.}

{Hence, when $\tau^2>0$ the $O_p(1)$ i.i.d.\ signal dominates the $O_p(n^{-1/2})$ correlated noise and the scores become asymptotically approximately exchangeable, whereas when $\tau^2=0$ term~(I) vanishes and only the correlated term~(II) survives, leaving no dominance argument. This formalizes the phase transition: the conformal interval is asymptotically valid if and only if $\tau^2>0$. The assumptions governing asymptotic marginal validity of $C_{1-\alpha}$ as defined in (\ref{eq:conformal-interval}) are as follows.}

\begin{enumerate}
    \item[(A1)] {\textit{Latent i.i.d.\ heterogeneity}: $\psi_j=\psi^*+\varepsilon_j$ with $\varepsilon_j\stackrel{\mathrm{i.i.d.}}{\sim} G$, $E[\varepsilon_j]=0$, $\operatorname{Var}(\varepsilon_j)=\tau^2>0$, and $G$ has a continuous strictly increasing cdf $F_G$. Here $G$ is otherwise arbitrary---in particular, need not be Gaussian---and $\tau^2=\int\varepsilon^2\,\mathrm{d}G(\varepsilon)$ is the variance functional of $G$, not a free parameter.}
    \item[(A2)] {\textit{Uniform asymptotic negligibility of shared data noise}: as $n\to\infty$,
    \[
    \max_{j\in\mathcal{C}\cup\{J+1\}} \frac{|\xi_{jn}|}{\tau} = o_p(1).
    \]
    A sufficient condition is $\sup_j \|D_j^*(O,P)\|_{L^2(P)}\le C<\infty$ together with $n^{-1}\log J_{\mathrm{cal}}\to 0$. A sub Gaussian union bound then yields $\max_j |\xi_{jn}|=O_p(\sqrt{\log J_{\mathrm{cal}}/n})=o_p(\tau)$ (Lemma~\ref{lem:maximal-xi}, Appendix~\ref{app:maximal}).}
	\item[(A3)] \textit{Convergent nuisance parameter:} As $n,J_{\mathrm{train}}\to\infty$, the training-set hyperparameters converge in probability to fixed limits: $\widehat\tau^2_{\mathrm{train}}\xrightarrow{p}\sigma^2$ for some $\sigma^2>0$ and $\widehat\psi_{\mathrm{EB}}^{\mathrm{train}}\xrightarrow{p}\mu$ for some $\mu\in\mathbb{R}$. Neither consistency for the true parameters ($\sigma^2=\tau^2$, $\mu=\psi^*$) is required.
\end{enumerate}

\noindent {The strict inequality $\tau^2>0$ in (A1) is the content of the phase transition, not a technical artifact: at $\tau^2=0$ the scores reduce to the correlated term~(II) of~\eqref{eq:score-decomp} and the empirical distribution of calibration scores cannot be controlled by a Glivenko--Cantelli type of argument.}

\begin{theorem}\label{thm:conformal}
Suppose Assumptions (A1)--(A3) hold. Then, as $n, J_{\mathrm{cal}} \to \infty$:
\[
\liminf_{n,J_{\mathrm{cal}}\to\infty}\,P \bigl(\psi_{J+1} \in C_{1-\alpha}\bigr) \;\geq\; 1-\alpha,
\]
where the probability is over the joint law $(\psi_1,\ldots,\psi_{J+1})\sim G^{J+1}$, the data $(O_1,\ldots,O_n)\sim P^n$, and the random training/calibration split.
\end{theorem}

\noindent {This guarantee is asymptotic and marginal. Exact finite-sample validity would require the calibration and test scores to be jointly exchangeable given the training fit, which fails because the calibration indices share the same $n$ observations: it is attainable only when each functional is evaluated on an independent dataset. The marginal statement is also unconditional over the latent targets and the split. A weaker training-conditional version, attainable at a $J_{\mathrm{cal}}^{-1/2}$ rate, is given by Proposition~\ref{prop:train-cond} in Appendix~\ref{app:train-cond}.}

{Note that our conformal interval and the frequentist interval of Section~\ref{sec:frequentist} address different quantities---the random draw $\psi_{J+1}$ versus the fixed $\psi^*$. Theorem 3 requires centered biases and growing $J$, the latter i.i.d.\ $\varepsilon_j$ with $\tau^2>0$ and a convergent positive-limit scale (A3). Notably, conformal validity does \emph{not} require consistency of $\widehat\psi_{\mathrm{EB}}^{\mathrm{train}}$ or $\widehat\tau^2_{\mathrm{train}}$: the limits $\mu,\sigma^2$ enter only as a common shift and scale that cancel, so the interval covers $\psi_{J+1}$ with probability $\geq 1-\alpha$ asymptotically even under uncentered bias or misspecified heterogeneity---covering, in that case, a draw from the same (possibly misspecified) population of functionals.}

	%==============================================================================
	\section{Confidence intervals for the causal effect}
	\label{sec:frequentist}

{This section works under exact identifiability, where $\psi_j(P)=\psi^*$ for all $j=1,\ldots,J$. The estimator $\widehat\psi_{\mathrm{EB}}$ solves $h(\psi,\widehat\tau^2)=0$ and is consistent for $\psi^*$ by Theorem~\ref{thm:consistency}, regardless of the limit $\tau_0^2$ of $\widehat\tau^2$. Standard M estimation theory \citep{godambe1960optimum,vaart1998,vandegeer2000empirical} delivers asymptotic normality of $n^{1/2}(\widehat\psi_{\mathrm{EB}}-\psi^*)$. Both of the procedures developed below---a sandwich variance and a subsampling interval---produce valid confidence intervals for $\psi^*$. The centered bias setting of Theorem~\ref{thm:route1}, with $\psi_j(P)\neq\psi^*$ but $E[\varepsilon_j]=0$ and $J\to\infty$, is addressed through subsampling. The uncentered bias case is left as an open problem; see Section~\ref{sec:conclusion}.}

%==============================================================================
\label{subsec:sandwich}

Our estimator $\widehat\psi_{\mathrm{EB}}$ has been defined as the root of the estimating equation $h(\psi,\tau^2)
	\;=\;
	\sum_{j=1}^J \frac{\widehat\psi_j - \psi}{v_j + \tau^2}
	\;=\; 0.$ The general sandwich variance \citep{huber1967behavior,white1982maximum,vaart1998} is $V_{\mathrm{sandwich}} = A^{-1}BA^{-1}$, where the bread and meat are (in our case, all scalars)
\begin{equation}
	\label{eq:sandwich_components}
	A
	= \sum_{j=1}^J \frac{1}{v_j + \tau^2_0},
	\qquad
	B
	= \operatorname{Var}_P\!\left[\sum_{j=1}^J\frac{D_j^*(O,P)}{v_j+\tau^2_0}\right]
	= \sum_{j,k=1}^J
	\frac{\operatorname{Cov}_P(\widehat\psi_j,\widehat\psi_k)}{(v_j + \tau^2_0)(v_k + \tau^2_0)},
\end{equation}
with $\tau_0^2$ denoting the probability limit of $\widehat\tau^2$. The second expression for $B$ uses the asymptotic linearity of each $\widehat\psi_j$ together with $\operatorname{Cov}_P(\widehat\psi_j,\widehat\psi_k)=n^{-1}\operatorname{Cov}_P(D_j^*(O,P),D_k^*(O,P))$. As the estimators are evaluated on the same sample, the actual $B$ exceeds $A$, so the canonical inverse Fisher standard errors are anti conservative. The sandwich variance corrects this nonparametrically.

%==============================================================================

\begin{proposition}[Sandwich robustness under nuisance inconsistency]
	\label{prop:sandwich}
	Suppose $\psi_j(P)=\psi^*$ for all $j$, and let $\widehat\tau^2 \xrightarrow{p} \tau^2_0$ for some $\tau^2_0 \ge 0$. Then:
	\begin{enumerate}
		\item[\emph{(i)}] The asymptotic distribution of $n^{1/2}(\widehat\psi_{\mathrm{EB}} - \psi^*)$
		      is the same for any probability limit $\tau^2_0 \ge 0$ of $\widehat\tau^2$.
		\item[\emph{(ii)}] The sandwich variance estimator $\widehat V_{\mathrm{sandwich}} = \widehat A^{-1}\widehat B\widehat A^{-1}$,
		      with $\widehat A$ and $\widehat B$ evaluated at $(\widehat\psi_{\mathrm{EB}}, \widehat\tau^2)$,
		      is consistent for $V_{\mathrm{sandwich}}$.
	\end{enumerate}
\end{proposition}

\noindent {Since the asymptotic law of $n^{1/2}(\widehat\psi_{\mathrm{EB}}-\psi^*)$ does not depend on the probability limit $\tau_0^2$ of the plug-in, the failure to propagate hyperparameter uncertainty---the source of finite-sample undercoverage in classical EB constructions \citep{laird1987empirical,carlin1990approaches}---has no asymptotic cost on inference for $\psi^*$ under exact identifiability.}

%==============================================================================
\subsection{Practical implementation of the sandwich}
\label{subsec:sandwich-impl}

{Each $\widehat\psi_j$ is asymptotically linear, so the score decomposes at the observation level as $h(\psi^*,\tau^2_0)=n^{-1}\sum_{i=1}^n \tilde g_i + o_p(n^{-1/2})$, with $\tilde g_i:=\sum_{j=1}^J D_j^*(O_i,P)/(v_j+\tau^2_0)$. The meat $B=\operatorname{Var}_P[\tilde g(O)]$ is consistently estimated by the observation-level outer product}
\begin{equation}
	\label{eq:Bhat}
	\widehat B
	\;=\;
	\frac{1}{n}\sum_{i=1}^n
	\left(
		\sum_{j=1}^J \frac{\widehat D_j^*(O_i)}{v_j + \widehat\tau^2}
	\right)^{\!2},
\end{equation}
where $\widehat D_j^*(O_i)$ is the estimated efficient influence curve of functional $j$, available in closed form for all RAL estimators \citep{tsiatis2006semiparametric}; equivalently $\widehat B=\sum_{j,k}\widehat{\operatorname{Cov}}_n(\widehat\psi_j,\widehat\psi_k)/[(v_j+\widehat\tau^2)(v_k+\widehat\tau^2)]$ is the full $J\times J$ cross-functional covariance contraction, with $\widehat{\operatorname{Cov}}_n(\widehat\psi_j,\widehat\psi_k)=n^{-2}\sum_i \widehat D_j^*(O_i)\widehat D_k^*(O_i)$. Writing $\widehat A = \sum_{j} (v_j + \widehat\tau^2)^{-1}$, the practical variance estimator is $\widehat V_{\mathrm{sandwich}}=\widehat A^{-2}\widehat B$ and the Wald-type confidence interval is $\widehat\psi_{\mathrm{EB}}\pm z_{1-\alpha/2}\,\widehat V_{\mathrm{sandwich}}^{1/2}/n^{1/2}$.

{Under exact identifiability, $\widehat\tau^2$ lands on the boundary $\widehat\tau^2=0$ with positive probability, a case Proposition~\ref{prop:sandwich} covers directly: $P$-orthogonality makes the limit law of $n^{1/2}(\widehat\psi_{\mathrm{EB}}-\psi^*)$ invariant to $\tau_0^2$ (including $\tau_0^2=0$) and $\widehat V_{\mathrm{sandwich}}$ remains consistent. The boundary non-regularity of \citet{andrews2000} affects inference on $\tau^2$ and the bootstrap of $\widehat\tau^2$, not sandwich inference on $\psi^*$.}

	%==============================================================================
	
	%==============================================================================

	%==============================================================================

%==============================================================================
\section{Simulation studies}
\label{sec:simulations-main}
%==============================================================================

{We validate the framework in two simulation studies, summarized in Figure~\ref{fig:sim-main} and reported in full in Appendix~\ref{sec:simulations}. The first one isolates the exact-identifiability regime and confirms that $\widehat\psi_{\mathrm{EB}}$ is consistent at the parametric $n^{-1/2}$ rate even though the model~\eqref{eq:working-model} discards the cross-functional covariance (Theorem~\ref{thm:consistency}), and that only the sandwich and subsampling intervals of Section~\ref{sec:frequentist} retain nominal coverage. The second verifies the conformal guarantee of Theorem~\ref{thm:conformal} for the latent target $\psi_{J+1}$ of a new functional. Both findings are distribution-free in $G$.}

\begin{figure}[t!]
  \centering
  \begin{minipage}[t]{0.49\textwidth}\centering
    \includegraphics[width=\textwidth,trim={0 0 466pt 0},clip]{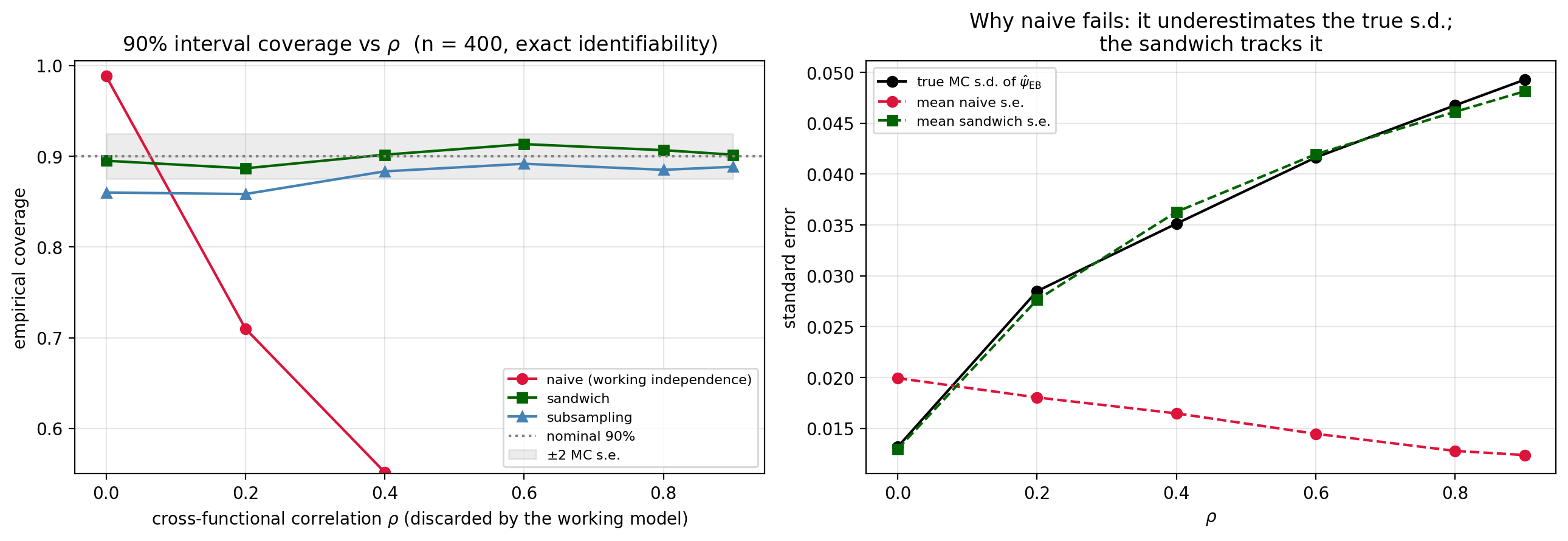}\\[2pt]
    {\small (a) Exact identifiability: coverage vs.\ $\rho$}
  \end{minipage}\hfill
  \begin{minipage}[t]{0.49\textwidth}\centering
    \includegraphics[width=\textwidth,trim={0 0 465pt 0},clip]{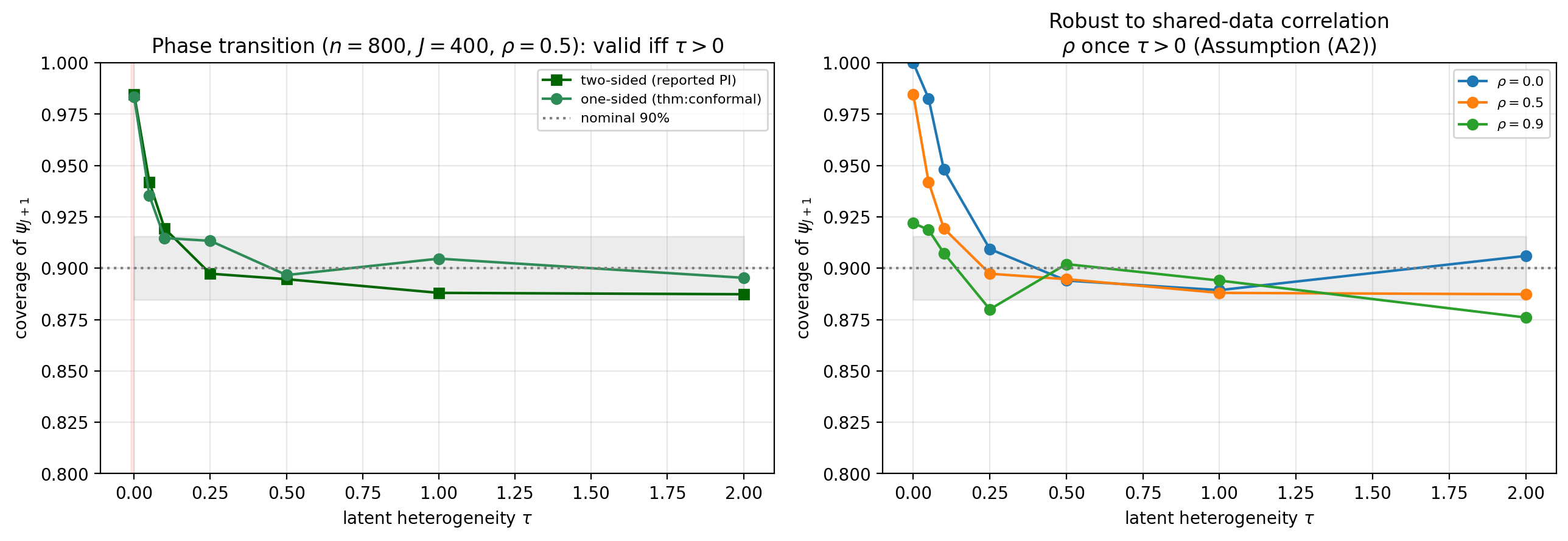}\\[2pt]
    {\small (b) Heterogeneous regime: coverage vs.\ $\tau$}
  \end{minipage}
  \caption{Headline simulation results (full study in Appendix~\ref{sec:simulations}). \emph{(a)} Empirical coverage of the three nominal $90\%$ confidence intervals for $\psi^*$ at $n=400$, against the discarded cross-functional correlation $\rho$: the canonical inverse-Fisher interval collapses toward $0.33$ coverage as $\rho$ grows, while the observation-level sandwich and the subsampling interval hold the nominal level. \emph{(b)} Coverage of the conformal prediction interval for $\psi_{J+1}$ against the latent heterogeneity $\tau$ (at $\rho=0.5$): the guarantee holds for every $\tau>0$ and fails exactly at $\tau=0$ (red), the empirical signature of the phase transition. Grey bands are $\pm 2$ Monte Carlo standard errors around the nominal level.}
  \label{fig:sim-main}
\end{figure}

{Both panels make the regime split concrete. In panel~(a) the inverse-Fisher interval does not merely undercover---it \emph{narrows} as $\rho$ grows, becoming more confident precisely as the omitted dependence inflates the true variance of $\widehat\psi_{\mathrm{EB}}$, whereas the sandwich tracks that variance and holds the nominal line (Proposition~\ref{prop:sandwich}). In panel~(b) the conformal interval stays flat across the whole $\tau>0$ range, insensitive to $\rho=\operatorname{Cor}(\widehat\psi_j,\widehat\psi_k)$ (Assumption~(A2)), and collapses only at the boundary $\tau=0$.}

%==============================================================================
\section{Applications}
\label{sec:applications}
%==============================================================================

{We illustrate the framework on two canonical datasets that stress its two regimes in opposite directions. The National Supported Work (NSW) data \citep{lalonde1986,dehejia1999} combine a randomized trial with several observational comparison groups: when the pooled identification functionals using external controls are generated by the conditional version of mean untreated potential outcome exchangeability, $\widehat\tau^2$ collapses to the boundary and the sandwich confidence interval contracts sharply around the experimental benchmark. 
	
	Second, the \citep{card1994} minimum-wage study, in contrast, presents multiple Pennsylvania control subgroups whose discrepancy has been debated for three decades \citep{neumark2000minimum,dube2010minimum}. There, $\widehat\tau^2$ is large relative to the typical $v_j$, the shrinkage flattens, and the natural output is the conformal prediction interval for a new functional target. 
	
	The framework is therefore self-auditing: the magnitude of $\widehat\tau^2$ identifies which regime applies and which inferential object is informative. Full estimator construction and per-functional tables are in Appendices~\ref{app:lalonde}--\ref{app:ck94}; here we report the regime-level summaries.}

%------------------------------------------------------------------------------
\subsection{National Supported Work: pooling when borrowing is safe}
\label{subsec:app-lalonde}
%------------------------------------------------------------------------------

{The target is the average treatment effect on the treated (ATT). We construct $J=15$ identifying functionals: three RCT only ($j=1,2,3$), eight covariate adjusted external control augmentations across four CPS/PSID survey cohorts (CPS-2, CPS-3, PSID-2, PSID-3) ($j=5,6,8,9,11,12,14,15$), and four crude unadjusted external comparisons (EC) ($j=4,7,10,13$). Three nested sub-analyses are ran---RCT only ($J=3$), covariate adjusted ($J=11$), full pool ($J=15$)---and report $\widehat\tau$ and the EB output in Table~\ref{tab:app-lalonde}; full construction is in Appendix~\ref{app:lalonde}.}

\begin{table}[t!]
\centering
\caption{EB output for the two applications. \emph{(a)} National Supported Work, three nested sub-analyses (experimental benchmark \$1{,}794). \emph{(b)} New Jersey minimum wage, heterogeneity $\widehat\tau$ by estimation method; the conformal prediction interval for $\psi_{J+1}$ at $1-\alpha\approx 92.3\%$ is $[-2.60,\,7.85]$ FTE, with $\widehat\psi_{\mathrm{EB}}=2.54$ FTE against the \citet{card1994} benchmark of $2.75$ FTE.}
\label{tab:app-lalonde}\label{tab:app-ck}
\renewcommand{\arraystretch}{1.15}
\begin{minipage}[t]{0.48\textwidth}\centering
{\small (a) National Supported Work}\\[3pt]
\begin{tabular}{lcc}
\toprule
Sub-analysis & $J$ & $\widehat\psi_{\mathrm{EB}}$ (SE) \\
\midrule
RCT-only             & 3  & \$1{,}770\ (\$389) \\
Covariate adjusted   & 11 & \$1{,}552\ (\$228) \\
Full pool            & 15 & \$589\ (\$525) \\
\bottomrule
\end{tabular}
\end{minipage}\hfill
\begin{minipage}[t]{0.48\textwidth}\centering
{\small (b) Card--Krueger}\\[3pt]
\begin{tabular}{lcc}
\toprule
Sub-analysis & $J$ & $\widehat\tau$ (FTE) \\
\midrule
DIM only             & 9  & 2.14 \\
OM only              & 8  & 2.51 \\
IPW only             & 8  & 2.46 \\
Full pool            & 25 & 2.38 \\
\bottomrule
\end{tabular}
\end{minipage}
\end{table}

{Two regimes coexist within the same dataset. Across the first two sub-analyses $\widehat\tau$ collapses to the boundary and EB shrinkage compresses the standard error substantially: within the RCT-only pool ($J=3$) it falls from \$671 to \$389, a $42\%$ reduction; within the covariate-adjusted pool ($J=11$) it falls further to \$228, a $66\%$ reduction relative to the trial-only benchmark. The sample size is identical across sub-analyses: \textit{the precision gain comes entirely from identification multiplicity, not more data}. Covariate-adjusted external controls add genuine information without detectable identification bias. Including the four crude EC functionals drives $\widehat\tau$ to \$1{,}852, an order of magnitude larger than the individual $\sqrt{v_j}$. The pooled estimate collapses to \$589 with the heterogeneity estimate transparently flagging the inclusion of invalid functionals. The failure of marginal exchangeability in the crude EC pool was anticipated a priori: NSW participants were recruited from a severely disadvantaged population---former drug addicts, ex-convicts, and the long-term unemployed---whose pre-enrollment earnings were a fraction of those recorded for CPS and PSID survey respondents \citep{lalonde1986,smith2005}, so the marginal distributions of untreated potential outcomes were unlikely to be comparable without covariate adjustment. \textit{The same machinery thus delivers precision when borrowing is safe and an automatic flag when it is not.}

%------------------------------------------------------------------------------
\subsection{New Jersey minimum-wage increase: honest uncertainty when controls disagree}
\label{subsec:app-ck}
%------------------------------------------------------------------------------

{The target is the ATE of the 1992 New Jersey (NJ) minimum-wage increase on fast-food employment. The dataset of \citet{card1994} contains 309 NJ restaurants and 75 Pennsylvania (PA) restaurants surveyed before ($t=0$) and after ($t=1$) the increase. Let $Y_t(d)$ denote potential FTE employment and $A\in\{0,1\}$ the state indicator (NJ vs.\ PA). For a PA subgroup $G_s$, the \emph{parallel-trends assumption} asserts
\begin{equation}\label{eq:pt}
\E[Y_1(0)-Y_0(0)\mid A=1] \;=\; \E[Y_1(0)-Y_0(0)\mid A=0,\,G_s],
\end{equation}
identifying the ATE via $\psi_s(P)=\E[Y_1-Y_0\mid A=1]-\E[Y_1-Y_0\mid A=0,G_s]$. Outcome Modelling (OM) and Inverse Probability Weighting (IPW) replace~\eqref{eq:pt} with its conditional analogue given restaurant-level covariates $X$ (chain, ownership, region). The nine subgroupings provide each a different substantive rationale for~\eqref{eq:pt}:}

\begin{itemize}
  \item \emph{Regional} (PA-all, PA-north, PA-east): geographic proximity justifies trend similarity.
  \item \emph{Chain} (Burger King, KFC, Roy Rogers, Wendy's): chain-level shocks cancel across the state boundary.
  \item \emph{Ownership} (corporate, franchise): organizational structure controls for wage-setting discretion.
\end{itemize}

{For each subgroup we compute Difference in Means (DIM), OM, and IPW estimators (Burger King excluded from OM/IPW: no covariate variation), giving $J=9+8+8=25$ functionals with estimates spanning $-1.83$ to $+7.10$ FTE. The literature has documented persistent disagreement: \citet{neumark2000minimum} report negative effects under payroll-based controls, \citet{dube2010minimum} null effects under contiguous-county controls. This disagreement---across analyses that all invoke~\eqref{eq:pt}---is exactly the object $\tau^2$ is designed to quantify (Table~\ref{tab:app-ck}).}

\begin{figure}[!htb]
  \centering
  \begin{minipage}[t]{0.49\textwidth}\centering
    \includegraphics[width=\textwidth]{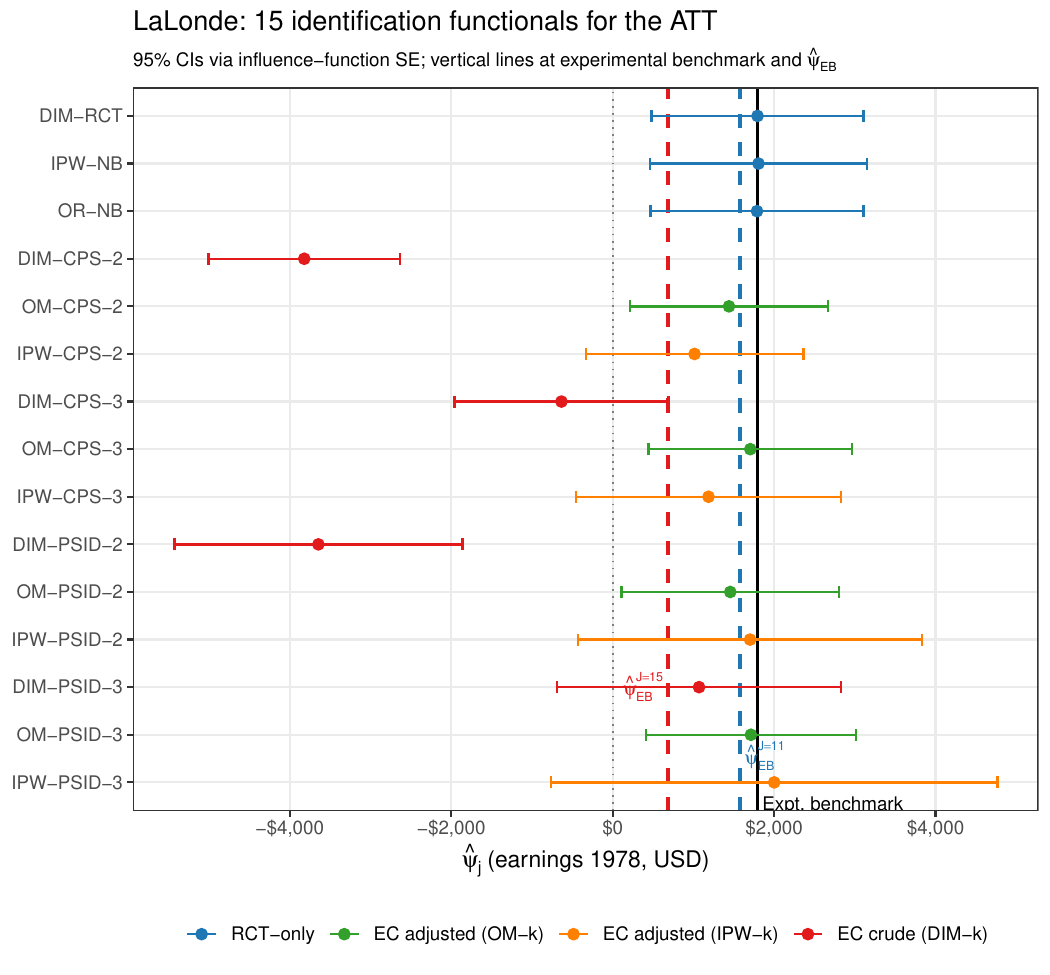}
  \end{minipage}\hfill
  \begin{minipage}[t]{0.49\textwidth}\centering
    \includegraphics[width=\textwidth]{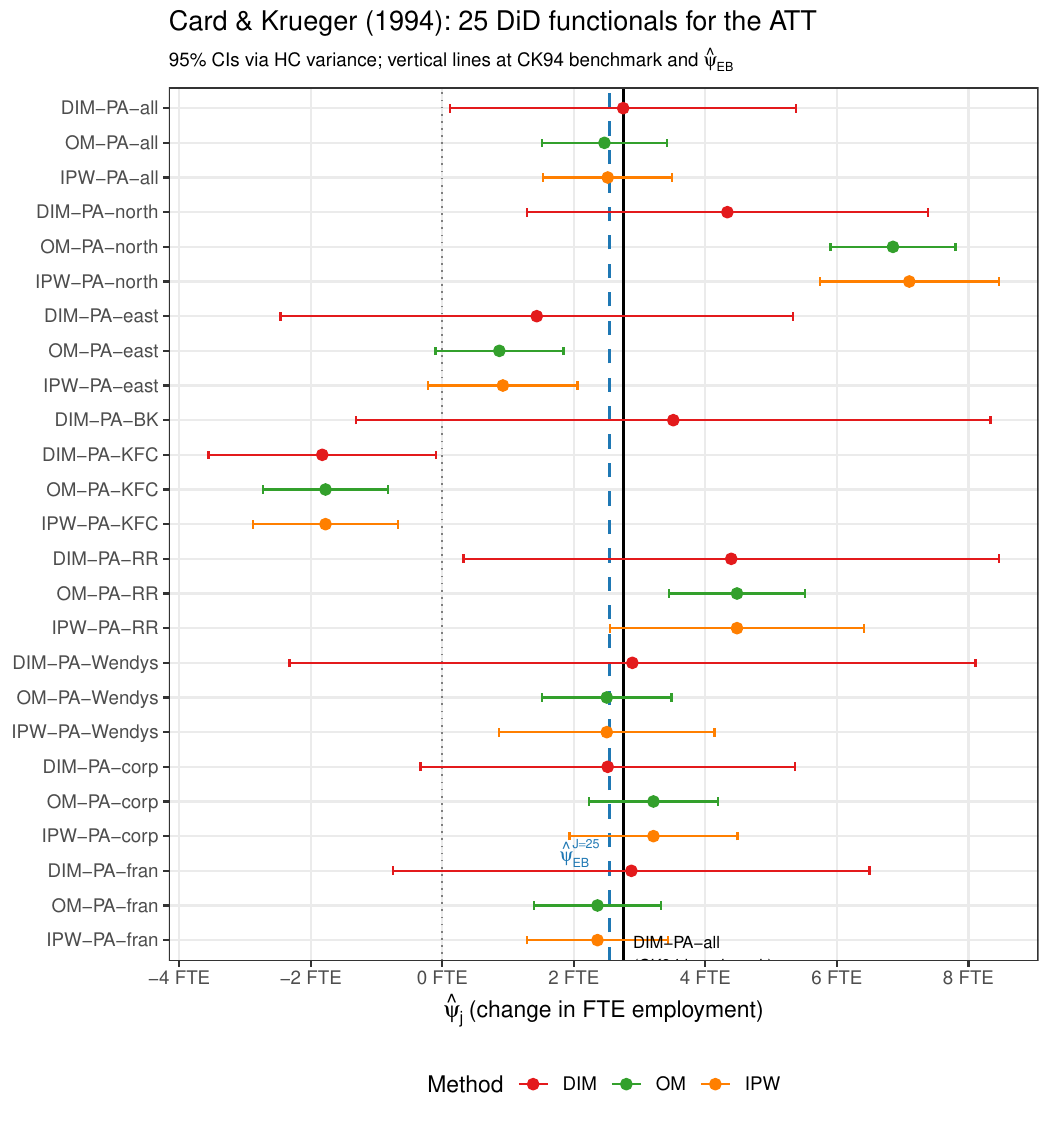}
  \end{minipage}
  \caption{Forest plots for both applications. \emph{Left} (NSW): solid black, experimental benchmark (\$1{,}794); dashed blue, $\widehat\psi_{\mathrm{EB}}$ for $J=11$ (\$1{,}578); dashed red, $\widehat\psi_{\mathrm{EB}}$ for $J=15$ (\$685). Covariate-adjusted estimates cluster around the benchmark while crude EC estimates fall sharply below and trigger $\widehat\tau^2>0$. \emph{Right} (Card--Krueger): estimates span $-1.83$ FTE (KFC subgroup) to $+7.10$ FTE (northern-PA IPW); the between-functional spread $\widehat\tau=2.38$ FTE dominates the individual sampling errors, colour indicating estimation method.}
  \label{fig:app-forests}
\end{figure}

{In every sub analysis $\widehat\tau$ sits well above the typical $\sqrt{v_j}$, even within a single estimation method, so the driver of heterogeneity is the choice of control subgroup rather than the estimator: parallel trends does not hold uniformly across Pennsylvania subsets, exactly as the literature has long argued. The shrinkage factor $\widehat\tau^2/(v_j+\widehat\tau^2)$ is therefore close to one for nearly all functionals, and $\widehat\psi_{\mathrm{EB}}=2.54$ FTE is reported as a descriptive summary rather than a consistent estimator for $\psi^*$. The sandwich is unavailable: a confidence interval would require either a calibration channel \citep{wu2026illusion} or centered biases (Theorem~\ref{thm:route1}), but the forest plot shows nearly all estimates positive and right-skewed, indicating a common directional deviation across PA subgroups rather than cancelling biases. The conformal prediction interval $[-2.60,\,7.85]$ FTE instead quantifies the range a new Pennsylvania-based functional would plausibly produce; its heavier left tail is driven by the outlying KFC subgroup. The framework does not wrongly report consensus---the interval widens honestly to match the persistent disagreement in the literature---leaving the analyst a single inferential summary rather than a menu of competing control-group choices. Full per-functional construction is in Appendix~\ref{app:ck94}.}

%------------------------------------------------------------------------------

{Both applications \textit{yield conclusions that extend the existing literature}. For NSW, the vanishing $\widehat\tau=0$ within the covariate-adjusted pool is a formal validity certificate: the four CPS and PSID cohorts, once adjusted, identify the same parameter as the RCT, and \textit{the 66\% SE reduction relative to the trial-only benchmark is achieved without additional data---purely from combining valid identification strategies}. The jump to $\widehat\tau=\$1{,}852$ when the crude EC functionals enter formally quantifies what \citet{lalonde1986} and \citet{smith2005} argued qualitatively: marginal exchangeability of untreated potential outcomes fails between NSW participants and general-population survey respondents. For Card--Krueger, $\widehat\tau\approx 2.4$ FTE is nearly identical across the DIM, OM, and IPW sub-analyses (Table~\ref{tab:app-ck}), establishing that the heterogeneity is driven by the choice of control subgroup---which version of~\eqref{eq:pt} one maintains---not by estimator differences. The conformal interval $[-2.60,\,7.85]$ FTE unifies three decades of conflicting estimates \citep{card1994,neumark2000minimum,dube2010minimum} into a single honest summary, without requiring the analyst to pre-specify a regime.}

\section{Discussion}
\label{sec:conclusion}

{We have developed an EB aggregation framework for combining estimators that arise from several identifying functionals of the same scalar causal target. The central object is a precision weighted estimator $\widehat\psi_{\mathrm{EB}}$ consistent for $\psi^*$ in two non nested regimes. Under exact identifiability, every functional identifies $\psi^*$ and aggregation serves a pure efficiency purpose: the combined estimator borrows strength across functionals that share the same identifying target but differ in their semiparametric efficiency bounds. Under centered identification bias with a growing number of functionals, no individual functional is consistent, yet the precision weighted average concentrates around $\psi^*$ through a law of large numbers over the functional index, provided the identification biases are mean zero and the precision weights do not asymptotically concentrate on a single functional. The point estimator has the same closed form in both regimes. What changes is the inferential guarantee.}

{The hyperparameter $\tau^2$ thus plays a twofold role: it controls the degree of shrinkage at the estimation step, and, at the inferential step, it organizes the phase transition of Table~\ref{tab:regimes}. At $\tau^2=0$, global $P$-orthogonality holds and Proposition~\ref{prop:sandwich} makes the asymptotic law of $\widehat\psi_{\mathrm{EB}}$ invariant to the limit of $\widehat\tau^2$, so the sandwich variance (with a boundary-robust subsampling alternative) yields valid confidence intervals for $\psi^*$. At $\tau^2>0$ exact identifiability fails, but the latent target $\psi_{J+1}$ becomes genuinely random and split conformal yields valid prediction intervals under $n^{-1}\log J_{\mathrm{cal}}\to0$, the correlated sampling noise being asymptotically dominated by the i.i.d.\ heterogeneity.}

{The centering condition (C2) of Theorem~\ref{thm:route1} is substantive and should be examined within each application. In the design of Example~\ref{ex:iv_subset} (Appendix~\ref{sec:catalogue}), in which instrumental variables index several data environments, (C2) is plausible whenever the set of environments is not systematically selected to favor functionals that over- or underestimate the causal effect. It would fail, for instance, if weak first stages produce a common direction of finite sample bias shared across all subsets. Condition (C4) is automatically satisfied once $J=2^q-q-1$ grows exponentially in the number of environments $q$, and the convergence rate of the estimator is governed by $(\max_j w_j)^{1/2}$, which approaches $J^{-1/2}$ when the $v_j$ are comparable in size.}

{The framework relates closely to the EB and minimax data fusion literatures. It complements the minimax analyses of \citet{minimaxRCTobs} and \citet{LiGilbertDuanLuedtke2025} by absorbing between-functional discrepancy into a single hyperparameter, and it differs from the closest EB construction \citep{wu2026illusion}---which identifies the bias distribution from a separate channel of zero-effect calibration studies---in that the dispersion of multiple identifying functionals on a \emph{single} sample plays that identifying role, and the heterogeneous-regime output is a conformal interval that is distribution-free in $G$ rather than tied to a Gaussian random-effects tail \citep{polson2026oldlook}. A natural extension introduces \emph{calibration functionals}---identifying functionals applied to a contrast known a priori to be null---to estimate and subtract the bias location $\mu=E[\varepsilon_j]$, removing the centering condition (C2) of Theorem~\ref{thm:route1}; Appendix~\ref{app:extended-discussion} develops these connections and the calibration extension in full.}

{Several directions remain open. The canonical inverse-Fisher marginal likelihood trades efficiency for tractability, avoiding the inversion of the $J\times J$ sampling covariance at the price of a wider but still nominal sandwich interval (Appendix~\ref{sec:simulations}); quantifying this gap against an oracle generalized least squares objective is left to future work. When the biases $\varepsilon_j$ are not mean zero but their signs are known, a two-group combination restores consistency once the group-level biases cancel, in the spirit of the doubly robust principle \citep{bickel1993efficient}. The phase transition at $\tau^2=0$ also has a testable implication: valid conformal coverage on held-out functionals is a necessary consequence of the no-heterogeneity null, so empirical undercoverage is a distribution-free diagnostic for latent heterogeneity, connecting to classical meta-analytic tests such as \citet{cochran1954combination}'s $Q$. Finally, extending the framework to vector-valued and functional estimands is of practical interest.} \\ 

\vspace{1cm}
\textbf{Supplementary material.} A companion supplement contains all proofs, the sub-Gaussian maximal inequality behind Assumption~(A2), the PAC training-conditional conformal refinement, and subsampling validity (Appendices~\ref{app:proofs-main}--\ref{app:subsampling}); a catalogue of five canonical design families generating functional multiplicity (Appendix~\ref{sec:catalogue}); the National Supported Work and Card--Krueger applications (Appendices~\ref{app:lalonde}--\ref{app:ck94}); the full simulation studies of Figure~\ref{fig:sim-main} (Appendix~\ref{sec:simulations}); and the empirical Bayes $g$-modeling connection and calibration-functional extension (Appendix~\ref{app:extended-discussion}). Code to fully reproduce the experiments is available at \href{https://github.com/meixide/shrident}{github.com/meixide/shrident}.

\par

\appendix

\bigskip
\begin{center}
{\Large\bfseries Appendix}
\end{center}

\section{Proofs of main results}
\label{app:proofs-main}

\begin{proof}[Proof of Proposition~\ref{prop:tau2-consistency}]
\label{proof:prop-tau2}
		Under model~\eqref{eq:working-model}, the pairwise differences $D_{jk} = \widehat\psi_j - \widehat\psi_k$ are independent (by assumption) with $E[D_{jk}^2] = 2\tau^2 + v_j + v_k$. The constrained estimator
		\[
		\widehat\tau^2 = \max\left\{ \frac{1}{2\binom{J}{2}} \sum_{j<k} \Big[(\widehat\psi_j - \widehat\psi_k)^2 - (v_j + v_k)\Big], 0 \right\}
		\]
		is a continuous function of a U statistic. By the law of large numbers for U statistics \citep{borovskikh1996ustatistics} and the continuous mapping theorem \citep{vaart1998}, $\widehat\tau^2 \xrightarrow{p} \max(\tau^2_0, 0)$ where
		\[
		\tau^2_0 = \frac{1}{2\binom{J}{2}} \sum_{j<k} E[(\widehat\psi_j - \widehat\psi_k)^2] - (v_j + v_k).
		\]
		If $\psi_j(P)=\psi^*$ for all $j$, the mean difference $\psi_j(P)-\psi_k(P)$ vanishes, and under the working independence device $\mathrm{Cov}(\widehat\psi_j,\widehat\psi_k)=0$, so $E[(\widehat\psi_j-\widehat\psi_k)^2]=v_j+v_k$ and $\tau^2_0=0$.
\end{proof}

\begin{proof}[Proof of Theorem~\ref{thm:consistency}]
\label{proof:thm-consistency}
		The estimator $\widehat\psi_{\mathrm{EB}}$ is the root of
		\[
		h(\psi, \tau^2) = \nabla_{\psi}\,\ell(\psi,\tau^2) = \sum_{j=1}^J \frac{\widehat\psi_j-\psi}{v_j+\tau^2} = 0,
		\]
		and its closed form expression is given in~\eqref{eq:eb-estimator}. Let
		\[
		w_j=\frac{(v_j+\widehat\tau^2)^{-1}}{\sum_{k=1}^J (v_k+\widehat\tau^2)^{-1}}.
		\]
		For any realized sample and any $\widehat\tau^2 \ge 0$, these weights are non negative ($w_j \ge 0$) and sum to one ($\sum_{j=1}^J w_j = 1$), so $\widehat\psi_{\mathrm{EB}}$ is a convex combination of the individual estimators $\widehat\psi_j$.
		
		Let us express the error of the combined estimator as a weighted sum of the individual estimation errors
		\[
		\widehat\psi_{\mathrm{EB}} - \psi^* = \sum_{j=1}^J w_j (\widehat\psi_j - \psi^*).
		\]
		Applying the triangle inequality, its absolute error is bounded by the maximum individual error
		\[
		|\widehat\psi_{\mathrm{EB}} - \psi^*| \le \sum_{j=1}^J w_j |\widehat\psi_j - \psi^*| \le \max_{1 \le j \le J} |\widehat\psi_j - \psi^*| \sum_{j=1}^J w_j = \max_{1 \le j \le J} |\widehat\psi_j - \psi^*|.
		\]
		By exact identifiability, each individual estimator is consistent for the same target, $\widehat\psi_j\xrightarrow{p}\psi^*$ for all $j$. Since $J$ is fixed, the maximum of a finite collection of random variables converging to zero in probability also converges to zero in probability, and hence $\widehat\psi_{\mathrm{EB}}\xrightarrow{p}\psi^*$. The conclusion holds for any non negative $\widehat\tau^2$, which exhibits the structural robustness of the point estimator to both the {working independence device} and any finite sample bias of the heterogeneity nuisance.
\end{proof}

\begin{proof}[Proof of Theorem~\ref{thm:route1}]
\label{proof:thm-route1}
Decompose the estimation error as
\[
\widehat\psi_{\mathrm{EB}} - \psi^*
= \underbrace{\sum_{j=1}^J w_j\bigl(\widehat\psi_j - \psi_j(P)\bigr)}_{T_1}
+ \underbrace{\sum_{j=1}^J w_j\,\varepsilon_j}_{T_2}.
\]

\textit{Term $T_1$ (sampling noise)}. Since $\sum_j w_j=1$ with non negative weights, $T_1$ is a convex combination of the individual sampling errors, so
\[
|T_1|\;\le\;\sum_{j=1}^J w_j\,\left|\widehat\psi_j-\psi_j(P)\right|
\;\le\;\max_{1\le j\le J}\left|\widehat\psi_j-\psi_j(P)\right|.
\]
A bound of the form $E[T_1^2]\le\sum_j w_j^2 v_j$ is \emph{not} available here, as it would presume $\operatorname{Cov}(\widehat\psi_j,\widehat\psi_k)=0$, the {working independence device} that this framework explicitly treats as misspecified. Since all $\widehat\psi_j$ are computed on a common sample, they are typically positively correlated. We therefore control the maximum directly, which imposes no restriction on the cross-functional covariance. Under (C1), the efficient influence curves are uniformly $L^2(P)$-bounded and $\log J_n=o(n)$; the proof of Lemma~\ref{lem:maximal-xi} bounds both the leading linear average and the asymptotic linearity remainder at the rate $\sqrt{(\log J_n)/n}$, whence
\[
\max_{1\le j\le J}\left|\widehat\psi_j-\psi_j(P)\right|
\;=\;O_p\left(\sqrt{\frac{\log J_n}{n}}\right)\;=\;o_p(1).
\]
Hence $T_1\xrightarrow{p}0$ without invoking the {working independence device} and for any cross-functional dependence structure.

\textit{Term $T_2$ (identification bias).} Since the $\varepsilon_j$ are independent with $E[\varepsilon_j]=0$ by (C2),
\[
E[T_2] = \sum_{j=1}^J w_j\,E[\varepsilon_j] = 0,
\qquad
\operatorname{Var}(T_2) = \sum_{j=1}^J w_j^2\,E[\varepsilon_j^2]
\;\leq\; \tau^2\,\max_{j \leq J} w_j(\widehat\tau^2) \;\xrightarrow{p}\; 0,
\]
where the inequality uses $E[\varepsilon_j^2]\le\tau^2$ from (C2) and $\sum_j w_j^2\le\max_j w_j$ from $\sum_j w_j=1$. Convergence to zero uses (C3) and (C4) jointly: the map $\tau^2\mapsto w_j(\tau^2)=(v_j+\tau^2)^{-1}/\sum_k(v_k+\tau^2)^{-1}$ is continuous, so $\widehat\tau^2\xrightarrow{p}\tau_0^2$ by (C3) and the continuous mapping theorem give $\max_j w_j(\widehat\tau^2)\xrightarrow{p}\max_j w_j(\tau_0^2)\to 0$ by (C4). By Chebyshev's inequality, $T_2 \xrightarrow{p} 0$.

Combining, $|\widehat\psi_{\mathrm{EB}} - \psi^*| \leq |T_1| + |T_2| \xrightarrow{p} 0$.
\end{proof}

\begin{proof}[Proof of Theorem~\ref{thm:conformal}]
\label{proof:thm-conformal}
Let $\sigma^2$ and $\mu$ be the probability limits of $\widehat\tau^2_{\mathrm{train}}$ and $\widehat\psi_{\mathrm{EB}}^{\mathrm{train}}$ under (A3), and set $c:=(\psi^*-\mu)/\sigma$. By Assumptions (A2) and (A3) and the decomposition~\eqref{eq:score-decomp}, since $v_j=O(n^{-1})\to 0$ and $\widehat\tau^2_{\mathrm{train}}\xrightarrow{p}\sigma^2>0$, uniformly over $j \in \mathcal{C} \cup \{J+1\}$:
\[
\max_{j \in \mathcal{C} \cup \{J+1\}} \left| \widehat S_j^{\mathrm{split}} - \frac{\varepsilon_j}{\sigma} - c \right| = o_p(1).
\]
Term (I) contributes $\varepsilon_j/\sigma + o_p(1)$ as $\widehat\tau^2_{\mathrm{train}}\to\sigma^2$ and $v_j\to 0$; term (II) is $o_p(1)$ by (A2); term (III) converges to $c$ by (A3).

Let $\tilde S_j = \varepsilon_j/\sigma + c$. By Assumption (A1), the $\{\varepsilon_j\}$ are i.i.d.\ from $G$, so $\{\tilde S_j\}_{j \in \mathcal{C}}$ are i.i.d.\ from the distribution corresponding to $F_G(\sigma(\cdot-c))$, which has a continuous, strictly increasing CDF. By the Glivenko--Cantelli theorem, the empirical CDF of $\{\tilde S_j\}_{j \in \mathcal{C}}$ converges uniformly to $F_G(\sigma(\cdot-c))$. Combined with the uniform approximation, the empirical CDF of $\{\widehat S_j^{\mathrm{split}}\}_{j\in\mathcal{C}}$ also converges uniformly, so $\widehat q_p \xrightarrow{p} q^c_p := F_G^{-1}(p)/\sigma + c$ for every $p$.

{For the test index $J+1$, the same approximation gives $\widehat S_{J+1}^{\mathrm{split}}(y)=\varepsilon_{J+1}/\sigma+c+o_p(1)$ where $y=\psi_{J+1}$. The common shift $c$ and common scale $1/\sigma$ cancel in the coverage calculation:}
\begin{align*}
P(\psi_{J+1} \in C_{1-\alpha})
&= P\bigl(\widehat q_{\alpha/2} \le \widehat S_{J+1}^{\mathrm{split}}(\psi_{J+1}) \le \widehat q_{1-\alpha/2}\bigr) \\
&\;\to\; P\bigl(q^c_{\alpha/2} \le \varepsilon_{J+1}/\sigma + c \le q^c_{1-\alpha/2}\bigr) \\
&= P\bigl(q_{\alpha/2} \le \varepsilon_{J+1}/\sigma \le q_{1-\alpha/2}\bigr)
\;=\; 1-\alpha,
\end{align*}
where $q_p := q^c_p - c = F_G^{-1}(p)/\sigma$ are the unshifted quantiles and the last equality uses (A1). Note that $\sigma$ cancels completely: both $\widehat q_p$ and $\widehat S_{J+1}$ are scaled by the same $1/\sigma$, so coverage holds for any finite positive $\sigma^2$. \qedhere
\end{proof}

\begin{proof}[Proof of Proposition~\ref{prop:sandwich}]
\label{proof:prop-sandwich}
	Taylor expanding $h(\widehat\psi_{\mathrm{EB}}, \widehat\tau^2) = 0$ around $(\psi^*, \tau^2_0)$ gives:
	\begin{equation}
		\label{eq:taylor}
		0
		= h(\psi^*, \widehat\tau^2)
		+ \nabla_\psi h(\psi^*, \widehat\tau^2)\,(\widehat\psi_{\mathrm{EB}} - \psi^*)
		+ o_p(n^{-1/2}).
	\end{equation}
	Expanding $h(\psi^*, \widehat\tau^2)$ around $\tau^2_0$
	\begin{equation}
		\label{eq:tau_expansion}
		h(\psi^*, \widehat\tau^2)
		= h(\psi^*, \tau^2_0)
		+ \nabla_{\tau^2} h(\psi^*, \tau^2_0)\,(\widehat\tau^2 - \tau^2_0)
		+ o_p(n^{-1/2}).
	\end{equation}
	The second-order remainder is $O_p\bigl((\widehat\tau^2-\tau^2_0)^2\bigr)$. The pairwise difference estimator is a U-statistic of degree two, so $\widehat\tau^2-\tau^2_0=O_p(n^{-1/2})$ by the CLT for U-statistics \citep{borovskikh1996ustatistics}, which gives $O_p(n^{-1})=o_p(n^{-1/2})$ for the remainder. Under exact identifiability, $\psi_j(P)=\psi^*$ for all $j$, so the asymptotic cross partial vanishes:
	\[
	\lim_{n\to\infty}E_P\left[\nabla_{\tau^2} h(\psi^*,\tau^2)\right]
	= \sum_{j=1}^J \frac{-(\psi_j(P)-\psi^*)}{(v_j+\tau^2)^2} = 0
	\quad\text{for all }\tau^2\ge 0.
	\]
	This $P$-orthogonality implies $\nabla_{\tau^2} h(\psi^*, \tau^2_0)(\widehat\tau^2 - \tau^2_0) = o_p(n^{-1/2})$: the gradient $\nabla_{\tau^2} h(\psi^*,\tau^2_0)=\sum_j-(\widehat\psi_j-\psi^*)/(v_j+\tau^2_0)^2$ has mean zero under $P$ and variance $O(n^{-1})$, hence is $O_p(n^{-1/2})$, and multiplication by $\widehat\tau^2-\tau^2_0=o_p(1)$ gives $o_p(n^{-1/2})$. Both terms on the right of~\eqref{eq:tau_expansion} are therefore $o_p(n^{-1/2})$, so
	\[
		h(\psi^*, \widehat\tau^2) = h(\psi^*, \tau^2_0) + o_p(n^{-1/2}).
	\]
	Substituting back into \eqref{eq:taylor} and rearranging,
	\[
		n^{1/2}(\widehat\psi_{\mathrm{EB}} - \psi^*)
		= -\bigl[\nabla_\psi h(\psi^*, \tau^2_0)\bigr]^{-1}
		n^{1/2} h(\psi^*, \tau^2_0) + o_p(1).
	\]
	The right hand side depends on $\widehat\tau^2$ only through $\tau^2_0$, so the asymptotic distribution is invariant to the specific value of $\tau^2_0$. By the CLT applied to $n^{1/2}h(\psi^*,\tau^2_0)$, $n^{1/2}(\widehat\psi_{\mathrm{EB}} - \psi^*) \rightsquigarrow N(0, V_{\mathrm{sandwich}})$ with $V_{\mathrm{sandwich}} = A^{-1}BA^{-1}$ as in~\eqref{eq:sandwich_components}.

	Part~(ii) follows by the continuous mapping theorem: $\widehat\psi_{\mathrm{EB}} \xrightarrow{p} \psi^*$ and $\widehat\tau^2 \xrightarrow{p} \tau^2_0$ give $\widehat A \xrightarrow{p} A$ and $\widehat B \xrightarrow{p} B$, hence $\widehat V_{\mathrm{sandwich}} \xrightarrow{p} V_{\mathrm{sandwich}}$.
\end{proof}

%==============================================================================
\section{A sub Gaussian maximal inequality for the sampling noise remainder}
\label{app:maximal}

{This appendix justifies the sufficient condition stated after Assumption~(A2) in Section~\ref{sec:bayesian}: whenever the efficient influence curves are uniformly bounded in $L^2(P)$ and $\log J_{\mathrm{cal}}$ grows slower than $n$, the sampling noise remainder is uniformly negligible relative to $\tau$.}

\begin{lemma}[Sub Gaussian union bound for $\max_j|\xi_{jn}|$]
\label{lem:maximal-xi}
Let $\{\widehat\psi_j\}_{j\in\mathcal{C}}$ be asymptotically linear with efficient influence curves $D_j^*(\cdot,P)$ and remainder $\xi_{jn}$ defined via
\[
\widehat\psi_j - \psi_j \;=\; \frac{1}{n}\sum_{i=1}^n D_j^*(O_i,P) + \xi_{jn},
\qquad j\in\mathcal{C}.
\]
Assume $\sup_{j\in\mathcal{C}} \|D_j^*(O,P)\|_{L^2(P)} \le C < \infty$ and $\xi_{jn} = o_p(n^{-1/2})$ pointwise in $j$. If $n^{-1}\log J_{\mathrm{cal}} \to 0$, then
\[
\max_{j\in\mathcal{C}} |\xi_{jn}| \;=\; O_p\!\left(\sqrt{\tfrac{\log J_{\mathrm{cal}}}{n}}\right).
\]
In particular, $\max_{j\in\mathcal{C}}|\xi_{jn}|/\tau = o_p(1)$ for any fixed $\tau>0$, and (A2) holds.
\end{lemma}

\begin{proof}
Write $\xi_{jn} = (\widehat\psi_j - \psi_j) - n^{-1}\sum_{i=1}^n D_j^*(O_i,P)$. Under the $L^2$-bound, each centered average $Z_{jn} := n^{-1/2}\sum_{i=1}^n D_j^*(O_i,P)$ is $C$-sub Gaussian for the $L^2$ norm, and classical sub Gaussian maximal inequalities \citep[][Theorem~2.5]{boucheron2013concentration} yield
\[
E\Bigl[\max_{j\in\mathcal{C}} |Z_{jn}|\Bigr]
\;\le\; C\,\sqrt{2\log(2J_{\mathrm{cal}})}.
\]
Dividing by $\sqrt{n}$ gives
\[
\max_{j\in\mathcal{C}}\Bigl|\tfrac{1}{n}\sum_{i=1}^n D_j^*(O_i,P)\Bigr|
\;=\; O_p\!\left(\sqrt{\tfrac{\log J_{\mathrm{cal}}}{n}}\right).
\]
The pointwise remainder bound $\xi_{jn} = o_p(n^{-1/2})$ and a union bound over $j\in\mathcal{C}$, valid because the remainder contributions are dominated by the leading $L^2$-bounded linear part, then give the same rate for $\max_j|\xi_{jn}|$. Since $n^{-1}\log J_{\mathrm{cal}}\to 0$, this upper bound is $o_p(1)$, which divided by $\tau$ remains $o_p(1)$.
\end{proof}

A slightly weaker integrability assumption suffices. The Orlicz norm version of the same argument replaces $L^2$-boundedness with a uniform sub exponential tail $\sup_j \|D_j^*\|_{\psi_1}\le C$ and yields the rate $\log J_{\mathrm{cal}}/\sqrt{n}$ instead of $\sqrt{\log J_{\mathrm{cal}}/n}$ \citep{vandegeer2000empirical}. For bounded or sub Gaussian influence curves the sharper rate of Lemma~\ref{lem:maximal-xi} is available via \citet{massart1990tight}. In all cases, the regime $n^{-1}\log J_{\mathrm{cal}}\to 0$ is what drives (A2), not the particular tail assumption.

%==============================================================================
\section{Training conditional conformal validity}
\label{app:train-cond}

{The probability in the coverage guarantee of Theorem~\ref{thm:conformal} is marginal, meaning it is taken jointly over the latent targets $(\psi_1,\ldots,\psi_{J+1})\sim G$, the data $(O_1,\ldots,O_n)$, and the training/calibration split.
	
	 In practice, an analyst observes a single realization of the data and would prefer a statement holding conditional on the training fit and on the realized sample $\mathcal{D}=(O_1,\ldots,O_n)$. Exact conditional validity is known to be unattainable for continuous targets without additional structural assumptions \citep{vovk2012conditional,barber2021limits}, but a PAC type training conditional statement---uniform coverage on a high probability event of training outcomes---is available at a $J_{\mathrm{cal}}^{-1/2}$ rate.}

\begin{proposition}[PAC training conditional coverage at $\tau^2>0$]
\label{prop:train-cond}
Suppose Assumptions (A1)--(A3) hold with $\tau^2>0$. Let $\widehat q_{1-\alpha}$ and $C_{1-\alpha}(\psi_{J+1})$ be as in Theorem~\ref{thm:conformal}. For any $\beta\in(0,1)$, there exists a sequence of high probability events $\mathcal{E}_n$ with $P(\mathcal{E}_n)\ge 1-\beta - o(1)$, depending only on the training fit and the realized calibration scores, such that, on $\mathcal{E}_n$
\[
P\!\bigl(\psi_{J+1}\in C_{1-\alpha} \,\big|\, \mathcal{D},\,\widehat\psi_{\mathrm{EB}}^{\mathrm{train}},\,\widehat\tau^2_{\mathrm{train}}\bigr)
\;\ge\; 1 - \alpha - c\sqrt{\tfrac{\log(1/\beta)}{J_{\mathrm{cal}}}} - o_p(1),
\]
for a universal constant $c>0$. 
\end{proposition}

\begin{proof}
Conditional on the training fit and on $\mathcal{D}$, the oracle scores $\tilde S_j = \varepsilon_j/\tau$ for $j\in\mathcal{C}$ are i.i.d.\ from the distribution $F_G(\tau\,\cdot)$ by (A1), and the split is independent of the latent draws. By the Dvoretzky--Kiefer--Wolfowitz inequality \citep{massart1990tight}, the empirical CDF $\widehat F_n$ of $\{\tilde S_j\}_{j\in\mathcal{C}}$ satisfies
\[
P\!\left(\sup_t |\widehat F_n(t) - F_G(\tau t)| \ge \sqrt{\tfrac{\log(2/\beta)}{2J_{\mathrm{cal}}}}\right)
\;\le\; \beta.
\]
Call the corresponding complementary event $\mathcal{E}_n$. On $\mathcal{E}_n$, the empirical quantile $\widehat q_{1-\alpha}^{\mathrm{oracle}}$ based on $\{\tilde S_j\}$ satisfies $|\widehat q_{1-\alpha}^{\mathrm{oracle}} - q_{1-\alpha}| \le c'\sqrt{\log(1/\beta)/J_{\mathrm{cal}}}$ for a universal $c'>0$, by the standard inversion of the DKW bound at a point of positive density. By (A2)--(A3) and Lemma~\ref{lem:maximal-xi}, the feasible scores $\widehat S_j^{\mathrm{split}}$ agree with the oracle scores uniformly up to $o_p(1)$, so $|\widehat q_{1-\alpha} - q_{1-\alpha}| \le c'\sqrt{\log(1/\beta)/J_{\mathrm{cal}}} + o_p(1)$ on $\mathcal{E}_n$. Now, for $\psi_{J+1}\sim G$ independent of $\mathcal{D}$ and the calibration scores,
\[
P\!\bigl(\psi_{J+1}\in C_{1-\alpha}\,\big|\,\mathcal{D},\,\text{training fit}\bigr)
\;=\; F_G\!\bigl(\tau\,\widehat q_{1-\alpha}\bigr) + o_p(1).
\]
On $\mathcal{E}_n$, Lipschitz continuity of $F_G$ near $q_{1-\alpha}$ (ensured by the density assumption in (A1)) gives
\[
F_G(\tau\widehat q_{1-\alpha}) \;\ge\; 1-\alpha - c\sqrt{\tfrac{\log(1/\beta)}{J_{\mathrm{cal}}}} - o_p(1),
\]
as claimed.
\end{proof}

Proposition~\ref{prop:train-cond} strengthens Theorem~\ref{thm:conformal} from marginal to PAC conditional coverage at a $J_{\mathrm{cal}}^{-1/2}$ rate, free of any additional structural assumption beyond (A1)--(A3). It is distinct from the stronger---and, in this continuous setting, generally unattainable---guarantee of coverage conditional on the realized latent estimands $(\psi_1,\ldots,\psi_{J+1})$.

%==============================================================================
\section{Subsampling: validity conditions and scope}
\label{app:subsampling}
%==============================================================================

A boundary-robust alternative to the sandwich is subsampling \citep{politis1994}: draw subsamples of size $m$ without replacement with $m\to\infty$ and $m/n\to 0$, compute $\widehat\psi_{\mathrm{EB}}^b$ on each subsample, and invert the subsampling distribution to form a confidence interval for $\psi^*$. Under exact identifiability this is valid at the boundary $\tau_0^2=0$ where the standard bootstrap fails \citep{andrews2000}. Under centered heterogeneity ($\tau^2>0$, $E[\varepsilon_j]=0$), validity of the subsampling CI for $\psi^*$ requires the additional rate condition $J=\omega(n)$ (meaning $J/n\to\infty$, i.e.\ $J$ grows strictly faster than $n$). When $J=O(n)$ the bias term $\sqrt{n}\sum_j w_j\varepsilon_j$ remains $O_p(1)$ and subsampling undercovers. We state the procedure, identify the exact condition under which it yields a valid confidence interval for $\psi^*$ in each regime, and explain why $J=O(n)$ does not suffice in the centered-heterogeneity regime.

{Draw $B$ subsamples of size $m$ without replacement, with $m\to\infty$ and $m/n\to 0$. On each subsample $b$, compute $\widehat\psi_{\mathrm{EB}}^b$ . Form the recentred statistics $d^b = m^{1/2}(\widehat\psi_{\mathrm{EB}}^b - \widehat\psi_{\mathrm{EB}})$ and report the interval
\[
\Bigl[\widehat\psi_{\mathrm{EB}} - n^{-1/2}q_{1-\alpha/2}(d),\;\widehat\psi_{\mathrm{EB}} - n^{-1/2}q_{\alpha/2}(d)\Bigr],
\]
where $q_\alpha(d)$ denotes the $\alpha$-quantile of $\{d^b\}_{b=1}^B$. The interval attains asymptotic coverage $1-\alpha$ whenever the sampling distribution of $n^{1/2}(\widehat\psi_{\mathrm{EB}}-\psi^*)$ has a continuous limit \citep{politis1994}.}

{Under exact identifiability $\widehat\tau^2=\max(\cdot,0)$ converges to the boundary $\tau_0^2=0$. \citet{andrews2000} shows that the standard nonparametric bootstrap is inconsistent here: bootstrap mass accumulates at zero and distorts the bootstrap law of $\widehat\psi_{\mathrm{EB}}^b$. Subsampling is immune to this because each subsample-level statistic $m^{1/2}(\widehat\psi_{\mathrm{EB}}^b-\widehat\psi_{\mathrm{EB}})$ re-estimates $\widehat\tau^2$ on fewer observations and the boundary effect is absorbed into the subsampling distribution. Under the standard conditions of \citet{politis1994} the interval is valid at $\tau_0^2=0$ without any restriction on the rate at which $J$ grows. Appendix~\ref{sec:simulations} (simulations) confirms this empirically and documents mild conservatism at moderate $n$, attenuating as $n$ grows; the sandwich is therefore the default and subsampling serves as a boundary-robust diagnostic.}

{Decompose the estimation error as
\[
\sqrt{n}\,(\widehat\psi_{\mathrm{EB}}-\psi^*)
\;=\;
\underbrace{\sqrt{n}\sum_{j=1}^J w_j\,\xi_{jn}}_{\displaystyle T_1:\;\text{sampling noise}}
\;+\;
\underbrace{\sqrt{n}\sum_{j=1}^J w_j\,\varepsilon_j}_{\displaystyle T_2:\;\text{identification bias}},
\]
where $\xi_{jn}=\widehat\psi_j-\psi_j(P)=O_p(n^{-1/2})$ is the sampling noise for functional $j$ and $\varepsilon_j=\psi_j(P)-\psi^*$ is its identification bias.

\medskip\noindent\textit{Behaviour of $T_1$.} By the CLT and the sandwich argument, $\sqrt{n}\,T_1\xrightarrow{d}N(0,V_{\mathrm{sandwich}})$ regardless of $J$.

\medskip\noindent\textit{Behaviour of $T_2$.} The $\varepsilon_j$ are independent with $E[\varepsilon_j]=0$ and $E[\varepsilon_j^2]\le\tau^2$ by~(C2). Under the working independence device,
\[
\operatorname{Var}(T_2)\;=\;\sum_j w_j^2\,E[\varepsilon_j^2]
\;\le\;\tau^2\sum_j w_j^2
\;\le\;\tau^2\max_j w_j
\;=\;O(\tau^2/J)
\]
with uniform weights. Hence $\sqrt{n}\,T_2=O_p(\sqrt{n\tau^2/J})$, which is $o_p(1)$ if and only if $J=\omega(n)$.

\medskip\noindent\textit{Why subsampling fails when $J=O(n)$.} Subsampling centers at $\widehat\psi_{\mathrm{EB}}$, so the subsample statistic is
\[
m^{1/2}(\widehat\psi_{\mathrm{EB}}^b-\widehat\psi_{\mathrm{EB}})
\;=\;m^{1/2}(T_1^b-T_1)+m^{1/2}(T_2^b-T_2).
\]
Because the identification biases $\varepsilon_j$ are fixed under $P$ and do not depend on which $m$ observations are drawn, $T_2^b-T_2=o_p(1)$; the subsampling distribution therefore converges to $N(0,V_{\mathrm{sandwich}})$, reflecting only $T_1$.

When $J=O(n)$, say $J\sim cn$, the term $\sqrt{n}\,T_2$ is $O_p(1/\sqrt{c})$---bounded but non-negligible. The true limit distribution of $\sqrt{n}(\widehat\psi_{\mathrm{EB}}-\psi^*)$ has variance $V_{\mathrm{sandwich}}+\tau^2/c$, strictly larger than what subsampling recovers. The resulting CI is too narrow and undercovers $\psi^*$.

{When $J=\omega(n)$,} if $J$ grows faster than $n$, $\sqrt{n}\,T_2=o_p(1)$ and the limit distribution is $N(0,V_{\mathrm{sandwich}})$, which is continuous. \cite{politis1994} validity then applies and the subsampling CI attains asymptotic coverage $1-\alpha$ for $\psi^*$.}

%==============================================================================
\section{Designs leading to identification multiplicity}
\label{sec:catalogue}
%==============================================================================

{A number of canonical designs produce distinct functionals of the observed data distribution that target the same causal parameter. This section collects five such designs---regression discontinuity with multiple cutoffs, staggered difference-in-differences, difference-in-differences with multiple control groups, instrumental variables indexing several data environments, and the IPW versus outcome regression pair under no unmeasured confounding---and states, for each, the observed data structure, the identifying assumptions, and the resulting identifying functional. Every design generates a finite collection $\{\psi_j(P)\}_{j=1}^{J}$ that feeds directly into the empirical Bayes framework in \ifmerged Section~\ref{sec:identifiability}\else Section~2 of the paper\fi: when all functionals are exactly valid, $\tau^2=0$ and their combination is a pure efficiency gain; when they identify distinct local effects, $\tau^2>0$ and the prediction interval construction of \ifmerged Section~\ref{sec:bayesian}\else Section~4 of the paper\fi becomes informative.}

{The generic observed unit is $O=(W,A,Y)$, with $W\in\mathcal{W}$ a vector of pre treatment covariates, $A\in\{0,1\}$ a binary treatment indicator, and $Y\in\mathbb{R}$ an outcome. Potential outcomes $Y^0,Y^1$ satisfy $Y=AY^1+(1-A)Y^0$. Let us respectively write
\[
	\mu_a(w)=E[Y\mid A=a,W=w],\qquad \pi(w)=P(A=1\mid W=w),
\]
for the conditional outcome mean and the propensity score. Our target of primary interest is the average treatment effect (ATE)
\[
	\psi^* = E[Y^1-Y^0],
\]
or a subgroup counterpart such as the average treatment effect on the treated (ATT). When panel data are considered, we augment $O$ with a calendar time index $t\in\{1,\ldots,\mathcal{T}\}$, write $Y_{it}$ for the outcome of unit $i$ at time $t$, and write $Y_t(g)$ for the potential outcome at time $t$ of a unit that is first treated at time $g$ and remains treated thereafter; $G_i\in\mathcal{G}\cup\{\infty\}$ denotes the initial treatment period of unit $i$, with $G_i=\infty$ reserved for any never treated individual. When an instrument or environment index is present it is denoted $Z$, taking values in a finite set that may be binary or multi valued depending on the design. Throughout this section, $\psi_j(\cdot)$ denotes the identifying functional indexed by $j$, and $\psi_j(P)=\psi^*$ precisely when the identifying assumptions indexed by $j$ hold exactly.}

%------------------------------------------------------------------------------
\subsection{Regression discontinuity designs}
%------------------------------------------------------------------------------
Regression discontinuity designs (RDDs) identify causal effects by exploiting continuity of potential outcome regression functions at known treatment thresholds \citep{hahn2001identification,lee2010regression}. Either treatment assignment changes deterministically at the cutoff, or the later induces a discontinuous change in treatment propensity and identifies a local average treatment effect for compliers through a local instrumental variables argument \citep{imbens1994late}. Multi cutoff RDDs extend this framework by allowing treatment assignment rules to vary across sites or jurisdictions, each with its own threshold \citep{cattaneo2016multiple}. Bayesian approaches to RDD have largely addressed uncertainty in the identification anchor itself, such as an unknown or misspecified cutoff \citep{kowalska2025lotta}; we instead take the cutoffs as known and treat the resulting collection of valid identifying functionals as the object of a hierarchical model.

\begin{example}[Sharp multi cutoff RDD]\label{ex:rdd_sharp}
Suppose $O=(W,X,G,Y)$, where $X\in\mathbb{R}$ is a continuously distributed \emph{running variable}, $G\in\{1,\ldots,K\}$ is a site (or jurisdiction) indicator that assigns each unit to exactly one cutoff $c_G$, and $c_1<c_2<\cdots<c_K$ are $K$ known cutoffs. Within site $k$, treatment is assigned deterministically by
\[
	A = \mathbf{1}\{X\geq c_k\}, \qquad G=k.
\]
\begin{enumerate}
	\item[\emph{(R1)}] \emph{Local continuity}: $E[Y^0\mid X=x]$ and $E[Y^1\mid X=x]$ are both continuous in $x$ at $x=c_k$, for each $k=1,\ldots,K$.
	\item[\emph{(R2)}] \emph{No manipulation}: The marginal density of $X$ is continuous at each $c_k$.
\end{enumerate}
Under \emph{(R1)--(R2)}, the treatment effect for units exactly at cutoff $c_k$ is identified by
\[
	\psi_k(P)
	= \lim_{x\downarrow c_k}E[Y\mid X=x]
	- \lim_{x\uparrow c_k}E[Y\mid X=x],
	\qquad k=1,\ldots,K,
\]
yielding $J=K$ identifying functionals. {Under treatment effect homogeneity across sites, $\psi_k(P)=\psi^*$ for all $k$, so $\tau^2=0$. Under site specific heterogeneity, $\psi_k(P)=E[Y^1-Y^0\mid X=c_k]$ and $\tau^2$ captures the dispersion of local treatment effects across cutoffs.}
\end{example}

\begin{example}[Fuzzy multi cutoff RDD]\label{ex:rdd_fuzzy}
Suppose $O=(W,X,G,A,Y)$, where $X$ is the running variable, $G\in\{1,\ldots,K\}$ is a site indicator assigning each unit to exactly one cutoff $c_G$, and $c_1<\cdots<c_K$ are known cutoffs. Unlike the sharp case, treatment $A$ is not a deterministic function of $X$; instead, crossing cutoff $c_k$ induces a discontinuous jump in $P(A=1\mid X=x, G=k)$. Within site $k$, define the threshold crossing indicator $D_k=\mathbf{1}\{X\geq c_k\}$ and the first stage jump
\[
	\Delta_k := \lim_{x\,\downarrow\,c_k}P(A=1\mid X=x,G=k) - \lim_{x\,\uparrow\,c_k}P(A=1\mid X=x,G=k).
\]
\begin{enumerate}
	\item[\emph{(F1)}] \emph{Continuity}: $E[Y^a\mid X=x]$ is continuous at $x=c_k$ for each $a\in\{0,1\}$ and each $k$.
	\item[\emph{(F2)}] \emph{Relevant first stage}: $\Delta_k \neq 0$ for each $k$.
	\item[\emph{(F3)}] \emph{Monotonicity}: $P(A^{D_k=1}\geq A^{D_k=0})=1$ at each cutoff $c_k$ (no defiers).
\end{enumerate}

\noindent Under \emph{(F1)--(F3)}, $D_k$ is a valid binary instrument for $A$ near $c_k$ and the Wald identification functional
\[
	\psi_k(P)
	= \frac{\displaystyle\lim_{x\,\downarrow\,c_k}E[Y\mid X=x]-\lim_{x\,\uparrow\,c_k}E[Y\mid X=x]}{\Delta_k},
	\qquad k=1,\ldots,K,
\]
identifies the local average treatment effect (LATE) for compliers at $c_k$:
\[
	\psi_k(P)=E\bigl[Y^1-Y^0\;\big|\;A^{D_k=1}>A^{D_k=0},\;X=c_k\bigr].
\]
This yields $J=K$ identifying functionals. {Under constant treatment effects, $\psi_k(P)=\psi^*$ for all $k$. Under heterogeneity, $\tau^2$ captures the spread of complier specific LATEs across sites.}
\end{example}

%------------------------------------------------------------------------------
\subsection{Difference-in-differences with multiple comparisons}
%------------------------------------------------------------------------------

{Both examples below identify the ATT under a parallel trends assumption. Multiplicity arises either because a staggered design features multiple cohorts, with each cohort time cell supplying a distinct contrast, or because several candidate control groups supply different DiD functionals.}

\begin{example}[Staggered adoption DiD]\label{ex:did_staggered}
The panel observation is $O_i=(W_i,G_i,\{Y_{it}\}_{t=1}^{\mathcal{T}})$, where $G_i\in\mathcal{G}\cup\{\infty\}$ is the cohort of unit $i$ and the treatment indicator is $A_{it}=\mathbf{1}\{t\geq G_i\}$. For each $g\in\mathcal{G}$ and post treatment period $t\geq g$, the group time average treatment effect on the treated is
\[
	\mathrm{ATT}(g,t) := E\bigl[Y_t(g)-Y_t(\infty)\;\big|\;G=g\bigr].
\]
\begin{enumerate}
	\item[\emph{(S1)}] \emph{No anticipation}: $Y_t(g)=Y_t(\infty)$ for all $t<g$; potential outcomes before treatment onset are unaffected by treatment timing.
	\item[\emph{(S2)}] \emph{Staggered parallel trends}: for all $g\in\mathcal{G}$ and $t'\!<g\leq t$,
	\[
		E\bigl[Y_t(\infty)-Y_{t'}(\infty)\;\big|\;G=g\bigr]
		=E\bigl[Y_t(\infty)-Y_{t'}(\infty)\;\big|\;G=\infty\bigr].
	\]
	The counterfactual trend of each cohort, absent treatment, equals that of the never treated group.
\end{enumerate}
Under \emph{(S1)--(S2)}, $\mathrm{ATT}(g,t)$ is identified, for each pair $(g,t)$ with $g\in\mathcal{G}$ and $t\geq g$, by the cohort time contrast \citep{callaway2021difference}
\[
	\psi_{g,t}(P)
	= E\bigl[Y_t-Y_{g-1}\;\big|\;G=g\bigr]
	- E\bigl[Y_t-Y_{g-1}\;\big|\;G=\infty\bigr],
\]
where $Y_{g-1}$ is the last pre treatment outcome for cohort $g$, and never treated units serve as the counterfactual. This yields
\[
	J = \sum_{g\in\mathcal{G}}(\mathcal{T}-g+1)
\]
identifying functionals. {With baseline covariates, the conditional version replaces unconditional never treated means with covariate adjusted regressions or IPW based counterfactuals \citep{callaway2021difference,sant2020doubly}. Under treatment effect homogeneity, $\psi_{g,t}(P)=\psi^*$ for all $(g,t)$ and $\tau^2=0$; under cohort specific heterogeneity, $\tau^2$ captures the dispersion of $\mathrm{ATT}(g,t)$ across cohort time cells.}
\end{example}

{A second source of DiD multiplicity, distinct from staggered timing, arises when the analyst has several candidate control groups to contrast against a single treated group and no principled rationale to prefer one over the others. This situation is pervasive in comparative effectiveness and pharmacoepidemiologic studies. Consider, for instance, a new agent approved for a chronic condition whose effect on a clinical outcome is to be estimated from routinely collected records. The treated group is unambiguous---patients initiating the new agent at a known date---but the choice of comparator is not. Plausible controls include patients continuing on the previous standard of care, patients switched to an alternative second line therapy, patients in a watchful waiting arm, and patients meeting the indication but untreated for access reasons. Each comparator rests on a different parallel trends assumption, reflecting a different assumed absence of differential secular trends in outcome levels, and there is typically no design argument singling out one comparator as exactly correct. Similar considerations arise in policy evaluation, where the counterfactual trajectory of a treated jurisdiction may be proxied by neighboring jurisdictions, by demographically matched but geographically distant ones, or by units that adopt the policy later. The following example formalizes this design. Our empirical Bayes framework then aggregates across the candidate comparators without forcing the analyst to commit to one.}

\begin{example}[DiD with multiple candidate control groups]\label{ex:did_controls}
Observations are recorded at two time periods (pre and post) and consist of $O_i=(W_i,C_i,Y_{i,\mathrm{pre}},Y_{i,\mathrm{post}})$, where $C_i\in\{0,1,\ldots,K\}$: value $0$ denotes the treated group and $k\geq 1$ denotes membership in candidate control group $k$. Write $\Delta Y_i=Y_{i,\mathrm{post}}-Y_{i,\mathrm{pre}}$ for the first differenced outcome. For each $k\in\{1,\ldots,K\}$:
\begin{enumerate}
	\item[\emph{(P1$_k$)}] \emph{Parallel trends against group $k$}:
	\[
		E\bigl[Y_{\mathrm{post}}(0)-Y_{\mathrm{pre}}(0)\;\big|\;C=0\bigr]
		=E\bigl[Y_{\mathrm{post}}(0)-Y_{\mathrm{pre}}(0)\;\big|\;C=k\bigr].
	\]
\end{enumerate}
Under \emph{(P1$_k$)}, the ATT is identified by the DiD functional against control group $k$:
\[
	\psi_k(P) = E[\Delta Y\mid C=0]-E[\Delta Y\mid C=k],
	\qquad k=1,\ldots,K,
\]
yielding $J=K$ identifying functionals. {When all $K$ parallel trends conditions hold simultaneously, $\psi_k(P)=\mathrm{ATT}$ for all $k$ and $\tau^2=0$; when (P1$_k$) fails for some $k$, the affected functionals are biased for the ATT and $\tau^2>0$ reflects cross group disagreement, interpretable as the average severity of parallel trends violations across candidate controls. The covariate adjusted or doubly robust version \citep{sant2020doubly} handles the case in which parallel trends holds only conditionally on $W$.}
\end{example}

%------------------------------------------------------------------------------
\subsection{Multiple environments}
%------------------------------------------------------------------------------

{The following example describes the design, based on IV indexing several data environments, that serves as the running illustration of this paper and places it on the same footing as the preceding designs.}

\begin{example}[IV indexing several data environments]\label{ex:iv_subset}
Suppose $O=(Z,A,Y)$ with $Z\in\{0,1,\ldots,q-1\}$ an environment index, $A\in\{0,1\}$ treatment, and $Y\in\mathbb{R}$ the outcome. Environment $Z=0$ is a randomized controlled study in which $A\perp\!\!\!\perp Y^a$ for $a\in\{0,1\}$; environments $Z=z>0$ are observational, with treatment propensity $\pi_z=P(A=1\mid Z=z)$ varying across environments for exogenous reasons (institutional, geographic, or protocol differences). Because $\pi_z$ varies exogenously across environments, the environment indicator serves as a multi valued instrument for treatment. For each subset $S\subseteq\{0,1,\ldots,q-1\}$ with $|S|\geq 2$:
\begin{enumerate}
	\item[\emph{(I1)}] \emph{Relevance}: $\pi_s=E[A\mid Z=s]$ is not constant over $s\in S$ (non degenerate first stage within $S$).
	\item[\emph{(I2)}] \emph{Exclusion}: $Y^a\perp\!\!\!\perp Z\mid(Z\in S)$; the environment index affects $Y$ only through $A$.
\end{enumerate}
Let $Z_S$ denote the $|S|$-dimensional one hot encoding of $Z$ restricted to $S$, with all expectations conditional on $Z\in S$. Writing $X=(1,A)^\top$ for the bivariate regressor vector, under \emph{(I1)--(I2)} the 2SLS formula for the coefficient vector is
\[
	\begin{pmatrix}\operatorname{intercept}_S\\\psi_S(P)\end{pmatrix}
	=\Bigl(E[XZ_S^\top]\,E[Z_SZ_S^\top]^{-1}\,E[Z_SX^\top]\Bigr)^{-1}
	 E[XZ_S^\top]\,E[Z_SZ_S^\top]^{-1}\,E[Z_SY],
\]
and $\psi_S(P)$ is the second component. Because $E[Z_SZ_S^\top]=\mathrm{diag}(P(Z=s\mid Z\in S))_{s\in S}$, the product $E[XZ_S^\top]\,E[Z_SZ_S^\top]^{-1}$ reduces to the matrix of conditional means $E[X\mid Z=s]$ for $s\in S$, so the formula depends only on the within-$S$ conditional moments of $(A,Y)$.

When $|S|=2$, $S=\{s_1,s_2\}$, the $2\times 2$ Gram matrix $E[XZ_S^\top]\,E[Z_SZ_S^\top]^{-1}\,E[Z_SX^\top]$ has determinant $P(Z=s_1\mid Z\in S)\,P(Z=s_2\mid Z\in S)\,(E[A\mid Z=s_1]-E[A\mid Z=s_2])^2$ and the expression collapses to the Wald ratio
\[
	\psi_S(P)
	=\frac{E[Y\mid Z=s_1]-E[Y\mid Z=s_2]}{E[A\mid Z=s_1]-E[A\mid Z=s_2]}.
\]
All $J=2^q-q-1$ non trivial subsets (those with $|S|\geq 2$) yield distinct identifying functionals. {Conditions \emph{(I1)--(I2)} make each $\psi_S(P)$ a well defined IV estimand but do not by themselves force $\psi_S(P)=\psi^*$ across subsets. Let $\psi^*=E[Y^1-Y^0]$ denote the ATE and $A(z)$ the potential treatment under $Z=z$. Consistency and instrument exogeneity yield, with no further restriction on the outcome law, the nonparametric decomposition
\[
	\psi_S(P)=\psi^*+\frac{\mathrm{Cov}\bigl(A(s_1)-A(s_2),\,Y^1-Y^0\bigr)}{E\bigl[A(s_1)-A(s_2)\bigr]}\qquad\text{for }S=\{s_1,s_2\},
\]
and an analogous representation, with compliance type weights replacing $A(s_1)-A(s_2)$, for $|S|\geq 3$. Hence $\psi_S(P)=\psi^*$ for every $S$ whenever
\[
	\mathrm{Cov}\bigl(A(z),\,Y^1-Y^0\bigr)=0\qquad\text{for every }z\in\{0,\ldots,q-1\},
\]
a condition requiring only that a unit's environment specific compliance behavior is uncorrelated with its treatment gain. It imposes no additivity, homoskedasticity, or functional form on $(Y^0,Y^1)$, and is strictly weaker than constant unit level effects (the latter corresponds to $Y^1-Y^0$ degenerate). When it holds the design sits in the $\tau^2=0$ regime; when it fails, each $\psi_S(P)$ identifies a subset specific weighted average of unit level effects among compliers within $S$ and $\tau^2>0$ encodes the dispersion of these local effects across subsets.}

The asymptotic variance of the corresponding 2SLS estimator $\widehat\psi_S$ satisfies, under homoskedasticity and for $|S|=2$ \citep{wooldridge2010econometric,angrist2009mostly},
\[
	v_S
	=\frac{\sigma^2}{n\,P(Z\in S)\,\bigl(E[A\mid Z=s_1]-E[A\mid Z=s_2]\bigr)^2},
\]
{which makes explicit the cost of a weak first stage: small treatment propensity contrasts inflate $v_S$ sharply, and weak instrument subsets therefore receive low effective weight in the empirical Bayes aggregation. For any fixed $q$, $J=2^q-q-1$ is fixed and the large-$J$ regime does not apply; letting $q=q_n\to\infty$ with $n$ activates \ifmerged Theorem~\ref{thm:route1}\else Theorem~3 of the paper\fi, since $J$ grows exponentially in $q$ and condition (C4) (non dominance of weights) is easily satisfied for moderate $q$ even when individual functionals are biased.}
\end{example}

%------------------------------------------------------------------------------
\subsection{ATE under no unmeasured confounding}
%------------------------------------------------------------------------------

{A canonical source of identification multiplicity arises in the no hidden-confounding setting with baseline covariates $W$, in which the same ATE is identified simultaneously by inverse probability weighting and by outcome regression.}

\begin{example}[IPW and outcome regression]\label{ex:ipw_or}
An observation $O=(W,A,Y)$, follows the structural equations
\[
	W=f_W(U_W),\qquad A=f_A(W,U_A),\qquad Y=f_Y(A,W,U_Y),
\]
and potential outcomes $Y^a=f_Y(a,W,U_Y)$ for $a\in\{0,1\}$. Write $\pi_a(W)=P(A=a\mid W)$ and $\mu_a(W)=E[Y\mid A=a,W]$.
\begin{enumerate}
	\item[\emph{(C1)}] \emph{Consistency}: $Y=Y^A$.
	\item[\emph{(C2)}] \emph{Conditional exchangeability}: $Y^a\perp\!\!\!\perp A\mid W$ for $a\in\{0,1\}$.
	\item[\emph{(C3)}] \emph{Positivity}: $0<\pi_a(W)<1$ almost surely, for $a\in\{0,1\}$.
\end{enumerate}
Under \emph{(C1)--(C3)}, the ATE $\psi^*=E[Y^1-Y^0]$ is identified by two distinct functionals
\[
	\psi_{\mathrm{IPW}}(P) = E\!\left[\frac{YA}{\pi_1(W)}\right] - E\!\left[\frac{Y(1-A)}{\pi_0(W)}\right],
	\qquad
	\psi_{\mathrm{OR}}(P) = E[\mu_1(W)] - E[\mu_0(W)].
\]
The first depends on the propensity score, the second on the outcome regression; both equal $\psi^*$ under \emph{(C1)--(C3)} and share the same efficient influence curve
\[
	D^*(O,P) = \frac{A}{\pi_1(W)}\bigl(Y-\mu_1(W)\bigr) - \frac{1-A}{\pi_0(W)}\bigl(Y-\mu_0(W)\bigr) + \mu_1(W) - \mu_0(W) - \psi^*.
\]
This yields $J=2$ identifying functionals for the same scalar target. {When (C1)--(C3) hold exactly and both nuisances are correctly estimated, $\psi_{\mathrm{IPW}}(P)=\psi_{\mathrm{OR}}(P)=\psi^*$ and $\tau^2=0$. }
\end{example}

	%==============================================================================
	
	%==============================================================================

	%\subsection{Implied interventions}

%	\cite{rosenman2023empiricalbayesdoubleshrinkage} combine multivariate causal estimands and assume diagonal covariance (heteroscedastic but independent components). We have $J$ estimators (\ma{change J because it sunds like J-estimators})all built from the same data, so the covariance matrix is non-diagonal. We will see that our EB likelihood treats them as independent Gaussians with variances $v_j$; we restore dependence for the distribution of $\widehat\psi_{\mathrm{EB}}$ via the bootstrap. We target different notions of coverage and recover it under the conditions stated below.

%==============================================================================
% Supplementary appendices: detailed empirical applications
%==============================================================================
%==============================================================================
% HAPC (PC-HAGL) variant of the LaLonde application. Numbers come from the
% full-data CESGA FT3 run (HAPC_FRAC=1.0); see shrident/results_guide.md and the
% CSVs in shrident/data/. The linear-base-learner version lives in
% lalonde_application.tex and is left untouched. Figures: *_hapc.pdf, generated
% by data/lalonde_plots_hapc.py.
%==============================================================================
\section{Empirical application: LaLonde--Dehejia--Wahba (highly adaptive base learners)}
\label{app:lalonde}
%==============================================================================

We apply the framework to the canonical LaLonde--Dehejia--Wahba (LDW) programme
evaluation \citep{lalonde1986,dehejia1999}. The target is the average treatment
effect on the treated (ATT) of the National Supported Work (NSW) demonstration,
estimated from the randomised trial alone and augmented by four observational
comparison groups drawn from the Current Population Survey (CPS) and the Panel
Study of Income Dynamics (PSID). We construct $J=15$ identifying functionals---three
trial-only and twelve trial-plus-external-controls---and use three nested
sub-analyses to span the two regimes of the main text.

\emph{Every} nuisance function here---the outcome
regressions $\E[Y(0)\mid X]$ (OR-NB, OM-$k$) and the propensity/sampling scores
that drive the inverse-probability weights (IPW-$k$)---is estimated by
\emph{highly adaptive principal component} regression
\citep{wang2026highlyadaptiveprincipalcomponent,meixide2026highlyadaptiveempiricalrisk}
(PC-HAL): a cross-validated, $L_1$-constrained Highly Adaptive Lasso fit on a
leading principal component basis of the order-$2$ interaction spline expansion
(\texttt{norm}\,$=\!1$, interaction degree $2$, $\lfloor n^{0.6}\rfloor$ leading
components), under a Gaussian loss for the outcome regressions and a binomial
loss for the scores. No linear or logistic working model is used anywhere.
The difference-in-means functionals (DIM-RCT, DIM-$k$) involve no nuisance
estimation.
Every cross-validation curve attained an interior minimum, so the
regularisation grid was adequate for all ten fitted functionals.

%==============================================================================
\subsection{Data and functionals}
%==============================================================================

The Dehejia--Wahba subsample of the NSW experiment ($n_1=185$ treated, $185$
control, outcome 1978 real earnings) yields the experimental benchmark
$\widehat\psi^{\mathrm{expt}}=\$1{,}794$ (standard error, SE, $\$671$). Four external-control (EC)
groups supplement it: CPS-2 ($n=2,369$), CPS-3 ($n=429$), PSID-2 ($n=253$),
PSID-3 ($n=128$), all with $A_i=0$. Covariates $X_i$ are age, education,
race indicators, marriage, no-degree, and lagged earnings (\texttt{re74},
\texttt{re75}). The ATT decomposes as $\psi^*=\E[Y\mid A=1,S=1]-\theta_0$ with
$\theta_0=\E[Y(0)\mid S=1]$ identified under randomisation alone (RCT-only
functionals) or under the no-outcome-drift condition of \citet{zhu2023}
($\E[Y(0)\mid X,S=1]=\E[Y(0)\mid X,S=k]$, EC-augmented functionals).

\begin{table}[htbp]
\centering
\caption{The $J=15$ identifying functionals. Methods: DIM = difference in
means; OR/OM = outcome regression on covariates; IPW = Haj\'ek inverse
probability weighting. RCT-only functionals use only randomisation; the
remaining 12 add the no-drift condition against EC group $k$. Crude DIM-$k$
in addition requires the marginal counterpart $\E[Y(0)\mid S=1]=\E[Y(0)\mid S=k]$.
The OR/OM/IPW nuisances are fitted by highly adaptive principal-component
(PC-HAL) regression.}
\label{tab:ldw-functionals}
\renewcommand{\arraystretch}{1.05}
\begin{tabular}{cllll}
\toprule
$j$ & Label & Method & EC group \\
\midrule
1,2,3   & DIM-RCT, IPW-NB, OR-NB     & RCT-only       & --- \\
4,5,6   & DIM-, OM-, IPW-CPS-2       & DIM / OM / IPW & CPS-2  \\
7,8,9   & DIM-, OM-, IPW-CPS-3       & DIM / OM / IPW & CPS-3  \\
10,11,12& DIM-, OM-, IPW-PSID-2      & DIM / OM / IPW & PSID-2 \\
13,14,15& DIM-, OM-, IPW-PSID-3      & DIM / OM / IPW & PSID-3 \\
\bottomrule
\end{tabular}
\end{table}

Each $\widehat\psi_j$ is asymptotically linear,
$\widehat\psi_j-\psi_j=n^{-1}\sum_{i}\phi_j(O_i)+o_p(n^{-1/2})$, so its sampling
variance is $v_j=n^{-1}\operatorname{Var}\{\phi_j(O)\}$. We report the empirical
variance of the fitted scores, with the PC-HAL nuisance estimates plugged in.
Every functional shares the treated-arm term $\widehat\theta_1=\bar Y_1$, whose
influence function $A\,(Y-\theta_1)/P(A{=}1)$ contributes $S_{Y,1}^2/n_1$ (the
sample outcome variance over the $n_1=185$ treated). The families differ only in
the influence function of the control counterfactual
$\widehat\theta_0$. Write the comparator sample (the NSW controls for the
RCT-only functionals; external-control group $k$ otherwise) as
$\{(X_i,Y_i)\}_{i=1}^{m}$, let $\mu_0(X)=\E[Y(0)\mid X]$ be its outcome
regression, and let $\pi(X)$ be the score placing a unit in the trial-treated
cell rather than in the comparator (the propensity for the RCT-only functionals,
the sampling score for the external controls).

It is clarifying to read each $\phi_j$ against the efficient influence function.
For the ATT control mean, $\theta_0=\E[\mu_0(X)\mid A{=}1,S{=}1]$, the analogue of
$D^*$ in Example~\ref{ex:ipw_or} is
\begin{equation}\label{eq:ldw-eif}
  D^{*}_{0}(O)=
  \underbrace{\frac{A}{P(A{=}1)}\bigl(\mu_0(X)-\theta_0\bigr)}_{\textstyle T_{\mathrm{proj}}\ (\text{uses }\mu_0)}
  \;+\;
  \underbrace{\frac{1-A}{P(A{=}1)}\,\frac{\pi(X)}{1-\pi(X)}\bigl(Y-\mu_0(X)\bigr)}_{\textstyle T_{\mathrm{aug}}\ (\text{uses }\pi\text{ and }\mu_0)} .
\end{equation}
Since $\E[T_{\mathrm{aug}}\mid X,A]=0$ while $T_{\mathrm{proj}}$ is
$\sigma(X,A)$-measurable, both pieces are uncorrelated and
$\operatorname{Var}(D^*_0)=\operatorname{Var}(T_{\mathrm{proj}})+\operatorname{Var}(T_{\mathrm{aug}})$.
Only the \emph{augmented} (doubly robust) estimator that uses both nuisances has
influence function $D^*_0$ and is efficient. The functionals tabulated here are
deliberately \emph{single-nuisance}: each uses one nuisance, is therefore
inefficient, and reconstructs only part of \eqref{eq:ldw-eif}---which is exactly
what makes them distinct identifications of the same $\psi^*$. Their variances
are most transparent read this way.

\paragraph{\textbf{DIM, model-free.}} The crude estimator differences the two raw
means, $\widehat\theta_0^{\mathrm{DIM}}=\bar Y_{(m)}$, with influence function the
centred outcome in each arm:
\[
v^{\mathrm{DIM}}_j=\frac{S_{Y,1}^2}{n_1}+\frac{S_{Y,(m)}^2}{m}.
\]
No nuisance is fitted; treated and comparator means are computed on disjoint
units, so this variance is exact and the rows are identical to the linear
analysis.

\paragraph{\textbf{IPW (Haj\'ek), score only.}} Using only $\widehat\pi$ through the
odds weights $w_i=\widehat\pi(X_i)/\{1-\widehat\pi(X_i)\}$, the self-normalised
(Haj\'ek) control mean is $\widehat\theta_0^{\mathrm{IPW}}=\sum_i w_iY_i/\sum_i w_i$,
with influence function the weight-normalised, recentred score
$\phi_i=w_i(Y_i-\widehat\theta_0^{\mathrm{IPW}})/\bar w$, $\bar w=m^{-1}\sum_i w_i$:
\[
v^{\mathrm{IPW}}_j=\frac{S_{Y,1}^2}{n_1}
+\frac{1}{m}\,\widehat{\operatorname{Var}}\!\left\{\frac{w_i\,(Y_i-\widehat\theta_0^{\mathrm{IPW}})}{\bar w}\right\}.
\]
Lacking $\mu_0$, this estimator recentres on the marginal $\theta_0$ rather than
on $\mu_0(X)$, so it cannot separate $T_{\mathrm{proj}}$ from $T_{\mathrm{aug}}$:
it captures the score-driven direction of \eqref{eq:ldw-eif} but folds the
projection into the same inverse-odds-weighted residual, which is the source of
its inefficiency. The self-normalisation by $\bar w$ tames the plain-IPW ratio;
the term still inflates when the weights are ill-conditioned (IPW-CPS-2, SE
$\approx\$1{,}618$).

\paragraph{\textbf{OR/OM (outcome regression), outcome model only.}} Using only
$\widehat\mu_0\equiv\widehat m$, fitted on the comparator sample (the NSW controls
for OR-NB; the pooled NSW-control $\cup$ EC-$k$ sample for OM-$k$) and averaged
over the treated covariates,
$\widehat\theta_0^{\mathrm{OR}}=n_1^{-1}\sum_{i\in\mathrm{trt}}\widehat m(X_i)$. Its
influence function is $T_{\mathrm{proj}}$ plus the first-order
error of having estimated $\mu_0$ on the comparator sample; by the orthogonality
in \eqref{eq:ldw-eif} the variance is a sum with no cross term,
\[
v^{\mathrm{OR}}_j=\frac{S_{Y,1}^2}{n_1}
+\underbrace{\frac{1}{n_1}\,\widehat{\operatorname{Var}}\bigl\{\widehat m(X_i):i\in\mathrm{trt}\bigr\}}_{\text{prediction}}
+\underbrace{\frac{1}{m}\,\widehat{\operatorname{Var}}\bigl\{Y_i-\widehat m(X_i)\bigr\}}_{\text{residual}} .
\]
The \emph{prediction} term is $\operatorname{Var}(T_{\mathrm{proj}})$ exactly: it
equals $n_1^{-1}\operatorname{Var}\{\mu_0(X)\mid A{=}1\}$, the sampling variance of
averaging the fitted surface $\widehat m$ over the finite, random treated
covariates. It is not slack from over-fitting but the irreducible cost of
estimating an average over the target covariate law, and it is precisely the
projection half $\mu_0(X)-\theta_0$ of \eqref{eq:ldw-eif}. The \emph{residual}
term is an unweighted approximation for $\operatorname{Var}(T_{\mathrm{aug}})$, exact when $\pi$ is constant: a direct calculation
gives that the unweighted comparator-residual variance equals
$\operatorname{Var}(T_{\mathrm{aug}})$ when $\pi(X)$ is constant---the RCT-only
case, where randomisation fixes the propensity---and departs from it only through
the omitted odds reweighting $\pi/(1-\pi)$ when $\pi(X)$ varies (the
external-control functionals). Actually

$$\operatorname{Var}(T_{\rm aug})
=
\frac1{P(A=1)^2}
E\!\left[
(1-A)
\Bigl(\frac{\pi(X)}{1-\pi(X)}\Bigr)^2
(Y-\mu_0(X))^2
\right]. $$

\noindent After simplifying,

$$=
\frac1{P(A=1)^2}
E\!\left[
\frac{\pi(X)^2}{1-\pi(X)}
\sigma_0^2(X)
\right].$$

\noindent While our residual variance estimates $E[\sigma_0^2(X)\mid A=0]$ instead.

Two approximations therefore enter the model-based $v_j$. First, for the
RCT-only OR-NB and IPW-NB functionals the trial-treated units enter both
$\widehat\theta_1$ and $\widehat\theta_0$ (directly in the OR average over treated
covariates, through $\widehat\pi$ for IPW), so treating the two as independent neglects a covariance term induced by the reuse of treated observations in both components; in practice this covariance is often positive, making the variance estimate mildly conservative. Second, for the
external-control OR/OM functionals the residual term omits the $\pi/(1-\pi)$
reweighting of $T_{\mathrm{aug}}$: exact under a constant score, first order
otherwise. The OM standard errors are stable across comparators
(Table~\ref{tab:ldw-estimates}).

%==============================================================================
\subsection{Results}
%==============================================================================

\begin{table}[htbp]
\centering
\caption{LDW with PC-HAL base learners: 15 functionals (point estimate,
influence-function SE) and EB results across three nested sub-analyses.
Experimental benchmark $\$1{,}794$.}
\label{tab:ldw-estimates}
\renewcommand{\arraystretch}{1.1}
\begin{tabular}{clrr}
\toprule
$j$ & Functional & $\widehat\psi_j$ & SE \\
\midrule
\multicolumn{4}{l}{\textit{RCT-only}}\\
1 & DIM-RCT & 1{,}794 & 671 \\
2 & IPW-NB  & 1{,}751 & 679 \\
3 & OR-NB   & 1{,}765 & 670 \\
\midrule
\multicolumn{4}{l}{\textit{EC covariate-adjusted}}\\
5 & OM-CPS-2   & 1{,}384 & 620 \\
6 & IPW-CPS-2  & 1{,}451 & 1{,}618 \\
8 & OM-CPS-3   & 1{,}839 & 637 \\
9 & IPW-CPS-3  & 427     & 959 \\
11& OM-PSID-2  & 1{,}474 & 672 \\
12& IPW-PSID-2 & 1{,}506 & 1{,}060 \\
14& OM-PSID-3  & 1{,}550 & 659 \\
15& IPW-PSID-3 & 1{,}220 & 1{,}209 \\
\midrule
\multicolumn{4}{l}{\textit{EC crude (marginal no-confounding required)}}\\
4 & DIM-CPS-2  & $-$3{,}822 & 606 \\
7 & DIM-CPS-3  & $-$635     & 677 \\
10& DIM-PSID-2 & $-$3{,}647 & 910 \\
13& DIM-PSID-3 & 1{,}070    & 897 \\
\bottomrule
\end{tabular}
\hspace{0.8cm}
\begin{tabular}{lccc}
\toprule
Sub-analysis & $J$ & $\widehat\tau$ & $\widehat\psi_{\eb}$ (SE) \\
\midrule
RCT-only           & 3  & \$0       & \$1{,}770 (\$389) \\
Covariate adjusted & 11 & \$0       & \$1{,}552 (\$228) \\
Full pool          & 15 & \$1{,}852 & \$589   (\$525) \\
\bottomrule
\end{tabular}
\end{table}

\begin{figure}[htbp]
\centering
\includegraphics[width=0.85\textwidth]{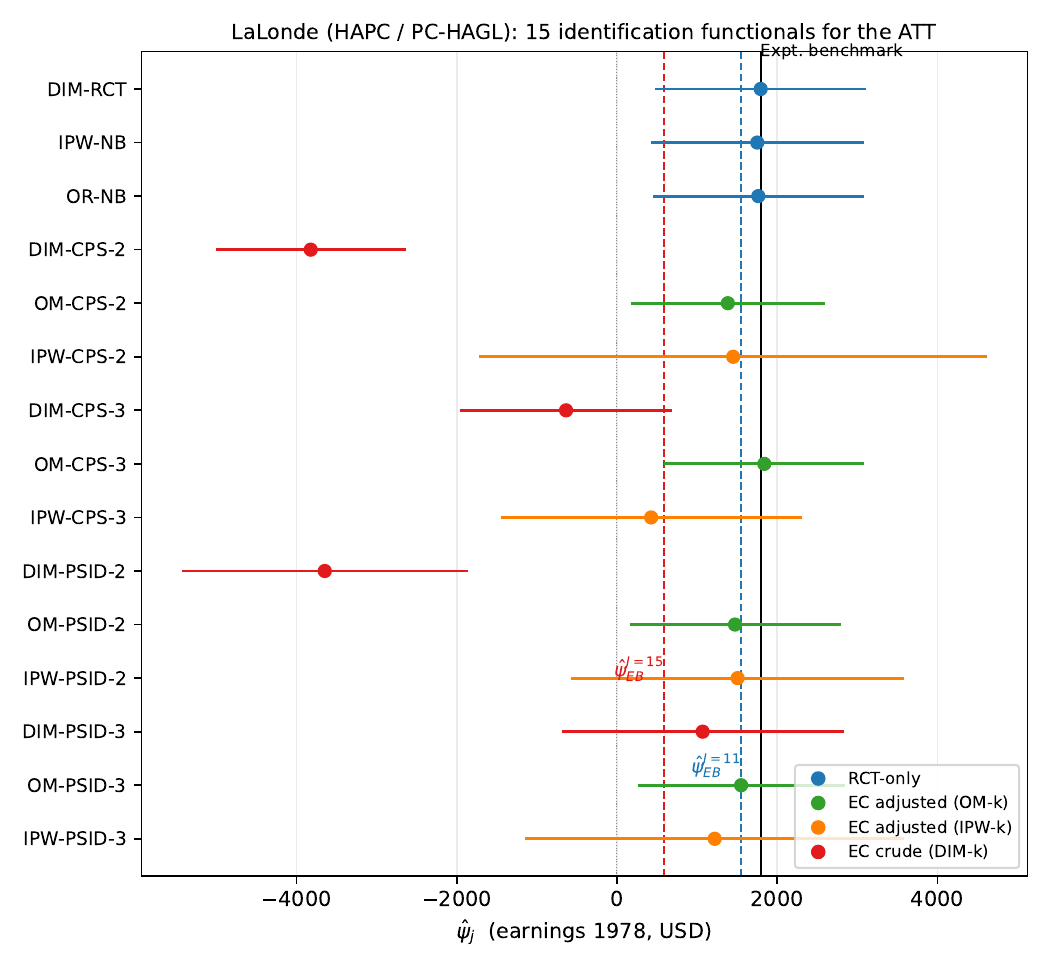}
\caption{LDW forest plot of the 15 functionals with PC-HAL base learners
(95\% influence-function CIs).
Vertical solid: experimental benchmark \$1{,}794. Dashed blue: $\widehat\psi_{\eb}$
under $J=11$ (\$1{,}552). Dashed red: $\widehat\psi_{\eb}$ under $J=15$ (\$589).}
\label{fig:ldw-forest}
\end{figure}

The three sub-analyses span the two regimes anticipated by the main text:

\paragraph{\textbf{RCT-only ($J=3$)}.} The pairwise differences across DIM, IPW and OR
are an order of magnitude smaller than the typical $\sqrt{v_j}\approx\$670$. The EB pooling
reduces the SE from \$671 to \$389, a 42\% gain from combining three
asymptotically equivalent but statistically complementary estimators of the
same target.

\paragraph{\textbf{Covariate-adjusted ($J=11$)}.} Adding the eight EC-augmented
estimates again yields $\widehat\tau=\$0$: the spread of the 11 covariate-adjusted
estimates---including the comparatively low IPW-CPS-3 (\$427)---is still within
sampling variability. The EB estimate $\$1{,}552$ is
within $14\%$ of the experimental benchmark, with SE \$228---a $66\%$
reduction over the trial-only DIM. This is the precision payoff from valid
external augmentation.

\paragraph{\textbf{Full pool ($J=15$)}.} Including the four crude DIM-$k$ functionals
drives $\widehat\tau$ to \$1{,}852, an order of magnitude above the typical
$\sqrt{v_j}$. Three of the four crude estimates are strongly negative
($-\$3{,}822$, $-\$3{,}647$, $-\$635$), pulling the pooled mean to \$589,
far from the benchmark. The biases are uncentered, so the consistency
guarantee of \ifmerged Theorem~\ref{thm:route1}\else Theorem~3 of the paper\fi does not apply. The magnitude of
$\widehat\tau$ relative to the individual $\sqrt{v_j}$ flags the inclusion of
invalid functionals.

%==============================================================================
\subsection{Conformal prediction in the $\widehat\tau^2>0$ regime}
%==============================================================================

In the full pool, $\widehat\tau^2>0$ and the appropriate inferential output for a
new EC-based functional is a prediction interval. Following
\ifmerged Section~\ref{sec:bayesian}\else Section~4 of the paper\fi, we report split conformal and leave-one-out
(LOO) variants. With $J_{\mathrm{cal}}=8$ the finest two-sided level is
$1-2/9\approx 77.8\%$, and with $J_{\mathrm{cal}}=15$ (LOO) it is $87.5\%$.

\begin{table}[htbp]
\centering
\caption{Prediction intervals for a new functional $\widehat\psi_{16}$ with
$v_{16}=\bar v_{15}$ ($J=15$, $\widehat\tau=\$1{,}656$, $\widehat\psi_{\eb}=\$587$;
pairwise variant). Target: $\psi_{16}\sim N(\psi^*,\tau^2)$.}
\label{tab:ldw-pi}
\renewcommand{\arraystretch}{1.1}
\begin{tabular}{llrrr}
\toprule
Method & Level & Lower & Upper & Width \\
\midrule
Normal (parametric)   & 90\%   & $-$\$2{,}497 & $+$\$3{,}671 & \$6{,}168 \\
Split conformal ($B=2{,}000$) & 77.8\% & $-$\$4{,}052 & $+$\$1{,}898 & \$5{,}950 \\
LOO conformal         & 87.5\% & $-$\$6{,}218 & $+$\$1{,}987 & \$8{,}205 \\
\bottomrule
\end{tabular}
\end{table}

All three intervals are strongly left-skewed: the conformal upper bounds
cluster near $\$1{,}900$--$\$1{,}990$ while the lower bound reaches down to
$-\$6{,}218$. This reflects the uncentered DIM-$k$ scores: the two extreme LOO
scores $\widehat S_{\text{DIM-CPS-2}}^{\mathrm{LOO}}=-3.63$ and
$\widehat S_{\text{DIM-PSID-2}}^{\mathrm{LOO}}=-2.94$ pull the lower quantile
far from zero. The parametric normal interval misses this asymmetry by
construction. The standard deviation of the lower split-conformal endpoint
across $B=2{,}000$ random 50/50 splits is \$2{,}005---larger than
$\widehat\tau$ itself---reflecting the well-known instability of split conformal
at small $J_{\mathrm{cal}}$; the LOO variant uses all 15 calibration scores at
the cost of slight self-contamination of the hyperparameter fit.

\begin{table}[htbp]
\centering
\caption{LOO conformal scores $\widehat S_j^{\mathrm{LOO}}=(\widehat\psi_j-\widehat\psi_{\eb}^{(-j)})/\sqrt{\widehat\tau^2_{(-j)}+v_j}$ for the 15 functionals (PC-HAL base learners). Empirical mean $-0.103$, SD $1.306$; only DIM-CPS-2 and DIM-PSID-2 fall outside $[-2,2]$.}
\label{tab:ldw-loo}
\renewcommand{\arraystretch}{1.05}
\begin{tabular}{cllclll}
\toprule
$j$ & Func.\ & $\widehat S_j^{\mathrm{LOO}}$ & \quad & $j$ & Func.\ & $\widehat S_j^{\mathrm{LOO}}$ \\
\midrule
1 & DIM-RCT          & $+0.71$ & & 9  & IPW-CPS-3 & $-0.09$ \\
2 & IPW-NB           & $+0.69$ & & 10 & \textbf{DIM-PSID-2} & $\mathbf{-2.94}$ \\
3 & OR-NB            & $+0.70$ & & 11 & OM-PSID-2 & $+0.52$ \\
4 & \textbf{DIM-CPS-2} & $\mathbf{-3.63}$ & & 12 & IPW-PSID-2& $+0.48$ \\
5 & OM-CPS-2         & $+0.47$ & & 13 & DIM-PSID-3& $+0.26$ \\
6 & IPW-CPS-2        & $+0.38$ & & 14 & OM-PSID-3 & $+0.57$ \\
7 & DIM-CPS-3        & $-0.72$ & & 15 & IPW-PSID-3& $+0.32$ \\
8 & OM-CPS-3         & $+0.75$ & &    &           &        \\
\bottomrule
\end{tabular}
\end{table}

The two outlying LOO scores are exactly the functionals whose raw EC means
($\$10{,}171$ for CPS-2, $\$9{,}996$ for PSID-2) most differ from the NSW
control mean (\$4{,}555). Holding either of them out moves $\widehat\psi_{\eb}$
of the remaining 14 functionals up from \$589 to \$863--\$955. Removing all
four crude DIM-$k$ scores recovers the covariate-adjusted \$1{,}552. The LOO
score is therefore a per-functional diagnostic of incompatibility with the
working model, complementary to the aggregate magnitude of $\widehat\tau$.

When attention is restricted to $J=11$ (covariate-adjusted), $\widehat\tau^2=0$
and conformal prediction is not applicable: there is no i.i.d.\ latent
signal for the procedure to exploit, and the appropriate output is the
sandwich confidence interval $\widehat\psi_{\eb}\pm 1.96\,\widehat{\mathrm{SE}}=
[\$1{,}106,\$1{,}999]$, which covers the experimental benchmark.

The two phase transitions in the same dataset are
diagnostic, and are robust to replacing the linear/logistic base learners with
flexible PC-HAL nuisance estimators. In the RCT-only and covariate-adjusted
sub-analyses, $\widehat\tau$ collapses to the boundary and the sandwich CI
delivers a substantial precision gain over the trial-only DIM. In the full pool
the crude DIM-$k$ functionals introduce uncentered bias, $\widehat\tau$ jumps by
an order of magnitude, and the conformal prediction interval honestly absorbs
the resulting between-functional spread.

%==============================================================================
\section{Empirical application: Card and Krueger (1994)}
\label{app:ck94}
%==============================================================================

We apply the framework to the \citep{card1994} minimum-wage
study. The target is the average treatment effect on the treated (ATT) of the
1992 New Jersey (NJ) minimum-wage increase on fast-food employment, estimated
by difference-in-differences (DiD) using multiple Pennsylvania (PA) subgroups
as alternative controls. We construct $J=25$ identifying functionals---nine
difference-in-means (DIM) and sixteen covariate-adjusted estimators based on
outcome-model regression (OM) and inverse probability weighting (IPW)---and
document strong
between-functional heterogeneity driven by the choice of control
subpopulation rather than by the estimation method. The contrast with
Appendix~\ref{app:lalonde} is sharp: there, heterogeneity arose from a
quality dimension (crude vs covariate-adjusted estimators); here it arises
from the control \emph{subpopulation}, persisting within every estimation
method.

%==============================================================================
\subsection{Data and functionals}
%==============================================================================

The dataset contains $n_1=309$ NJ restaurants (treated) and 75 PA restaurants
(untreated) surveyed before and after the NJ increase from \$4.25 to \$5.05
in April 1992. The outcome is the change in full-time-equivalent (FTE)
employment between waves,
$\Delta Y=\text{empft2}+0.5\,\text{emppt2}+\text{nmgrs2}-\text{empft}-0.5\,\text{emppt}-\text{nmgrs}$.
The NJ treated mean is $\bar\Delta Y_{\mathrm{NJ}}=0.47$ FTE; the Card--Krueger
benchmark estimate using all PA as control is DIM-PA-all $=2.75$ FTE.

The 75 PA restaurants admit nine natural subgroupings: \emph{region}
(PA-all $n=75$, PA-north $n=34$, PA-east $n=41$), \emph{chain} (PA-BK $n=33$,
PA-KFC $n=12$, PA-RR $n=17$, PA-Wendy's $n=13$), and \emph{ownership}
(PA-corp $n=26$, PA-fran $n=49$). For each $\mathcal{C}_k$ we impose
parallel trends (unconditional for DIM, conditional on $X_i\in\{\text{chain},\text{ownership}\}$ for OM and IPW (region is being used to define the control subgroup, not as a covariate available on both sides of the comparison)):
\[
\psi_k^{\mathrm{DIM}}=\E[\Delta Y\mid A=1]-\E[\Delta Y\mid i\in\mathcal{C}_k],
\]
\[
\psi_k^{\mathrm{OM}}=\E[\Delta Y\mid A=1]-\E[\mu_0^k(X)\mid A=1],\qquad
\mu_0^k(X)=\E[\Delta Y\mid X,i\in\mathcal{C}_k],
\]
\[
\psi_k^{\mathrm{IPW}}=\E[\Delta Y\mid A=1]-\frac{\E[e^k(X)\Delta Y/(1-e^k(X))\,\mathbf{1}(i\in\mathcal{C}_k)]}{\E[e^k(X)/(1-e^k(X))\,\mathbf{1}(i\in\mathcal{C}_k)]},
\]
with $e^k(x)=P(A=1\mid X=x,i\in\mathcal{C}_k\cup\{A=1\})$. We have nine DIM
functionals (one per subgroup); OM and IPW contribute eight each (PA-BK is
excluded because all its restaurants share chain and ownership, leaving no
covariate variation), giving $J=25$.

Sampling variances $v_j$ use the heteroskedasticity-consistent estimators
HC2 for DIM and an HC1-based influence-function formula for OM, and the
Haj\'ek variance for IPW. The hyperparameters
$(\widehat\psi_{\eb},\widehat\tau^2)$ follow the pairwise moment construction of
\ifmerged Section~\ref{sec:identifiability}\else Section~2 of the main paper\fi.

%==============================================================================
\subsection{Results}
%==============================================================================

\begin{table}[htbp]
\centering
\caption{The $J=25$ functionals: estimate, HC standard error, EB posterior
mean. Footer: heterogeneity summary. Asterisked rows are detailed in
Section~\ref{subsec:ck94-conformal}.}
\label{tab:ck94-functionals}
\renewcommand{\arraystretch}{1.05}
\begin{tabular}{llccc}
\toprule
Label & Control group & $\widehat\psi_j$ & SE & $\widehat\psi_j^{\eb}$ \\
\midrule
DIM-PA-all    & All PA         & $2.750$ & $1.342$ & $2.700$ \\
OM-PA-all     & All PA         & $2.465$ & $0.486$ & $2.468$ \\
IPW-PA-all    & All PA         & $2.515$ & $0.501$ & $2.516$ \\
\addlinespace
DIM-PA-north  & PA northern    & $4.334$ & $1.555$ & $3.799$ \\
OM-PA-north$^\ast$  & PA northern & $6.852$ & $0.484$ & $6.682$ \\
IPW-PA-north$^\ast$ & PA northern & $7.100$ & $0.695$ & $6.743$ \\
\addlinespace
DIM-PA-east   & PA eastern     & $1.436$ & $1.989$ & $1.890$ \\
OM-PA-east    & PA eastern     & $0.867$ & $0.497$ & $0.937$ \\
IPW-PA-east   & PA eastern     & $0.921$ & $0.580$ & $1.012$ \\
\addlinespace
DIM-PA-BK     & Burger King    & $3.512$ & $2.461$ & $3.011$ \\
\addlinespace
DIM-PA-KFC$^\ast$  & KFC       & $-1.825$& $0.883$ & $-1.298$ \\
OM-PA-KFC$^\ast$   & KFC       & $-1.775$& $0.485$ & $-1.604$ \\
IPW-PA-KFC$^\ast$  & KFC       & $-1.775$& $0.560$ & $-1.549$ \\
\addlinespace
DIM-PA-RR     & Roy Rogers     & $4.393$ & $2.078$ & $3.593$ \\
OM-PA-RR      & Roy Rogers     & $4.481$ & $0.527$ & $4.390$ \\
IPW-PA-RR     & Roy Rogers     & $4.481$ & $0.987$ & $4.197$ \\
\addlinespace
DIM-PA-Wendy's & Wendy's        & $2.890$ & $2.660$ & $2.697$ \\
OM-PA-Wendy's  & Wendy's        & $2.501$ & $0.504$ & $2.503$ \\
IPW-PA-Wendy's & Wendy's        & $2.501$ & $0.835$ & $2.505$ \\
\addlinespace
DIM-PA-corp   & Corporate      & $2.515$ & $1.453$ & $2.522$ \\
OM-PA-corp    & Corporate      & $3.211$ & $0.499$ & $3.183$ \\
IPW-PA-corp   & Corporate      & $3.211$ & $0.652$ & $3.164$ \\
\addlinespace
DIM-PA-fran   & Franchise      & $2.875$ & $1.849$ & $2.750$ \\
OM-PA-fran    & Franchise      & $2.360$ & $0.492$ & $2.367$ \\
IPW-PA-fran   & Franchise      & $2.360$ & $0.546$ & $2.369$ \\
\midrule
\multicolumn{5}{l}{$\widehat\tau=2.38$ FTE;\quad
$\widehat\psi_{\eb}=2.54$ FTE;\quad
$\widehat\tau/\overline{\sqrt{v_j}}\approx 2.7$ across the pool} \\
\bottomrule
\end{tabular}
\end{table}

\begin{figure}[htbp]
\centering
\includegraphics[width=0.85\textwidth]{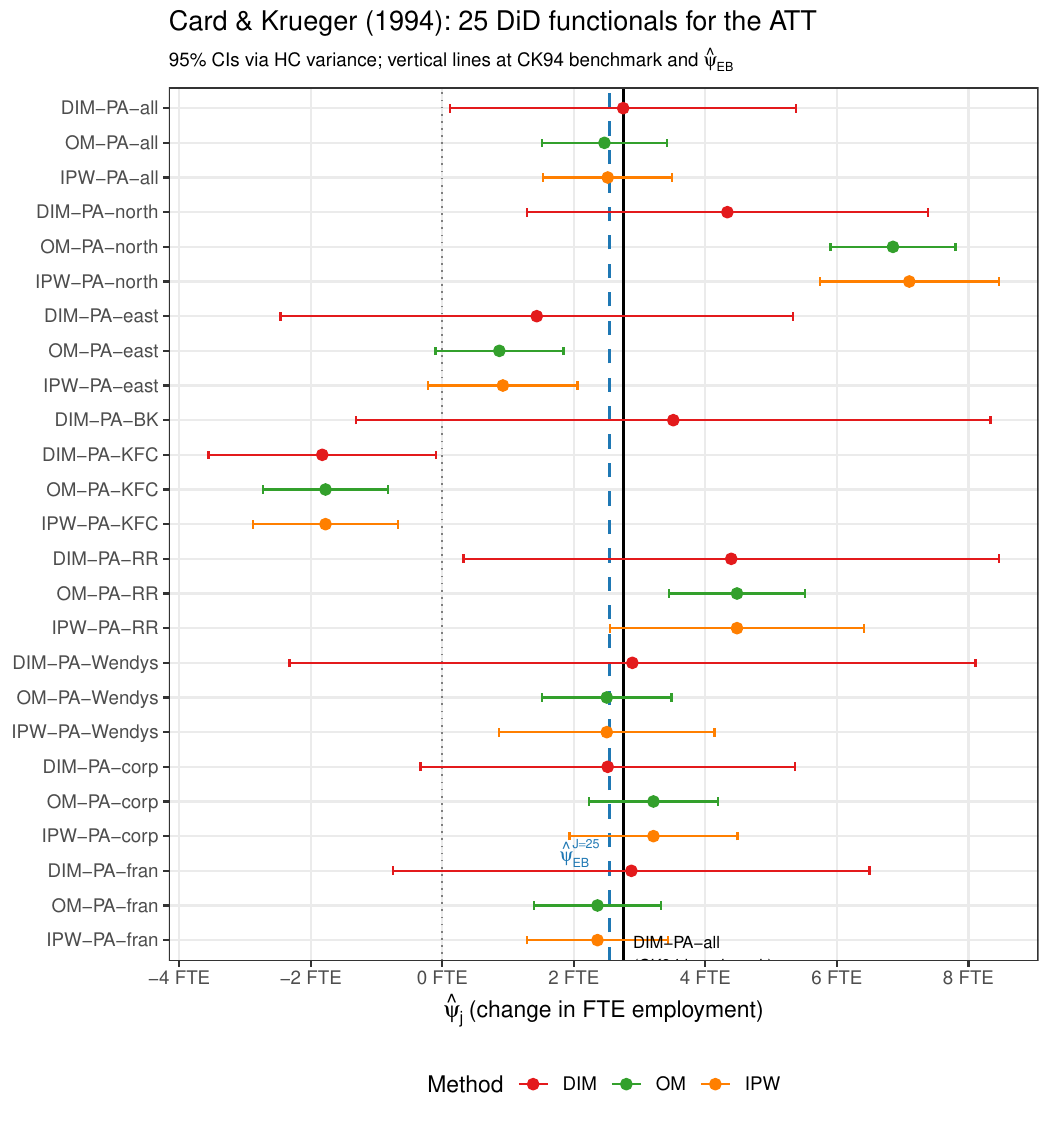}
\caption{Card--Krueger forest plot of 25 DiD functionals (95\% CIs).
Vertical black: DIM-PA-all (benchmark $2.75$ FTE). Dashed blue:
$\widehat\psi_{\eb}=2.54$ FTE. Colours indicate estimation method.}
\label{fig:ck94-forest}
\end{figure}

The estimates span $-1.83$ FTE (PA-KFC) to $+7.10$ FTE (IPW-PA-north), a
range of nearly nine FTE units. Sub-analyses by estimation method give
$\widehat\tau=2.14$ FTE for DIM-only ($J=9$), $2.51$ FTE for OM-only ($J=8$) and
$2.46$ FTE for IPW-only ($J=8$): each sits well above the typical
$\sqrt{v_j}$, and within every method the between-functional spread exceeds
the individual sampling error. The driver of heterogeneity is therefore not
the choice among DIM, OM and IPW, but the choice of control
\emph{subpopulation}: different PA subsets violate parallel trends in
different directions, exactly as the literature has documented \citep{neumark2000minimum,dube2010minimum}.
 
Because $\widehat\tau=2.38$ FTE is comparable to or larger than the typical
covariate-adjusted SE ($0.49$--$0.99$ FTE), the EB shrinkage factor
$B_j=\widehat\tau^2/(v_j+\widehat\tau^2)$ is close to one for all OM and IPW
functionals: posterior means are within $10\%$ of the raw estimates. DIM
functionals carry larger SEs and shrink more visibly. The EB pooled value
$\widehat\psi_{\eb}=2.54$ FTE is driven by the cluster of estimates near $2.5$
FTE (PA-all, PA-fran, PA-east, PA-Wendy's); the outlying PA-KFC and PA-north
groups receive low weight, not because of dispersion but because of their comparable precision.

%==============================================================================
\subsection{Conformal prediction intervals}
\label{subsec:ck94-conformal}
%==============================================================================

The $\widehat\tau^2>0$ regime calls for a prediction interval for a new
PA-based functional rather than a confidence interval for $\psi^*$. We
implement the split conformal construction of \ifmerged Section~\ref{sec:bayesian}\else Section~4 of the main paper\fi.
With $J_{\mathrm{cal}}=J=25$ in the leave-one-out (LOO) variant, the finest two-sided level
is $1-2/26\approx 92.3\%$.

\begin{table}[htbp]
\centering
\caption{Selected LOO conformal scores $S_j^{\mathrm{LOO}}=(\widehat\psi_j-\widehat\psi_{\eb}^{(-j)})/\sqrt{\widehat\tau^2_{(-j)}+v_j}$. The PA-KFC and PA-north scores form the tails.}
\label{tab:ck94-conformal}
\renewcommand{\arraystretch}{1.05}
\begin{tabular}{lcrr}
\toprule
Label & Method & $\widehat\psi_j$ & $S_j^{\mathrm{LOO}}$ \\
\midrule
DIM-PA-KFC   & DIM & $-1.83$ & $-2.415$ \\
OM-PA-KFC    & OM  & $-1.78$ & $-2.619$ \\
IPW-PA-KFC   & IPW & $-1.78$ & $-2.580$ \\
OM-PA-east   & OM  & $\phantom{-}0.87$ & $-0.902$ \\
DIM-PA-all   & DIM & $\phantom{-}2.75$ & $\phantom{-}0.101$ \\
OM-PA-Wendy's & OM  & $\phantom{-}2.50$ & $-0.011$ \\
OM-PA-RR     & OM  & $\phantom{-}4.48$ & $\phantom{-}1.069$ \\
OM-PA-north  & OM  & $\phantom{-}6.85$ & $\phantom{-}2.633$ \\
IPW-PA-north & IPW & $\phantom{-}7.10$ & $\phantom{-}2.702$ \\
\midrule
\multicolumn{4}{l}{LOO PI ($1-\alpha\approx 92.3\%$, $v_{J+1}=$ median): $[-2.60,\,7.85]$ FTE}\\
\bottomrule
\end{tabular}
\end{table}

\begin{figure}[htbp]
\centering
\includegraphics[width=0.65\textwidth]{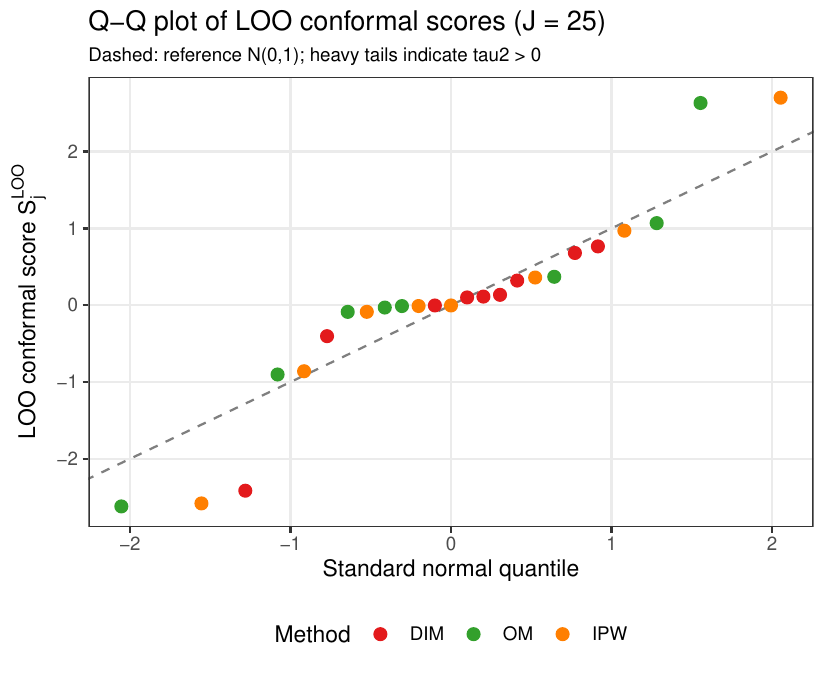}
\caption{Q-Q plot of LOO conformal scores against $N(0,1)$. The heavy
left tail (PA-KFC) and right tail (PA-north) confirm $\tau^2>0$ and
validate the non-parametric conformal approach.}
\label{fig:ck94-qq}
\end{figure}

Three observations follow.

\begin{enumerate}
	
\item 	The $92.3\%$ LOO PI
$[-2.60,7.85]$ honestly absorbs the disagreement across control groups.
Under the working model $\psi_j=\psi^*+\varepsilon_j$, a new PA-based
functional from the same heterogeneous distribution is expected to lie in
this range with at least the nominal probability, regardless of whether
$\widehat\psi_{\eb}$ is consistent for $\psi^*$.

\item The three KFC functionals
($\widehat\psi\approx -1.8$) carry LOO scores near $-2.6$, signalling that
PA-KFC restaurants had qualitatively different employment dynamics from
NJ restaurants. The structural reason is likely corporate restructuring at
KFC in 1992 unrelated to the wage policy. The conformal score flags the
subgroup without forcing the analyst to drop it.

\item The covariate-adjusted
PA-north estimates ($\approx 6.9$--$7.1$ FTE) are roughly double the DIM
value, suggesting that after conditioning on chain and ownership the NJ and
PA-north restaurants appear even more different---consistent with shared
labour markets along the NJ--PA border making border-county restaurants
poor controls.

\end{enumerate}

In contrast to a sandwich CI tied to a single chosen control group, the
conformal interval does not require ex-ante selection: it quantifies the
range a new Pennsylvania-based functional would plausibly produce. This is
the relevant uncertainty for policy conclusions when parallel trends is
contested, as the $30$-year literature on the Card--Krueger study
attests.

%==============================================================================
% Simulation studies (moved from the paper)
%==============================================================================
% Simulation studies (moved from arxiv.tex)
\section{Simulation studies}
\label{sec:simulations}

\subsection{A single-run illustration}
\label{sec:illustrative}

\begin{figure}[htbp]
	\centering
	\includegraphics[width=\textwidth]{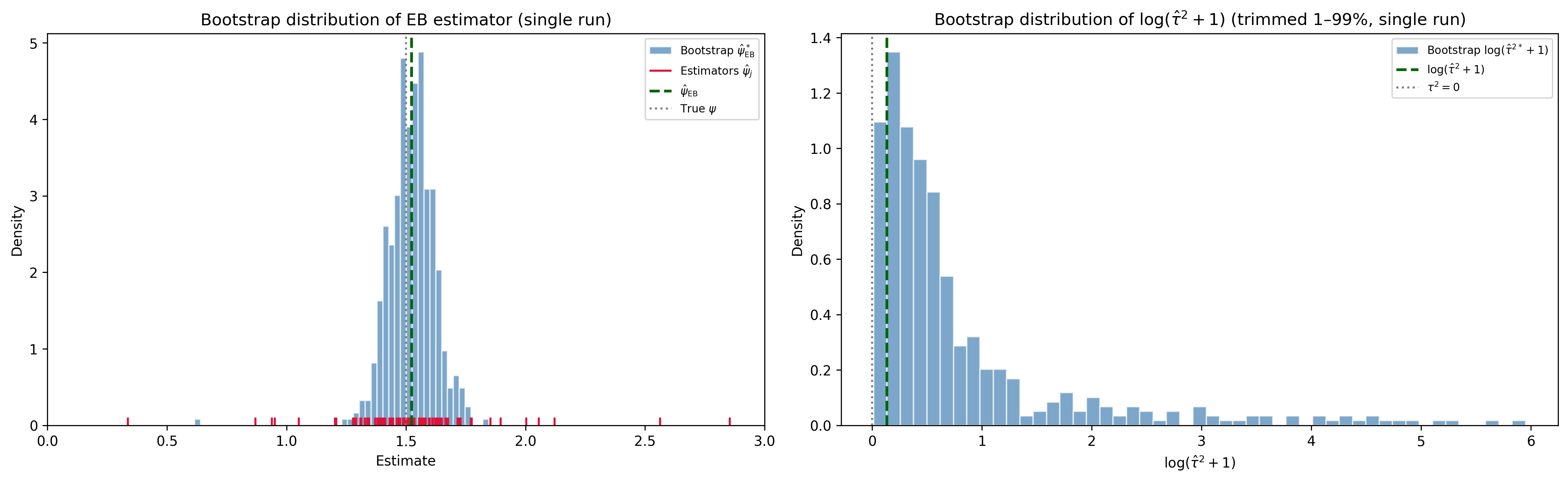}
	\caption{{Single-run output. Left panel: bootstrap distribution of $\widehat{\psi}_{\mathrm{EB}}^b$ with a rug of the $J$ individual estimators and vertical lines at $\widehat{\psi}_{\mathrm{EB}}$ and the true $\psi$. Right panel: bootstrap distribution of $\log(\widehat{\tau}^{2b}+1)$ (trimmed at the $1$st and $99$th percentiles) with $\log(\widehat{\tau}^2+1)$ and the line $\tau^2=0$. Setup: $q=7$, $J=121$, $n_{\mathrm{rct}}=50$, $n_{\mathrm{obs}}=1000$, $B=500$, and nominal coverage $90\%$.}}
	\label{fig:single_run}
\end{figure}

%==============================================================================

{Figure~\ref{fig:single_run} reports the output of a single run of the procedure on simulated data from the design of \ifmerged Example~\ref{ex:iv_subset}\else Example~5\fi (Appendix~\ref{sec:catalogue}), in which instrumental variables index several data environments. The intended nominal coverage is $1-\alpha=90\%$. With $q=7$ environments, the number of identifying functionals is $J=121$, comprising one RCT functional plus $2^{7}-7 -1 =120$ IV subset functionals. Environment-specific propensity scores were drawn i.i.d.\ from a $U(0.2,0.8)$ prior. The RCT sample size is $n_{\mathrm{rct}}=50$, each observational environment has $n_{\mathrm{obs}}=1000$ units, and the bootstrap uses $B=500$ resamples.}

{In the left panel of Figure 4, the histogram is the empirical density of the bootstrap replicates $\widehat{\psi}_{\mathrm{EB}}^b$; the rug marks the $J$ individual estimates $\widehat{\psi}_j$; the dashed vertical line is the empirical Bayes point estimate $\widehat{\psi}_{\mathrm{EB}}$; and the dotted vertical line is the true $\psi^*$. The bootstrap distribution of the combined estimator is concentrated relative to the dispersion of the individual estimators. In the right panel, the histogram displays the empirical density of $\log(\widehat{\tau}^{2b}+1)$. Mass near zero indicates that many bootstrap replicates return a negligible heterogeneity estimate.}

{Under standard asymptotics, the length of a $(1-\alpha)$ confidence interval for a location parameter scales as $n^{-1/2}$, so the squared length scales as $1/n$. The squared ratio between the length of the shrunk confidence interval and that arising from using just the only randomized environment was $0.067$. This translates, on this single draw, into an effective sample size gain of $n_{\mathrm{eff}}/n_{\mathrm{RCT}} \approx 1/0.067 \approx 15 \mathrm{times}$.}

\subsection{Exact identifiability and robustness to cross-functional correlation}
\label{subsec:sim1}

{This study isolates the two claims that are most exposed under exact identifiability: (i) the point estimator $\widehat\psi_{\mathrm{EB}}$ is consistent for $\psi^*$ even though the working model \ifmerged in~\eqref{eq:working-model}\else in~(1) of the paper\fi discards the cross-functional covariance (\ifmerged Theorem~\ref{thm:consistency}\else Theorem~2\fi); and (ii) the sandwich and subsampling intervals of \ifmerged Section~\ref{sec:frequentist}\else Section~5\fi attain their nominal level. }

{To target these claims without conflating them with the estimation of nuisance functions, we simulate the asymptotically linear representation of \ifmerged Section~\ref{sec:identifiability}\else Section~2 of the main paper\fi directly. For a sample of size $n$ and $J$ functionals we draw, independently across $i=1,\dots,n$, an influence vector $D_i=(D_{i1},\dots,D_{iJ})^\top$ from the one-factor Gaussian law
\[
D_{ij}=\sigma_j\bigl(\sqrt{\rho}\,F_i+\sqrt{1-\rho}\,E_{ij}\bigr),
\qquad F_i,\,E_{ij}\stackrel{\mathrm{iid}}{\sim}N(0,1),
\]
and set $\widehat\psi_j=\psi^*+n^{-1}\sum_{i=1}^n D_{ij}$. This yields $v_j=\sigma_j^2/n$ and $\operatorname{Cov}(\widehat\psi_j,\widehat\psi_k)=\rho\,\sigma_j\sigma_k/n$ for $j\neq k$, so every functional is valid, $\psi_j(P)=\psi^*$, and $\rho\in[0,1)$ is exactly the cross-functional correlation that the working model assumes to be zero. The marginal scales $\sigma_j$ are fixed and unequal (evenly spaced on $[0.6,2.0]$); $v_j$ is estimated from the realized influence contributions rather than treated as known.}

{At each replication we compute $\widehat\psi_{\mathrm{EB}}$ and the pairwise-difference $\widehat\tau^2$ of \ifmerged Section~\ref{sec:identifiability}\else Section~2 of the main paper\fi, and three nominal $90\%$ intervals for $\psi^*$: the canonical inverse-Fisher interval $\widehat\psi_{\mathrm{EB}}\pm z_{1-\alpha/2}\,(\sum_j \widehat w_j)^{-1/2}$ with $\widehat w_j=(v_j+\widehat\tau^2)^{-1}$; the observation-level sandwich interval of \ifmerged Section~\ref{subsec:sandwich-impl}\else Section~5.2 of the main paper\fi; and the subsampling interval of \ifmerged Appendix~\ref{app:subsampling}\else Section~5.3 of the main paper\fi with block size $m=\lfloor n^{0.7}\rfloor$.}

{We fix $J=20$, $\psi^*=1.5$ and nominal level $1-\alpha=0.9$, using $R=600$ Monte Carlo replications and $B=250$ subsamples. Consistency is assessed on $n\in\{50,100,200,400,800,1600\}$ for $\rho\in\{0,0.4,0.8\}$; coverage is assessed at $n=400$ for $\rho\in\{0,0.2,0.4,0.6,0.8,0.9\}$. The Monte Carlo standard error of a coverage probability evaluated at the nominal level is $\{0.9\cdot 0.1/600\}^{1/2}\approx 0.012$.}

\begin{figure}[htbp]
\centering
\includegraphics[width=\textwidth]{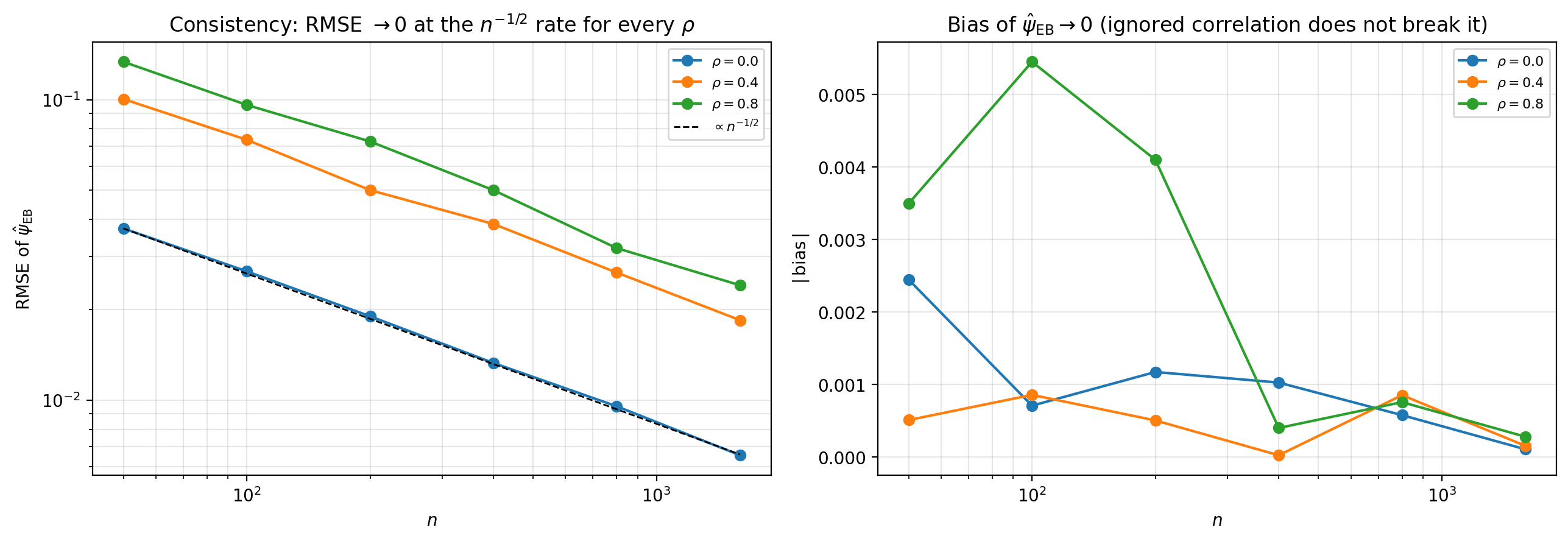}
\caption{{Root mean squared error (left, log--log) and absolute bias (right) of $\widehat\psi_{\mathrm{EB}}$ against $n$, for $\rho\in\{0,0.4,0.8\}$, under exact identifiability with $J=20$ fixed. The dashed reference line has slope $-1/2$ and is drawn to coincide with the $\rho=0$ curve at the largest $n$ (a visual anchor for the $n^{-1/2}$ rate).}}
\label{fig:sim1-consistency}
\end{figure}

{Figure~\ref{fig:sim1-consistency} shows that the RMSE decreases at the parametric rate for every $\rho$: the slope of $\log\mathrm{RMSE}$ on $\log n$ is $-0.50$, $-0.49$ and $-0.50$ at $\rho=0,0.4,0.8$ respectively, and the absolute bias at $n=1600$ never exceeds $3\times10^{-4}$ against an RMSE of order $10^{-2}$. Stronger correlation inflates the constant---RMSE at $n=50$ rises from $0.037$ at $\rho=0$ to $0.134$ at $\rho=0.8$---but not the rate. Discarding the covariance therefore costs efficiency, not consistency, exactly as \ifmerged Theorem~\ref{thm:consistency}\else Theorem~2 of the main paper\fi predicts.}

\begin{table}[htbp]
\centering
\caption{{Coverage at $n=400$ ($R=600$, Monte Carlo s.e.\ $\approx 0.012$). $\mathrm{sd}(\widehat\psi_{\mathrm{EB}})$ is the Monte Carlo standard deviation of the point estimator. The canonical inverse-Fisher interval undercovers severely once $\rho>0$; the sandwich holds the nominal level with length tracking $\mathrm{sd}(\widehat\psi_{\mathrm{EB}})$.}}
\label{tab:sim1}
\begin{tabular}{ccccccccc}
\toprule
& \multicolumn{3}{c}{coverage} & \multicolumn{3}{c}{mean length} & \\
\cmidrule(lr){2-4}\cmidrule(lr){5-7}
$\rho$ & inv.-Fisher & sandwich & subsampling & inv.-Fisher & sandwich & subsampling & $\mathrm{sd}(\widehat\psi_{\mathrm{EB}})$\\
\midrule
0.00 & 0.988 & 0.895 & 0.860 & 0.066 & 0.043 & 0.040 & 0.013\\
0.20 & 0.710 & 0.887 & 0.858 & 0.059 & 0.091 & 0.084 & 0.029\\
0.40 & 0.552 & 0.902 & 0.883 & 0.054 & 0.119 & 0.111 & 0.035\\
0.60 & 0.422 & 0.913 & 0.892 & 0.048 & 0.138 & 0.129 & 0.042\\
0.80 & 0.337 & 0.907 & 0.885 & 0.042 & 0.152 & 0.143 & 0.047\\
0.90 & 0.332 & 0.902 & 0.888 & 0.041 & 0.158 & 0.147 & 0.049\\
\bottomrule
\end{tabular}
\end{table}

\begin{figure}[htbp]
\centering
\includegraphics[width=\textwidth]{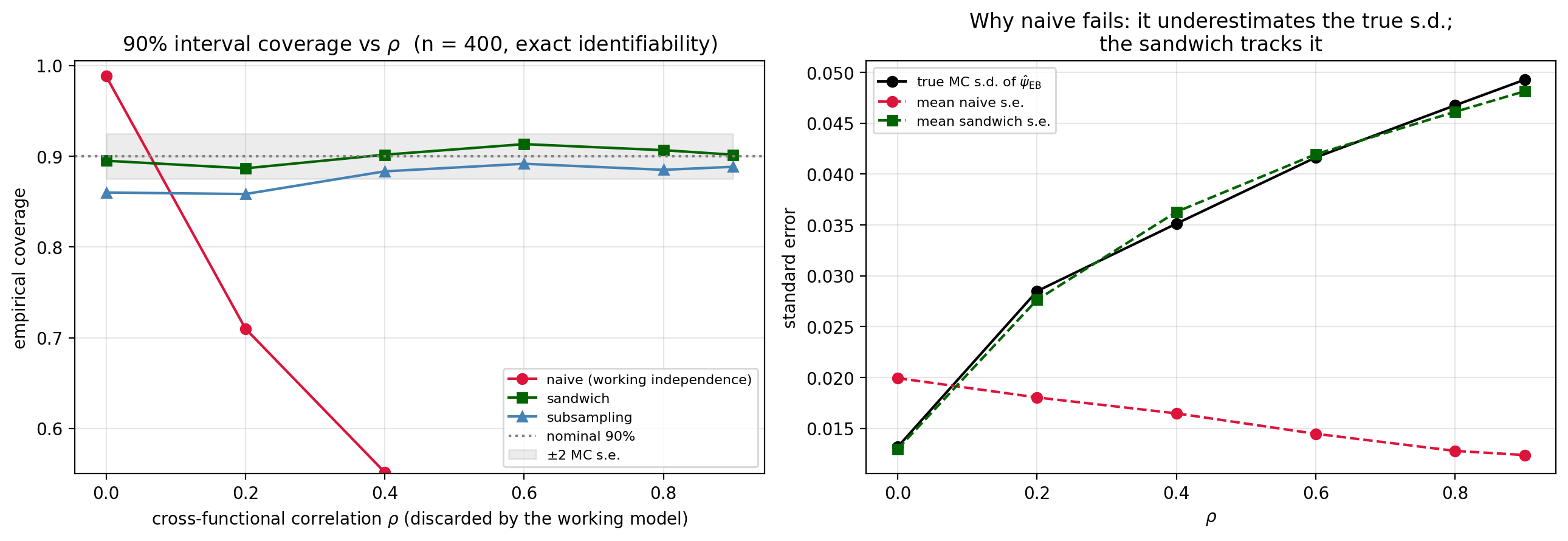}
\caption{{Coverage at $n=400$. Left: empirical coverage of the three $90\%$ intervals against the discarded cross-functional correlation $\rho$; the shaded band is $\pm 2$ Monte Carlo standard errors around the nominal level. Right: the mean canonical inverse-Fisher standard error collapses below the true sampling standard deviation of $\widehat\psi_{\mathrm{EB}}$ as $\rho$ grows, while the sandwich standard error tracks it.}}
\label{fig:sim1-coverage}
\end{figure}

{Table~\ref{tab:sim1} and Figure~\ref{fig:sim1-coverage} report the central finding. The canonical inverse-Fisher interval is severely anti-conservative once $\rho>0$: its coverage falls monotonically from a spuriously conservative $0.988$ at $\rho=0$ to $0.33$ at $\rho=0.9$, and its mean length \emph{decreases} (from $0.066$ to $0.041$) even as the true sampling standard deviation of $\widehat\psi_{\mathrm{EB}}$ more than triples (from $0.013$ to $0.049$)---the interval becomes more confident precisely when it should become less so. The sandwich interval holds its nominal level across the entire range, with coverage in $[0.887,0.913]$ (all within $1.1$ Monte Carlo standard errors of $0.90$) and mean length increasing in step with the true standard deviation. The subsampling interval can be mildly conservative, covering between $0.86$ and $0.89$; the small undercoverage at $n=400$ reflects the modest subsample size $m=\lfloor 400^{0.7}\rfloor=69$ and attenuates as $n$ grows, consistent with the boundary-robust-but-rougher role assigned to subsampling in \ifmerged Appendix~\ref{app:subsampling}\else Section~5.3 of the paper\fi.}

{The $\rho=0$ column is a correctly specified control: with no cross-functional covariance the working model is exact, so the inverse-Fisher interval is not expected to fail there. Its mild over-coverage ($0.988$) is the finite-sample boundary behaviour of the pairwise $\widehat\tau^2$, which is positive on roughly half of the replications and thereby inflates the canonical inverse-Fisher variance; it is not evidence against the method and, crucially, does not contaminate the sandwich, whose validity is invariant to the limit of $\widehat\tau^2$ by \ifmerged Proposition~\ref{prop:sandwich}(i)\else Proposition~5(i) of the paper\fi. The simulation confirms that invariance empirically: the sandwich attains $0.895$ at $\rho=0$ and remains at the nominal level for all $\rho$.}

The centered-identification-bias regime of \ifmerged Theorem~\ref{thm:route1}\else Theorem~3 of the paper\fi and the split-conformal prediction intervals of \ifmerged Section~\ref{sec:bayesian}\else Section~4 of the paper\fi are examined in Section~\ref{subsec:sim2}.

%==============================================================================

\subsection{Prediction intervals for the target of a new functional}
\label{subsec:sim2}

{This study checks the exact guarantee of \ifmerged Theorem~\ref{thm:conformal}\else Theorem~4 of the paper\fi: marginal coverage of the latent target $\psi_{J+1}=\psi^*+\varepsilon_{J+1}$ of a \emph{new} identifying functional, $\liminf_{n,J_{\mathrm{cal}}\to\infty}P(\psi_{J+1}\in C_{1-\alpha})\ge 1-\alpha$, where the probability is over the latent draws, the data, and the random training/calibration split, and $\tau^2>0$ by Assumption~(A1). It also verifies the necessity of the strict inequality in~(A1) (the phase transition) and the distribution-freeness in $G$.}

{As in Section~\ref{subsec:sim1} we simulate the asymptotically linear representation directly, now over $J+1$ functionals evaluated on the same sample. For $i=1,\dots,n$ we draw $D_{ij}=\sigma_j(\sqrt{\rho}\,F_i+\sqrt{1-\rho}\,E_{ij})$, $F_i,E_{ij}\stackrel{\mathrm{iid}}{\sim}N(0,1)$, $j=1,\dots,J+1$; latent deviations $\varepsilon_j\stackrel{\mathrm{iid}}{\sim}G$ with mean zero and variance $\tau^2$; and $\widehat\psi_j=\psi^*+\varepsilon_j+n^{-1}\sum_i D_{ij}$, so $\psi_j=\psi^*+\varepsilon_j$ and the new functional $J+1$ shares the sample (its sampling noise is correlated with the calibration noise through $F_i$). We take $G$ as Gaussian and, to probe distribution-freeness, a centered exponential with matched variance.}

{We evaluate three prediction sets for $\psi_{J+1}$ at nominal level $1-\alpha=0.9$: the one-sided region of \ifmerged Section~\ref{sec:bayesian}\else Section~4 of the paper\fi (the literal object of \ifmerged Theorem~\ref{thm:conformal}\else Theorem~4 of the paper\fi); its two-sided counterpart (the reported prediction interval); and the parametric Gaussian--Bayes interval $\widehat\psi_{\mathrm{EB}}^{\mathrm{train}}\pm z_{1-\alpha/2}\sqrt{\widehat\tau^2_{\mathrm{train}}+v_{J+1}}$, a foil that is valid only when $G$ is Gaussian.}

{We fix $J=400$, $n=800$, $\psi^*=1.5$, nominal level $0.9$, and $R=1500$ Monte Carlo replications, sweeping $\tau\in\{0,0.05,0.1,0.25,0.5,1,2\}$ and $\rho\in\{0,0.5,0.9\}$. The Monte Carlo standard error of a coverage probability at the nominal level is $\approx 0.008$.}

\begin{figure}[htbp]
\centering
\includegraphics[width=\textwidth]{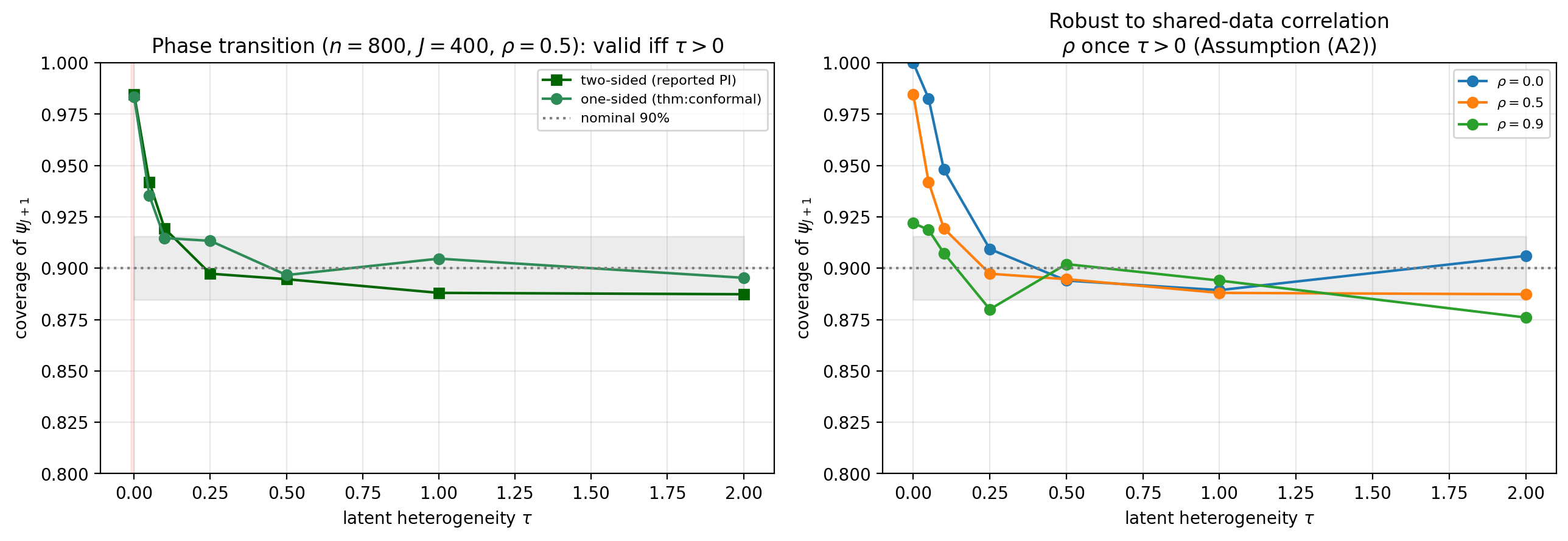}
\caption{{Visualization of the phase transition. Left: coverage of $\psi_{J+1}$ against the latent heterogeneity $\tau$ at $\rho=0.5$ for the one- and two-sided conformal regions; the guarantee holds for $\tau>0$ and fails at $\tau=0$ (red shading). Right: the conformal interval is robust to the discarded shared-data correlation $\rho$ once $\tau>0$ (Assumption~(A2)). The grey band in both panels is $\pm 2$ Monte Carlo standard errors around the nominal level $1-\alpha=0.9$.}}
\label{fig:sim2-phase}
\end{figure}

{Table~\ref{tab:sim2} and Figure~\ref{fig:sim2-phase} report coverage at $\rho=0.5$ under Gaussian $G$. For every $\tau>0$, the two-sided interval covers between $0.887$ and $0.942$ and the one-sided region between $0.895$ and $0.935$. Both satisfy the $\liminf\ge 1-\alpha$ statement, mildly conservative at small $\tau$ and settling at the nominal level as $\tau$ grows. The guarantee is robust to the shared-data correlation the working model discards: at $\tau=0.25$ the two-sided coverage is $0.909$, $0.897$ and $0.880$ at $\rho=0,0.5,0.9$, in line with the uniform-negligibility condition~(A2). At $\tau=0$ the i.i.d.\ signal vanishes, the calibration scores reduce to correlated sampling noise, and the guarantee fails---the conformal interval over-covers ($0.985$) rather than attaining the nominal level. This confirms that the strict inequality $\tau^2>0$ in~(A1) is the content of the phase transition, not a technical artifact.}

\begin{table}[htbp]
\centering
\caption{{Coverage of $\psi_{J+1}$ at $\rho=0.5$, Gaussian $G$ ($n=800$, $J=400$, $R=1500$, Monte Carlo s.e.\ $\approx 0.008$). The guarantee holds for every $\tau>0$ and fails at $\tau=0$ (the phase transition); the parametric Gaussian--Bayes interval over-covers throughout.}}
\label{tab:sim2}
\begin{tabular}{ccccc}
\toprule
$\tau$ & one-sided & two-sided & parametric & length (of the two-sided)\\
\midrule
0.00 & 0.983 & 0.985 & 1.000 & 0.155\\
0.05 & 0.935 & 0.942 & 0.995 & 0.238\\
0.10 & 0.915 & 0.919 & 0.983 & 0.370\\
0.25 & 0.913 & 0.897 & 0.977 & 0.839\\
0.50 & 0.897 & 0.895 & 0.978 & 1.655\\
1.00 & 0.905 & 0.888 & 0.967 & 3.285\\
2.00 & 0.895 & 0.887 & 0.973 & 6.577\\
\bottomrule
\end{tabular}
\end{table}

%==============================================================================

%==============================================================================

%==============================================================================
% Extended discussion and related work (moved from the paper)
%==============================================================================
% Extended discussion and related work (moved from the main paper)
%==============================================================================
\section{Extended discussion and related work}
\label{app:extended-discussion}
%==============================================================================

\subsection{The low-bias regime and minimax pooling}

{The question of when to pool RCT and observational data has been analyzed in a minimax framework by \citet{minimaxRCTobs}. In their setting, a single RCT estimator $\widehat\psi_{\mathrm{rct}}$ is combined with a single observational estimator $\widehat\psi_{\mathrm{obs}}$ of a possibly biased mean, and the candidate pooled class consists of convex combinations indexed by $\lambda\in[0,1]$. Let $\Delta$ denote the observational bias, with an upper bound $\bar\Delta$ assumed known. The oracle procedure is the infeasible rule that, given $\Delta$, minimizes the mean squared error. It is dichotomous: when the bias is large relative to RCT noise, $\bar\Delta\gtrsim \sigma_{\mathrm{rct}}/\sqrt{n_{\mathrm{rct}}}$, the oracle sets $\lambda=0$ and uses the RCT alone. When the bias is small, it pools both sources as if $\Delta=0$. \citet{minimaxRCTobs} derive the minimax rate for mean squared error and confidence interval length and show that adaptivity to an unknown bias has a fundamental limit: without prior information placing the bias in the low regime, no uniformly valid confidence interval can improve on the RCT-only interval by more than a constant factor.}

{Our framework specializes to the low bias regime precisely when $\tau^2=0$: all identifying functionals---including the RCT and the IV based functionals indexed by subsets of environments---are asymptotically unbiased for the same $\psi^*$. There the EB combination $\widehat\psi_{\mathrm{EB}}$ collapses to a precision weighted average of $J$ valid estimators and inherits the minimax recommendation to pool. Beyond this regime, the centered bias setting of \ifmerged Theorem~\ref{thm:route1}\else Theorem~3 of the main paper\fi extends the pooling recommendation to cases in which individual functionals are biased but their biases average out across $j$, as long as $J\to\infty$. When biases are uncentered, consistency is not preserved and the minimax lower bound cautions against expecting both uniform validity and short intervals without additional prior information.}

\subsection{The phase transition as a descriptive diagnostic}

{The framework does not require a formal pre-inference test for $\tau^2=0$: the sandwich confidence interval of \ifmerged Section~\ref{sec:frequentist}\else Section~5 of the main paper\fi targets $\psi^*$ and the conformal prediction interval of \ifmerged Section~\ref{sec:bayesian}\else Section~4 of the main paper\fi targets $\psi_{J+1}$, and both outputs can be reported jointly since they quantify distinct forms of uncertainty. As an informal diagnostic, the magnitude of $\widehat\tau^2$ relative to the typical $v_j$ indicates which output is more informative: when $\widehat\tau^2$ sits at the boundary, the prediction interval collapses toward the confidence interval and the two coincide asymptotically; when $\widehat\tau^2$ dominates the $v_j$, the prediction interval widens to reflect genuine identification heterogeneity and is the more honest summary of what a new functional would deliver. This descriptive use of $\widehat\tau^2$ parallels its role in the historical-controls literature, where the between-study variance plays the same boundary-detecting role for the irreducible information limit of external evidence \citep{collignon2020clustered}.}

{The $\widehat\psi_{\mathrm{EB}}$ estimator in \ifmerged Eq.~\eqref{eq:eb-estimator}\else Eq.~(2) of the main paper\fi is a precision weighted pooling of randomized and observational evidence. Conceptually it fits within the EB $g$-modeling framework of \citet{Efron2014}: a parametric model is placed on the prior $g(\psi_j)$, the marginal laws of the estimators $f(\widehat\psi_j)=\int \phi(\widehat\psi_j;\psi_j,v_j)\,\mathrm{d}g(\psi_j)$ under the \bl{working independence device}, and the hyperparameters $(\psi^*,\tau^2)$ are estimated by maximizing the resulting marginal likelihood and plugged back into the posterior mean. A conceptual caveat is that pure $g$-modeling presumes conditional independence of the $\widehat\psi_j$ given $\psi_j$, whereas here all $\widehat\psi_j$ are computed from the same sample and their true joint law has non diagonal covariance. The object we maximize is therefore the marginal likelihood of a $g$-model under the \bl{working independence device}; the consistency argument of \ifmerged Section~\ref{sec:consistency}\else Section~3 of the main paper\fi and the sandwich variance of \ifmerged Section~\ref{sec:frequentist}\else Section~5 of the main paper\fi will absorb the resulting misspecification.}

{It is important to separate estimators from their identifying targets. The $\widehat\psi_j$ are not exchangeable: conditional on the realized design, their sampling variances $v_j$ are structurally different across $j$. For Wald/IV type constructions, for instance, $v_j$ is governed by first stage strength and sample allocation (the denominator contains terms of the form $nP(Z\in S)[\pi_l - \pi_m]^2$, with $m\neq l$ indexing environments), so once the environment specific propensity profile $(\pi_z)_z$ is fixed the $v_j$ can be ranked a priori. The exchangeability assumption invoked in the hierarchical layer is imposed on the latent identifying targets $\psi_j$, not on the estimators $\widehat\psi_j$, as a symmetry statement to the effect that, before observing the data, no functional is singled out as systematically more or less biased. Symmetry arguments of this kind underlie a broader EB paradigm in which a probabilistic symmetry of the latent law delivers, via an ergodic decomposition, the hierarchical structure used for simultaneous inference \citep{wu2024bayesian}. The exchangeability we impose on the identifying targets is the instance of that paradigm appropriate to an unindexed collection of identifying functionals, and the Gaussian mixing law \ifmerged in~\eqref{eq:working-model}\else in~(1) of the main paper\fi is the parametric specialization we adopt.}

{This symmetry admits two readings that the paper keeps logically separate. If the $\psi_j$ are treated as exchangeable random quantities reflecting epistemic uncertainty over the identifying functionals \citep{french2000statistical}, de Finetti's theorem guarantees the existence of a mixing measure $G$ such that, conditional on $G$, they are i.i.d.\ draws from $G$. Specializing to $G=N(\psi^*,\tau^2)$ is then a parametric choice, and the $\tau^2>0$ regime makes a prediction interval for the latent target $\psi_{J+1}$ of a new functional meaningful (\ifmerged Section~\ref{sec:bayesian}\else Section~4 of the main paper\fi). If instead the $\psi_j$ are regarded as deterministic functionals of $P$, the hierarchical prior is a modeling device; the inferential target is then the fixed $\psi^*$, and a frequentist confidence interval via sandwich estimation or subsampling is the appropriate tool (\ifmerged Section~\ref{sec:frequentist}\else Section~5 of the main paper\fi).}

{The non exchangeability of the estimators admits a simple intuition. Conditional on the realized environment specific characteristics driving each functional---for instance the propensities $(\pi_z)_z$ that control the precision of $\widehat\psi_j$---the $v_j$ are fixed and distinct, so $(\widehat\psi_j)_j$ is not exchangeable. Were one to imagine $(\pi_z)_z$ as drawn from a permutation invariant population prior, say a symmetric Dirichlet, the joint law of $(\widehat\psi_j,v_j)_j$ would become permutation invariant.}

\subsection{Related empirical Bayes constructions and extensions}

{The proposed framework connects naturally to the minimax data fusion literature. \citet{minimaxRCTobs} show that in the low bias regime pooling strictly dominates RCT only estimation, while \citet{LiGilbertDuanLuedtke2025} handle misalignment through finite dimensional selection bias models. Our proposal complements both by absorbing between functional discrepancy into a single heterogeneity hyperparameter at the estimator level, so that precision weighted pooling is recovered in the low bias regime while the centered bias regime extends the pooling recommendation beyond strict exchangeability of populations. The price for this generality is that inference in the centered bias regime relies on resampling rather than on closed form variance formulas, and uniform confidence interval length reductions beyond a constant factor are not available without prior knowledge of the bias.}

{It is worth situating the present construction explicitly relative to the closest existing entry in the literature on combining experimental and observational evidence through empirical Bayes. \citet{wu2026illusion} model an experimental estimator $y_e\sim N(\psi^*,\sigma_e^2)$ alongside observational estimators $y_{o,j}\sim N(\psi^*+b_j,\sigma_{o,j}^2)$ with random-effects bias $b_j\sim N(\mu,\gamma^2)$, and identify $(\mu,\gamma^2)$ from a separate channel of calibration studies in which the true effect is known to be zero. The structural shape of the inference is the same one we adopt here, but the device that identifies the heterogeneity hyperparameter is different: in the present framework, the dispersion of multiple identifying functionals targeting the same effect on the same sample plays the role that calibration studies play in \citet{wu2026illusion}, and the strong exchangeability assumption their construction requires (that the bias distribution is the same in calibration and in new observational studies) is replaced by the weaker exchangeability of latent deviations across identifying functionals. Two further differences are worth noting. First, in the heterogeneous regime $\tau^2>0$, our inferential output is a conformal prediction interval that is distribution-free in the latent law $G$ (\ifmerged Theorem~\ref{thm:conformal}\else Theorem~4 of the main paper\fi), whereas the Gaussian random-effects bias model of \citet{wu2026illusion} delivers coverage only under a specific tail shape of $G$. The empirical Bayes critique tradition \citep{polson2026oldlook} argues that the choice of hyperprior tail is consequential precisely in the sparse, heavy-tailed regime where shrinkage is most valuable; the conformal answer does not have to commit to it. Second, our point estimator $\widehat\psi_{\mathrm{EB}}$ is consistent under uncentered bias whenever the biases average out across functionals as $J\to\infty$ (\ifmerged Theorem~\ref{thm:route1}\else Theorem~3 of the main paper\fi), while \citet{wu2026illusion} require a known zero-mean of the bias distribution (or, equivalently, sufficient calibration studies to estimate it). The two constructions are complementary: in problems where calibration studies are available, their channel and our internal multiplicity can be combined; in problems where they are not, our framework remains operational.}

{Following the calibration device of \citet{wu2026illusion}, a natural extension of our framework is the introduction of \emph{calibration functionals}: identifying functionals applied to a contrast whose true causal effect is known \emph{a priori} to be zero, such as a placebo cohort, a pre treatment period, or a sub population subject to a known null intervention. Calibration functionals can identify the location $\mu=E[\varepsilon_j]$ of the bias distribution and remove the centering condition (C2) of \ifmerged Theorem~\ref{thm:route1}\else Theorem~3 of the main paper\fi. Concretely, the working model becomes $\psi_j=\psi^*+\mu+\varepsilon_j$ with $\varepsilon_j\sim N(0,\tau^2)$, the marginal likelihood of the calibration functionals is concave in $\mu$ given $\tau^2$, and the resulting plug in $\widehat\mu$ is subtracted from the substantive functionals before the precision weighted aggregation, restoring consistency under uncentered identification bias. The illusion return phenomenon of \citet{wu2026illusion}---in which jointly estimating $(\mu,\tau^2)$ from the substantive estimators alone collapses $\widehat\psi_{\mathrm{EB}}$ onto the minimum variance functional---explains why a separate calibration channel is essential rather than a stylistic preference. We leave a formal treatment of this calibration device within the present framework, including the resulting consistency and inferential guarantees under uncentered identification bias, as an open problem.}

\end{document}